\documentclass[11pt,reqno]{amsart}
\usepackage[english]{babel}
\usepackage{mathabx}
\usepackage{physics}
\usepackage{braket}
\usepackage{dsfont}
\usepackage{xcolor}
\usepackage{graphicx}
\usepackage{tikz}
\usepackage{amssymb}
\usepackage[T1]{fontenc}
\usepackage{fontawesome}
\usepackage{multirow}
\usepackage{colortbl}

\usetikzlibrary{matrix}

\usepackage[margin=1in]{geometry}
\usepackage[pagebackref, colorlinks = true, linkcolor = blue, urlcolor  = blue, citecolor = red]{hyperref}

\newcommand{\jewel}{\text{\faDiamond}}
\renewcommand{\diamond}{\diamondsuit}
\renewcommand{\epsilon}{\varepsilon}
\renewcommand{\phi}{\varphi}
\newcommand{\kk}{\mathbf k}
\newcommand{\cdeg}{\gamma}

\newcommand{\otimesmin}{\otimes_{\mathrm{min}}}
\newcommand{\otimesmax}{\otimes_{\mathrm{max}}}

\newtheorem{thm}{Theorem}[section]
\newtheorem*{thm*}{Theorem}
\newtheorem{lem}[thm]{Lemma}
\newtheorem{cor}[thm]{Corollary}
\newtheorem{ex}[thm]{Example}
\newtheorem{defi}[thm]{Definition}
\newtheorem{prop}[thm]{Proposition}
\newtheorem{remark}[thm]{Remark}
\newtheorem{question}[thm]{Question}

\begin{document}
\author{Andreas Bluhm}
\email{bluhm@math.ku.dk}
\address{QMATH, Department of Mathematical Sciences, University of Copenhagen, Universitetsparken 5, 2100 Copenhagen, Denmark}

\author{Anna Jen\v{c}ov\'a}
\email{jenca@mat.savba.sk}
\address{Mathematical Institute, Slovak Academy of Sciences, Bratislava, Slovakia}

\author{Ion Nechita}
\email{nechita@irsamc.ups-tlse.fr}
\address{Laboratoire de Physique Th\'eorique, Universit\'e de Toulouse, CNRS, UPS, France}

\title[Incompatibility in GPTs, generalized spectrahedra, and tensor norms]
{Incompatibility in general probabilistic theories, generalized spectrahedra, and tensor norms}

\date{\today}

\maketitle

\begin{abstract}
    In this work, we investigate measurement incompatibility in general probabilistic theories (GPTs). We show several equivalent characterizations of compatible measurements. The first is in terms of the positivity of associated maps. The second relates compatibility to the inclusion of certain generalized spectrahedra. For this, we extend the theory of free spectrahedra to ordered vector spaces. The third characterization connects the compatibility of dichotomic measurements to the ratio of tensor crossnorms of Banach spaces. We use these characterizations to study the amount of incompatibility present in different GPTs, i.e.~their compatibility regions. For centrally symmetric GPTs, we show that the compatibility degree is given as the ratio of the injective and the projective norm of the tensor product of associated Banach spaces. This allows us to completely characterize the compatibility regions of several GPTs, and to obtain optimal universal bounds on the compatibility degree in terms of the 1-summing constants of the associated Banach spaces. Moreover, we find new bounds on the maximal incompatibility present in more than three qubit measurements. 
\end{abstract}

\tableofcontents

\section{Introduction} \label{sec:introduction}
Many of the phenomena that distinguish quantum mechanics from a classical theory can be traced back to the incompatibility of measurements \cite{Heisenberg1927, Bohr1928}. Two measurements are incompatible if there does not exist a third one which implements them simultaneously. The position and momentum observables are typical examples of this behavior. This non-classicality present in a collection of measurements is necessary for many tasks in quantum information theory, because compatible measurements cannot exhibit non-locality in terms of violation of a Bell inequality \cite{Fine1982, Brunner2014} or steering \cite{uola2015one}. In this sense, incompatibility is a resource for quantum processing tasks similar to entanglement \cite{Heinosaari2015}.

It is thus natural to ask how much of this resource is available in a given situation. As noise can destroy incompatibility in the same way as it can destroy entanglement, \emph{noise robustness} is a natural way to quantify incompatibility \cite{Designolle2019}. While many works have investigated noise robustness for concrete measurements (see \cite{Heinosaari2016} for a review), one can ask this question more generally, leading to the \emph{compatibility regions} studied in \cite{busch2013comparing, Heinosaari2014, bluhm2018joint, bluhm2020compatibility}:

\medskip

\textit{How much incompatibility can be found for a quantum system of a certain dimension, a certain number of measurements and a certain number of measurement outcomes?}

\medskip

In fact, the incompatibility of measurements is not restricted to quantum mechanics, but is present in all non-classical theories \cite{Plavala2016}. These theories can conveniently be described in the framework of \emph{general probabilistic theories} (GPTs). We refer the reader to \cite{lami2018non} for a good introduction. An example of a GPT is a theory which has a state space isomorphic to a square (see e.g. \cite{busch2013comparing, Jencova2017, jencova2018incompatible} or Example \ref{ex:hypercube} for $n = 2$). Such a GPT can be used to model a theory containing PR boxes \cite{Popescu1994} which maximally violate the CHSH inequality. One of the motivation behind the study of GPTs is to characterize quantum mechanics among the probabilistic physical theories (see \cite{Barnum2006, Barnum2007,wolf2009measurements, Stevens2014} for some examples). Thus, it is interesting to study measurement incompatibility and incompatibility regions not only in quantum mechanics but in GPTs \cite{busch2013comparing, Gudder2013, Jencova2017, jencova2018incompatible}. 

The main theme of this work is to provide several equivalent characterizations of compatibility in GPTs, from different perspectives, and using different mathematical theories. We combine ideas from \cite{jencova2018incompatible} on incompatibility in GPTs with ideas from \cite{bluhm2018joint, bluhm2020compatibility} on (free) spectrahedra and we relate them to the theory of tensor norms in Banach spaces. We characterize compatibility of GPT measurements in five different ways (see Table \ref{tab:main-results}), unearthing powerful connections between the underlying mathematical theories. Our first main contribution is thus \emph{conceptual}, placing the measurement compatibility problem at the intersection of functional analysis, free convexity theory, and Banach space theory. We expect that this realization will help identifying other points of contact between these fields of mathematics and quantum information theory, allowing for progress on both fronts. Our second main contribution is \emph{applicative}: in the second part of the paper, we exploit these new connections and use results for tensor norms from Banach space theory to new bounds on the robustness of compatibility in GPTs, by relating ratios of different tensor norms to amount of noise needed to break all forms of incompatibility of dichotomic GPT measurements. From this perspective, our work pursues a similar approach to \cite{aubrun2018universal}, where similar ideas were used in the setting of XOR non-local games.

In Section \ref{sec:main-results}, we start with an overview of the main results of this work. Subsequently, we review results on convex cones, positive maps and Banach spaces in Section \ref{sec:prelim}. These are used in Section \ref{sec:extension} to study the extension of positive maps, which is important for the rest of this work. Section \ref{sec:gen-spectrahedra} extends the theory of free spectrahedra, which was developed in \cite{helton_matricial_2013, davidson2016dilations, helton2019dilations} among others, to ordered vector spaces. These first few sections therefore lay the technical groundwork for the rest of the paper. 

From Section \ref{sec:gpt} on, we focus on GPTs, whose basics we review in this section. In the next sections, we find several equivalent characterization of measurement compatibility. First, we make the connection between compatibility of measurements and  positivity of associated maps in Section \ref{sec:maps}. The findings in this section complement the connection to entanglement breaking maps in \cite{jencova2018incompatible}. Section \ref{sec:gen-spectra-and-comp} connects compatibility of measurements to the inclusion of certain generalized spectrahedra, thereby extending the results in \cite{bluhm2018joint,bluhm2020compatibility}. After this section, we focus primarily on dichotomic measurements. Section \ref{sec:tensor-norms} connects  compatibility of measurements to norms on their associated tensors. In particular, Section \ref{sec:witnesses} introduces a form of incompatibility witnesses and connects them to related notions in \cite{jencova2018incompatible, bluhm2020compatibility}. In Section \ref{sec:centrally-symmetric-tensor}, we focus on GPTs with symmetric state space, the so-called centrally symmetric GPTs. For this class, the connection to tensor norms has a particularly nice form. In particular, the noise robustness of incompatible measurements is connected to ratios of injective and projective tensor norms of certain Banach spaces. 

This connection to tensor norms allows us to find in Section \ref{sec:incl-const} concrete bounds on the compatibility regions of certain GPTs of interest, in several cases characterizing them completely; on the way, we expose a relation to $1$-summing norms.  We collect the results concerning compatibility regions in Section \ref{sec:results} and compare them to previously known bounds in the literature.

\section{Main results}\label{sec:main-results}

The aim of this work is to study measurement compatibility in general probabilistic theories (GPTs). The GPTs provide a
framework for the study of physical theories permitting probabilistic mixtures. Important examples of GPTs are classical
probability theory, quantum mechanics and the GPT of quantum channels.

Any GPT is built on basic operational notions
of \emph{states} (or preparation procedures) and \emph{effects} (or dichotomic measurements) of the
theory, which are identified with certain positive elements in a pair of dual ordered vector spaces,
$(V,V^+)$ and $(A,A^+)$, respectively. The space of effects, $(A,A^+)$, contains a distinguished order unit
$\mathds{1}$, corresponding to the trace in quantum mechanics. Elements of the set $K := \{v \in V \, : \, \mathds 1(v)
= 1 \}$ represent  {states} of the theory and elements $f \in A$ for which $0 \leq f \leq \mathds{1}$ correspond
to  {effects}. Composite systems are described by tensor products of the corresponding ordered vector spaces.
Note 
that the definition of such a tensor product is not unique and depends on the theory in question. Nevertheless, there is
 a minimal and maximal tensor product cone, and the cone describing composite systems in any theory must lie between
these two. In particular, states in  the minimal cone are called \emph{separable}, whereas all the other states are
\emph{entangled}.

 We will present the framework of GPTs in more details in Section \ref{sec:gpt}, together with
measurements and their compatibility. Throughout this work, we assume that the space $V$ (and of course also $A$) is finite
dimensional and we will study finite tuples of measurements with a finite set of outcomes.

Let $\mathbf k=(k_1,\dots,k_g)\in \mathbb{N}^g$ be a $g$-tuple determining the number of outcomes of each measurement. Consider $g$-tuple of measurements $f=(f^{(1)},\dots,f^{(g)})$ in $A$ with, respectively, $k_1, \ldots, k_g$ outcomes. This work is concerned with the question of the \emph{compatibility} of the given measurements, that is the existence of a joint measurement $h$ with outcome set $[k_1] \times \cdots [k_g]$ such that the $f^{(i)}$ are the marginals of $h$: 
$$\forall i \in [g], \, \forall j \in [k_i], \quad 
f^{(i)}_j=\sum_{m_1,\dots,m_{i-1},m_{i+1},\dots, m_g} h_{m_1,\dots,m_{i-1},j,m_{i+1},\dots,\dots,m_g}.$$
The main insight of this paper is a five-fold characterization of compatibility, from three different perspective. We summarize the different characterizations of compatibility in Table \ref{tab:main-results}, and we detail the points (a)-(e) below.

\begin{table}
    \centering
\bgroup
\def\arraystretch{1.5}
\centering
\begin{tabular}{|l|l|l|} 
\hline
\multicolumn{3}{|c|}{{\cellcolor[rgb]{0.6,0.6,0.6}}\textbf{Equivalent formulations of compatibility in GPTs}}                                                                                                                                                                                                                                                                                                            \\ 
\hline
\multicolumn{3}{|l|}{\textit{Input}: a $g$-tuple of measurements $(f^{(1)}, \ldots, f^{(g)}) \in A^{\mathbf k}$ $\iff$ a linear map $\Phi^{(f)}$ $\iff$ a tensor $\phi^{(f)}$}                                                                                                                                                                                                                                                  \\ 
\hline
\rowcolor[rgb]{0.8,0.8,0.8} \multicolumn{1}{|c|}{Formulation}                                                                 & \multicolumn{1}{c|}{Result}                                      & \multicolumn{1}{c|}{Keywords and Methods}                                                                                                                                                                                \\ 
\hline
\rowcolor[rgb]{0.792,0.969,0.89} (a) $\Phi^{(f)}$ has a pos. ext. $(\mathbb R^{\mathbf{k}},\mathbb R^{\mathbf{k}}_+)\to (A,A^+)$ & {\cellcolor[rgb]{0.792,0.969,0.89}}                              & {\cellcolor[rgb]{0.792,0.969,0.89}}                                                                                                                                                                          \\
\rowcolor[rgb]{0.792,0.969,0.89} (b) $\Phi^{(f)}$ is entanglement breaking                                                          & {\cellcolor[rgb]{0.792,0.969,0.89}}                              & {\cellcolor[rgb]{0.792,0.969,0.89}}                                                                                                                                                                          \\
\rowcolor[rgb]{0.792,0.969,0.89} (c) $\varphi^{(f)}\in \operatorname{Ran}(J_{\mathbf{k}}\otimes \mathrm{id})$     & \multirow{-3}{*}{{\cellcolor[rgb]{0.792,0.969,0.89}}Theorem \ref{thm:etb_ext}} & \multirow{-3}{*}{{\cellcolor[rgb]{0.792,0.969,0.89}}\begin{tabular}[c]{@{}>{\cellcolor[rgb]{0.792,0.969,0.89}}l@{}}functional analysis, map extension,\\entanglement breaking, polysimplex\end{tabular}}  \\ 
\hline
\rowcolor[rgb]{0.929,1,0.925} (d) $\mathcal D_{\mathrm{GPT}\jewel}(\mathbf k; V, V^+)\subseteq \mathcal D_f(\mathbf k; V, V^+)$     & Theorem \ref{thm:inclusion_compatible}                                                      & \begin{tabular}[c]{@{}>{}l@{}} free convexity, \\ generalized free spectrahedron \end{tabular}\\
\hline
\rowcolor[rgb]{0.965,0.875,0.922} (e) $\|\bar\varphi^{(f)}\|_c\le 1\quad\qquad$ \tiny{(dichotomic measurements)}                                                              & Theorem \ref{thm:effect-tensors}                                                      & Banach space, tensor norm                                                                                                                                                                                  \\
\hline
\end{tabular}
\egroup    
\vspace{.2cm}
\caption{An overview of the five-fold characterization of compatibility of GPT measurements.}
    \label{tab:main-results}
\end{table}

We first define 
  a certain subspace $E_{\mathbf{k}}$ in   $\mathbb{R}^{\mathbf{k}}:=\mathbb{R}^{k_1\cdots k_g}$ and endow it with the positive cone $E_{\mathbf
k}^+$ inherited from the simplicial cone $\mathbb{R}^{\mathbf k}_+$ in $\mathbb{R}^{\mathbf{k}}$. The dual space is
identified as  $E_{\mathbf k}^*\equiv E_{\mathbf k}$, with duality given by the standard inner product in
$\mathbb {R}^{\mathbf{k}}$. We observe that the dual cone is obtained as $J_{\mathbf k}(\mathbb{R}^{\mathbf k}_+)$, 
 where $J_{\mathbf k}:\mathbb R^{\mathbf k}\to E_{\mathbf k}$ is the orthogonal projection.  

To a given $g$-tuple of measurements $f=(f^{(1)},\dots,f^{(g)})$ with outcomes determined by $\mathbf k$,
we associate a linear map $\Phi^{(f)}: E_{\mathbf k}\to A$. By the standard identification of linear maps with 
elements in the tensor products, there is a related vector $\varphi^{(f)}\in E_{\mathbf{k}}\otimes
A$. It is then shown  that such a map  is positive if and only if it
corresponds to a $g$-tuple of measurements, equivalently, $\varphi^{(f)}\in (E_{\mathbf k}^+)^*\otimes_{\mathrm{max}} A^+$.
 
As our first main result, we  show  in Theorem \ref{thm:etb_ext} that compatibility of $f$ is characterized by either of the following equivalent conditions:
\begin{enumerate}
\item[(a)] $\Phi^{(f)}$ has a positive extension to a map $(\mathbb R^{\mathbf{k}},\mathbb R^{\mathbf{k}}_+)\to (A,A^+)$;
\item[(b)] $\Phi^{(f)}$ is \emph{entanglement breaking}, in very much the same sense as the entanglement breaking channels in 
quantum theory;
\item[(c)] $\varphi^{(f)}\in (J_{\mathbf{k}}\otimes \mathrm{id})(\mathbb R^{\mathbf k}_+\otimes A^+)$.

\end{enumerate}
These conditions are closely related to the results of \cite{jencova2018incompatible}.
Note that the ordered vector space $(E_{\mathbf{k}},E_{\mathbf k}^+)$ plays a universal role, since it only depends on the number of measurements and the numbers of their outcomes, but not on the GPT under study.

The next part of the paper is motivated by the effort to extend the results of \cite{bluhm2018joint} and
\cite{bluhm2020compatibility} to GPTs. For this we need some generalization of the free spectrahedra considered e.g.\ in
\cite{helton_matricial_2013, davidson2016dilations, helton2019dilations}. A \emph{generalized spectrahedron} as defined
in Section \ref{sec:gen-spectrahedra} is determined  by a tuple $(a_1, \ldots, a_g) \in M^g$, where $(M,M^+)$ is an ordered vector space. Given another ordered vector space $(L, L^+)$ and a tensor cone $C$ on the tensor product in $M \otimes L$, the generalized spectrahedron is
\begin{equation*}
    \mathcal D_{a}(L,C) := \left\{(v_1, \ldots, v_g) \in L^g \, : \, \sum_{i=1}^g a_i \otimes v_i  \in C\right\}.
\end{equation*}
For a usual free spectrahedron, $M^+$, $L^+$, and $C$ are the cones of positive semidefinite matrices of different
dimensions. Note that the generalized spectrahedron can be seen as a representation of the cone $C$ in the space 
 $L^g\cong \mathbb R^g\otimes V$ and if $a$ is a basis of $M$, it is isomorphic to $C$.
The generalized spectrahedron $\mathcal D_{\mathrm{GPT}\jewel}(\mathbf k; V, V^+)$, called the \emph{GPT jewel},
is given by $(M,M^+)=(E_{\mathbf k}, E_{\mathbf k}^+)$ and $a=w$ a basis of $E_{\mathbf k}$.  It is a universal object
 corresponding to  the matrix jewel defined in \cite{bluhm2020compatibility} for quantum mechanics.  
We also define the generalized spectrahedron $\mathcal D_f(\mathbf k; V, V^+)$ determined by the measurements under study. 

The main result of this part is the following equivalent condition for compatibility of $f$, see Theorem
\ref{thm:inclusion_compatible}:
\begin{enumerate}
\item[(d)] $\mathcal D_{\mathrm{GPT}\jewel}(\mathbf k; V, V^+)\subseteq \mathcal D_f(\mathbf k; V, V^+)$.
\end{enumerate}
This shows that the conditions (a) - (c) can be expressed by inclusion of generalized spectrahedra and gives a
connection between  the results of \cite{bluhm2018joint,bluhm2020compatibility} and \cite{jencova2018incompatible}.

Another characterization is obtained if we restrict to dichotomic measurements ($k_1=\cdots = k_g = 2$). We use the symmetry of the cone $(E_g^+)^*$, where $E_g:=E_{(2,2,\ldots, 2)}$, and the 
decomposition  $ E_g = \mathbb R 1_{2^g} \oplus \ell^g_\infty$ to relate compatibility of effects to \emph{reasonable crossnorms} on the tensor product $\ell^g_\infty
\otimes A$ (where we always endow $A$ with its order unit norm). 
Namely, for any $f=(f_1,\dots,f_g)\in A^g$ we put
\[
\bar \varphi^{(f)}:= \varphi^{(f)} - 2^{-g} 1_{2^g} \otimes \mathds{1}\cong \sum_i e_i \otimes (2f_i-\mathds{1}) 
\in \ell_\infty^g \otimes A.
\]
Then $g$-tuples of effects are determined by the condition
$\|\bar\varphi^{(f)}\|_\varepsilon\le 1$,
where $\|\cdot\|_\varepsilon$ is the injective crossnorm in $\ell^g_\infty\otimes A$. As another main result, we find a reasonable crossnorm 
 $\|\cdot\|_c$ in $\ell^g_\infty\otimes A$, such that compatibility of a $g$-tuple $f$ of dichotomic measurements is equivalent to 
\begin{enumerate}
\item[(e)] $\|\bar\varphi^{(f)}\|_c\le 1$,
\end{enumerate}
see Theorem \ref{thm:effect-tensors}.  These different viewpoints will be helpful in determining the amount of incompatibility
available in a GPT. We also show in the appendices that the conditions (a) and (e) can be checked by conic programs.

In case a tuple of measurements is not compatible, they can be made compatible if one adds enough noise to them. In this work, we will consider \emph{white noise}, i.e.\ the effects of the noisy measurements are of the form 
\begin{equation}\label{eq:white-noise1}
     \tilde f^{(i)}_j = s_i f_j^{(i)} + (1-s_i) \frac{\mathds{1}}{k_i}.
\end{equation}
Here, $k_i$ is the number of outcomes of the $i$-th measurement and $s_i \in [0,1]$ is a parameter quantifying the amount of noise added to the $i$-th measurement. In other words, the device corresponding to the noisy measurement $\tilde f^{(i)}$ carries out the measurement $f^{(i)}$ with probability $s_i$. With probability $(1-s_i)$, the device outputs a random number in $\{1, \ldots, k_i\}$ where each outcome has equal probability.
The minimal amount of noise such that any collection of $g$ measurements with $k_i$ outcomes in the $i$-th measurement is compatible, is a measure of the amount of incompatibility available in the GPT. Therefore, we are interested in the set
$\Gamma(f)\subseteq [0,1]^g$  of $s_1,\dots,s_g$ such that the noisy measurements $\tilde f^{(i)}$ in Equation
\eqref{eq:white-noise1} are compatible. The set of $s\in [0,1]^g$ such that this is true for 
 any collection of $g$ measurements with outcome sets given by $\mathbf{k}$ will be called the
 \emph{compatibility region} for the GPT and denoted by $\Gamma(\mathbf k; V, V^+):=\bigcap_f \Gamma(f)$. 

The largest $s\in[0,1]$ such that $(s, \ldots, s) \in \Gamma(f)$ is the compatibility degree of $f$, denoted by
$\gamma(f)$. This  corresponds to the situation when the same amount of noise is added to each measurement in $f$.
The value such that this is true for all collections of measurements with outcomes given by $\mathbf k$
is the  \emph{compatibility degree} of the GPT:
\[
\gamma(\mathbf k; V, V^+):=\min_f \gamma(f)=\max\{s\in [0,1],\ (s,\dots,s)\in
\Gamma(\mathbf{k}, V,V^+)\}.
\]
Consider the set of vectors $s \in [0,1]^g$ that can be used to scale the GPT jewel $\mathcal D_{\mathrm{GPT}\jewel}$
such that it is contained in $\mathcal D_f(\mathbf k; V, V^+)$ for any collection of $g$ measurements $f^{(i)}$ with
$k_i$ outcomes. This set is the set of \emph{inclusion constants} of $\mathcal D_{\mathrm{GPT}\jewel}$, denoted by $\Delta(\mathbf k; V, V^+)$. Following along the lines of \cite{bluhm2018joint,bluhm2020compatibility}, we prove that the compatibility region of a GPT is equal to the set of inclusion constants of the corresponding jewel, see Theorem \ref{thm:delta-is-gamma}:
\[
\Gamma(\mathbf k; V,V^+) = \Delta(\mathbf k; V,V^+).
\]

We can get more results for dichotomic measurements.  In this case we write $\Gamma(g;V,V^+)$ for the compatibility
region, similarly $\gamma(g; V,V^+)$ for the compatibility degree. The GPT jewel in this case is called the 
 \emph{GPT diamond} and denoted by $\mathcal D_{\mathrm{GPT}\diamond}(g;V,V^+)$. 
Using the compatibility characterization (e) we immediately obtain a direct relation of the compatibility measures
 to the norm $\|\cdot\|_c$, in particular,
\[
\gamma(f)=1/\|\bar\varphi^{(f)}\|_c.
\]
From this relation, we find in Theorem \ref{thm:inclusion_Gamma} that we have 
\begin{equation}\label{eq:main_compreg}
\Gamma(g;V,V^+)=\{s\in [0,1]^g,\ \|s.\varphi\|_c\le 1,\ \forall \varphi\in \ell^g_\infty\otimes A,\
\|\varphi\|_\epsilon \le 1\}.
\end{equation}
In particular, the compatibility degree satisfies
\begin{equation}\label{eq:main_comdeg}
\gamma(g; V,V^+)= 1/{\max_{\|\varphi\|_{\epsilon}\le 1} \|\varphi\|_{c}}\ge 1/{\rho(\ell^g_\infty,A)}\ge
1/{\min\{g,\dim(V)\}},
\end{equation}
where the quantity $\rho(X,Y)$ for a pair of Banach spaces $X$, $Y$ was introduced in \cite[Eq.~(15)]{aubrun2018universal}:
\[
\rho(X,Y)= \max_{z\in X\otimes Y}\frac{\|z\|_{X\otimes_\pi Y}}{\|z\|_{X\otimes_\epsilon Y}}.
\]

A characterization which is dual to the one above can be obtained using \emph{incompatibility witnesses}. In this work,
an incompatibility witness is a tuple $z =(z_1, \ldots, z_g)$ such that there is a state $z_0 \in K$ such that
$(z_0,z_1, \ldots, z_g) \in \mathcal D_{\mathrm{GPT}\diamond}(g;V,V^+)$. The set of incompatibility witnesses
 is denoted $\mathcal P_{\mathrm{GPT}\diamond}(g;V, V^+)$ and we show that it can be identified with the unit ball of
the dual norm to $\|\cdot\|_c$. For the relation of this definition of incompatibility witnesses to the definitions
in \cite{jencova2018incompatible} and \cite{bluhm2020compatibility}, see Propositions \ref{prop:incomp_witness} and
\ref{prop:incomp_witness_2}. An incompatibility witness $z$ certifies incompatibility for some collection of effects  if
 and only if we have  $\sum_i \|z_i\|_V$ strictly larger than one. 
Theorem \ref{thm:pi=gamma=delta} gives another characterization of the compatibility region as the  set of vectors $s
\in [0,1]$ which shrink all incompatibility witnesses such that they cannot detect incompatibility anymore: 
\begin{equation*}
\Gamma(g; V, V^+)=\left\{(s_1, \ldots, s_g) \in [0,1]^g: \sum_{i = 1}^g s_i \norm{z_i}_V \leq 1 ~\forall (z_1, \ldots, z_g) \in \mathcal P_{\mathrm{GPT} \diamond}(g;V, V^+) \right\}.
\end{equation*}

Finally, we consider a special class of GPTs for which our results have a simpler form: the \emph{centrally symmetric}
GPTs. The state spaces of these GPTs are the unit balls of some norm on a vector space $\bar V$ and we have $V\cong
\mathbb R\oplus \bar V$.  Important examples are
the Bloch ball describing 2-level systems in quantum mechanics, or the hypercubic GPT. To get a better understanding of
our results, let us consider compatibility for a pair of qubit effects. Qubits are described by a centrally symmetric GPT, 
 where the corresponding normed space is $\ell^3_2$. Each qubit  effect has the form 
\[
f=\frac12(\alpha I+\mathbf a\cdot\sigma),
\]
where $\sigma=(\sigma_1,\sigma_2,\sigma_3)$ are Pauli matrices and $\mathbf a\in \mathbb R^3$ is a vector such that 
$\|\mathbf a\|_2\le \alpha\le 2-\|\mathbf a\|_2$. The effects are called \emph{unbiased} if $\alpha =1$. For a pair of
unbiased effects $f_1=\frac12(I+\mathbf a\cdot \sigma)$ and $f_2=\frac12(I+\mathbf b\cdot \sigma)$, 
we have
\[
\|\bar \varphi^{(f)}\|_c=\left\| e_1\otimes \mathbf a+ e_2\otimes \mathbf b\right\|_{\ell^2_\infty\otimes_\pi \ell^3_2}
=\frac12 (\|\mathbf a+\mathbf b\|_2+\|\mathbf a-\mathbf b\|_2),
\]
(see Section \ref{sec:euclidean})
so that the condition (e) becomes the well known compatibility condition for two unbiased qubit effects obtained in 
\cite{busch1986unsharp,busch2008approximate}. It follows that we  extended this condition  to $g$-tuples of 
arbitrary effect in any GPT.

In Theorem \ref{thm:pi=piprime}, we show that for this class of GPTs,
we can replace in Equation \eqref{eq:main_compreg} the space $A$ by $\bar A$ and the norm $\|\cdot\|_c$ by the projective norm in $\ell^g_\infty\otimes_\pi \bar A$, where $\bar A$ is the dual Banach
space to $\bar V$: 
\begin{align*}
\Gamma(g;V,V^+) &= \{s\in [0,1]^g: \|s.\bar z\|_{\ell_1^g\otimes_\pi \bar V} \leq 1, \ \forall 
 \|\bar z\|_{\ell_1^g\otimes_\epsilon\bar V}\le 1 \}\\
&=\{s\in [0,1]^g: \|s.\bar \phi\|_{\ell^g_\infty\otimes_\pi\bar A}\le 1,\ \forall 
\|\bar\phi\|_{\ell_\infty^g\otimes_\epsilon\bar A}\le 1\}.
\end{align*}
In particular, the compatibility degree is
\[
\gamma(g;V,V^+)=1/\rho(\ell_\infty^g, \bar A)\ge 1/\min\{g,\dim(\bar A)\}.
\]
In this case, the lower bound is attained by the hypercubic GPT. Note that this bound is larger than the lower bound in Equation \eqref{eq:main_comdeg}. The tightness of the general lower bound $\max\{1/g,1/\dim(V)\}$ remains an open question.

We also put forward a connection to 1-summing norms:
$$\lim_{g \to \infty} \gamma(g;V,V^+) = \frac{1}{\pi_1(\bar V)},$$
where $\pi_1$ is the 1-summing norm of the Banach space $(\bar V, \|\cdot \|_{\bar V})$. This relation allows us to prove new lower bounds for the compatibility degree of qubits, which are described by the Bloch ball, corresponding to $\bar V = \ell_2^3$.

\section{Preliminaries} \label{sec:prelim}
  
 \subsection{Notation and basic definitions}
 In this paper, we will assume all vector spaces to be finite dimensional and over the real field. For brevity, we will often write $[n]:=\{1, \ldots, n\}$ for $n \in \mathbb N$. Let $\mathcal M_n(\mathbb C)$ denote the $n \times n$ matrices with complex entries. The real vector space of Hermitian matrices will be written as $\mathcal M^{\mathrm{sa}}_n(\mathbb C)$. For a vector space $L$ and an element $a \in L^g \cong L\times \ldots \times L$, $g \in \mathbb N$, we will write $a_i$, $i\in [g]$, for its components without specifying this in advance, if no confusion can arise. For a convex set $K \subseteq \mathbb R^g$, let us write
 \begin{equation*}
     K^\circ := \{x \in  \mathbb R^g: \langle x, k \rangle \leq 1~\forall k \in K\}.
 \end{equation*}
 for the polar set of $K$. Moreover, let us write for the direct sum of two convex sets $K_i \in \mathbb R^{g_i}$, $g_i \in \mathbb N$, $i \in [2]$, 
 \begin{equation*}
     K_1 \oplus K_2 := \operatorname{conv}\{(x,0); (0,y): x \in K_1,~y \in K_2\}\subseteq \mathbb{R}^{g_1+g_2}.
 \end{equation*}
 If $K_1$, $K_2$ are polytopes containing $0$, it holds that $K_1 \oplus K_2 = (K_1^\circ \times K_2^\circ)^\circ$ \cite[Lemma 2.4]{Bremner1997}. We will denote the probability simplex in $\mathbb R^k$ (i.e.~the probability distributions on $k$ symbols) by $\Delta_k$.

\subsection{Convex cones}
Let $L$ be a finite-dimensional real vector space. A subset $L^+ \subseteq L$ is a \emph{convex cone} if $\lambda x + \mu y \in L^+$ for all $x$, $y \in L^+$ and all $\lambda$, $\mu \in \mathbb R_+$. Often, we will drop ``convex'' and talk simply about ``cones''. To avoid pathologies, we will assume all cones to be non-empty. A cone will be called \emph{generating} if $L = L^+ - L^+$. Moreover, it is \emph{pointed} if $L^+ \cap (-L^+) = \{0\}$ (sometimes this is called salient instead). A \emph{proper} cone is a convex cone which is closed, pointed and generating. A \emph{base} of a cone $C$ is a convex set $K \subset C$ such that for every $x \in C$, there is a unique $\lambda \geq 0$ such that $x \in \lambda K$. The \emph{dual cone} of $L^+$ is the closed convex cone
\begin{equation*}
(L^+)^\ast := \{f \in L^\ast: f(x) \geq 0 ~\forall x \in L \},
\end{equation*}
where $L^\ast$ is the dual vector space of $L$. Two cones $L^+_1$ and $L^+_2$ living in vector spaces $L_1$ and $L_2$,
respectively, are \emph{isomorphic} if there is a linear bijection $\Theta: L_1 \to L_2$ such that $\Theta(L^+_1) =
L^+_2$. A cone $L^+$ is called \emph{simplicial} if it is isomorphic to $\mathbb R^{d}_+$, where $d = \dim(V^+)$. 

A \emph{preordered vector space} is a tuple $(L, L^+)$, where $L$ is a vector space and $L^+$ is a convex cone. If $L^+$ is additionally pointed, it is an \emph{ordered vector space}. If $V^+$ is a proper cone, we will call $(V, V^+)$ a \emph{proper ordered vector space}. We write $y \geq x$ for $x$, $y \in V$ to mean that $y - x \in L^+$. If $(V, V^+)$ is a preordered vector space, then $(V^*, (V^+)^\ast)$ is its dual preordered vector space. If $(V, V^+)$ is a proper ordered vector space, so is $(V^*, (V^+)^\ast)$.

An \emph{order unit} $\mathds 1 \in V^+$ is an element such that for every $v \in L$, there is a $\lambda > 0$ such that $v \in \lambda[-\mathds 1, \mathds 1]$. By \cite[Lemma 1.7]{Aliprantis2007}, $\mathds 1 \in V^+$ is an order unit if and only if $\mathds 1 \in \operatorname{int} V^+$.
\begin{lem}\label{lem:trivial-cone}
Let $(L, L^+)$ be a preordered vector space and let $\mathds 1$ be an order unit. Then $L^+ = L$ if and only if $- \mathds 1 \in L^+$.
\end{lem}
\begin{proof}
One implication is clear. Thus, let $- \mathds 1 \in L^+$. Let $v \in L$. Since $\mathds 1$ is an order unit, there is a $\lambda > 0$ such that $v + \lambda \mathds 1 \in L^+$. Since $- \lambda \mathds 1 \in L^+$, it follows that $v \in L^+$. This proves the assertion since $v \in L$ was arbitrary.
\end{proof}

The following facts about cones and their duals will be useful:
\begin{thm}[Bipolar theorem, {\cite[Theorem 14.1]{Rockafellar1970}}]
Every non-empty closed convex cone $C \subseteq L$ satisfies $C^{\ast \ast} \cong C$, where we have identified $L \cong L^{\ast \ast}$.
\end{thm}

\begin{lem}[\cite{Mulansky1997}]
Let $C \subseteq L$ be a convex cone. The cone $C$ is generating if and only if $\operatorname{int}{C} \neq \emptyset$. If $\operatorname{cl}C$ is pointed, then $C^\ast$ is generating. If $C$ is generating, then $C^\ast$ is pointed. In particular, if $C$ is proper, then $C^\ast$ is proper.
\end{lem}

\subsection{Tensor products of cones}
There are in general infinitely many natural ways to define the tensor product of two cones $L^+_1 \subseteq L_1$ and $L^+_2 \subseteq L_2$. Among these, there is a minimal and a maximal choice: The \emph{minimal tensor product} of $L^+_1$ and $L^+_2$ is the cone
\begin{equation*}
L^+_1 \otimes_{\min} L^+_2 := \mathrm{conv}\{x \otimes y: x \in L^+_1, y \in L^+_2 \},
\end{equation*}
whereas the \emph{maximal tensor product} of $L^+_1$ and $L^+_2$ is defined as
\begin{equation*}
L^+_1 \otimes_{\max} L^+_2 := ((L^+_1)^\ast \otimes_{\min} (L^+_2)^\ast)^\ast.
\end{equation*}
It can be seen that if $L^+_1$ and $L^+_2$ are proper, $L^+_1 \otimes_{\min} L^+_2$ and $L^+_1 \otimes_{\max} L^+_2$ are proper as well \cite[Fact S23]{Aubrun2019}. Moreover, $L_1 \otimes_{\min} L_2$ is closed if $L_1$ and $L_2$ are \cite[Exercise 4.14]{aubrun_alice_2017}. We call $C$ a \emph{tensor cone}  for $L_1^+$ and $L_2^+$ if
\begin{equation*}
L_1^+ \otimes_{\min} L_2^+ \subseteq C \subseteq L_1^+ \otimes_{\max} L_2^+.
\end{equation*} 
From the recent work \cite{Aubrun2019a}, we know that the tensor product of two cones is unique if and only if one of the cones is simplicial. This solves a longstanding open problem from \cite{Namioka1969,barker1981theory}:
\begin{thm}[{\cite[Theorem A]{Aubrun2019a}}]
Let $L^+_1$ and $L^+_2$ be proper cones. Then, $L^+_1 \otimes_{\min} L^+_2 = L^+_1 \otimes_{\max} L^+_2$ if and only if $L^+_1$ or $L^+_2$ is simplicial.
\end{thm}

\subsection{Positive maps}\label{sec:positive_maps}
Let $(L,L^+)$ be a preordered vector space and let $(L^\ast,(L^+)^\ast)$ be its dual.  
Consider the identity map $\mathrm{id}_L : L \to L$ and the associated canonical evaluation tensor
$\chi_L \in L\otimes L^\ast$, defined by the following remarkable property:
\begin{equation} \label{eq:property-of-the-max-ent-state}
\forall v \in L, \, \forall \alpha \in L^\ast, \qquad \langle \chi_L, \alpha\otimes v \rangle = \alpha(v).
\end{equation}
Using coordinates, we have
\begin{equation}\label{eq:def-chi}
\chi_L = \sum_{i=1}^{\dim L} v_i \otimes \alpha_i \in L \otimes L^\ast \end{equation}
for $\{v_i\}_{i = 1}^{\dim L}$ a basis of $L$ and $\{\alpha_i\}_{i = 1}^{\dim L}$ the corresponding dual basis in $L^\ast$ (we have $\alpha_i(v_j) = \delta_{ij}$). Let $(M,M^+)$ be another preordered vector space and let $\Phi:M\to L^\ast$ be a linear map. We define the linear functional 
$s_\Phi: M\otimes L\to \mathbb R$ as
\begin{equation}\label{eq:def_sPhi}
s_\Phi(z)=\langle \chi_L, (\Phi\otimes \mathrm{id})(z)\rangle, \qquad z\in M\otimes L.
\end{equation}
We then have
\begin{equation}\label{eq:Phi_sPhi}
\Phi(w)=\sum_{i=1}^{\dim L} s_\Phi(w\otimes v_i)\alpha_i,\qquad \forall w\in M.
\end{equation}
Since $(M\otimes L)^*\cong M^*\otimes L^\ast$, $s_\Phi$ corresponds to an element $\varphi^\Phi\in M^*\otimes L^\ast$, determined
as
\begin{equation}\label{eq:def_phiPhi}
\langle \varphi^\Phi, w\otimes v\rangle= s_\Phi(w\otimes v)=\langle \Phi(w), v \rangle, \qquad w\in M,\ v\in L.
\end{equation}
One can check that 
\begin{equation}\label{eq:phiPhi-chi}
\varphi^\Phi = (\mathrm{id} \otimes \Phi)(\chi_{M^*}) = (\Phi^\ast \otimes \mathrm{id} )(\chi_{L}).
\end{equation}
Here, $\Phi^\ast:L \to M^*$ is the dual linear map. Note that we have $\chi_L=\varphi^{\mathrm{id}_{L^\ast}}$. In this paper, we shall often switch between the three different equivalent points of view:
\begin{center}
    \begin{tabular}{ccccc}
        linear map & $\qquad \longleftrightarrow \qquad$ &   linear form & $\qquad \longleftrightarrow \qquad$ & tensor \\
        $\Phi:M\to L^\ast$ && $s_\Phi: M\otimes L\to \mathbb R$ && $\varphi^\Phi\in M^*\otimes L^\ast$
    \end{tabular}
\end{center}

We  say that $\Phi$ is a \emph{positive map} $(M,M^+)\to (L^\ast,(L^+)^\ast)$ if $\Phi(M^+)\subseteq (L^+)^\ast$. It is quite clear that 
 this happens if and only if $s_\Phi:(M\otimes L, M^+\otimesmin L^+)\to \mathbb R$ is positive, equivalently, 
$\varphi^\Phi\in (M^+)^*\otimesmax (L^+)^\ast$; this gives the well-known correspondence between positive maps and the maximal
tensor product. In particular, we have
\begin{lem}\label{lem:chi-in-max}
	The identity map $\mathrm{id}_L : (L,L^+) \to (L, L^+)$ is positive. Moreover, $\chi_L \in L^+\otimesmax
(L^+)^\ast$.
\end{lem}

\begin{proof}
The first assertion is obvious. The second one follows from Equation~\eqref{eq:property-of-the-max-ent-state}.
\end{proof}

For a positive linear map $\Psi:(M,M^+)\to (L,L^+)$, the dual map $\Psi^{\ast}:(L^\ast,(L^+)^\ast)\to (M^\ast,(M^+)^\ast)$ is positive as well. Moreover, it holds that for preordered vector spaces $(L_i, L_i^+)$, $(M_i, M_i^+)$ and positive maps $\Phi_i: (L_i, L_i^+) \to (M_i, M_i^+)$, $i \in [2]$, both 
\begin{equation*}
    \Phi_1 \otimes \Phi_2: (L_1 \otimes L_2, L_1^+ \otimes_{\min} L_2^+) \to (M_1 \otimes M_2, M_1^+ \otimes_{\min} M_2^+) 
\end{equation*}
and 
\begin{equation*}
    \Phi_1 \otimes \Phi_2: (L_1 \otimes L_2, L_1^+ \otimes_{\max} L_2^+) \to (M_1 \otimes M_2, M_1^+ \otimes_{\max} M_2^+) 
\end{equation*}
are positive \cite[Section II.C]{jencova2018incompatible}. Both statements also follow from Lemma \ref{lem:ABCD-inclusion-cones} below.

The minimal tensor product is associated with a special kind of maps. We first prove the following lemma, clarifying
 what happens under reordering of tensor products of cones.

\begin{lem}\label{lem:ABCD-inclusion-cones}
	Let $A^+,B^+,C^+,D^+$ be cones. Then
	$$(A^+ \otimesmin B^+) \otimesmin (C^+ \otimesmax D^+) \subseteq (A^+ \otimesmin C^+) \otimesmax (B^+ \otimesmin
D^+).$$
\end{lem}
\begin{proof}
	Consider arbitrary $a \in A^+$, $b \in B^+$ and $e \in C^+ \otimesmax D^+$.
 We have to check that for any $\phi \in (A^*)^+ \otimesmax (C^*)^+$ and
$\psi \in (B^*)^+ \otimesmax (D^*)^+$, we have 
	$$\braket{a \otimes b \otimes e, \phi \otimes \psi} \geq 0.$$
	 We shall use (twice) the following fact: given cones $X^+,Y^+$ in vector spaces $X$, $Y$, $z \in X^+ \otimesmax
Y^+$ and $\sigma \in (X^+)^\ast$, the vector $y = \braket{\sigma, z} \in Y$ defined by
	$$\braket{\tau, y} = \braket{\sigma \otimes \tau,z}, \qquad \forall \tau \in Y^*$$
	is positive (i.e.~$y \in \overline{Y^+}\cong (Y^+)^{\ast\ast}$). This corresponds to the fact that the
evaluation of a positive map at a positive element is positive, and we leave its proof as an exercise for the reader. 
	
	Using the fact above, and writing $\gamma:= \braket{a,\phi} \in (C^+)^\ast$ and $\delta:=\braket{b,\psi} \in
(D^+)^\ast$, we have 
	$$\braket{a \otimes b \otimes e, \phi \otimes \psi} = \braket{e, \gamma \otimes \delta} \geq0,$$
	proving the statement of the lemma. 
\end{proof}

\begin{prop}\label{prop:EB}
	 Let $\Phi : (M,M^+) \to (L^\ast, (L^+)^\ast)$ be a positive map between preordered vector spaces. The following conditions are equivalent:
	\begin{enumerate}
		\item The map $\Phi \otimes \mathrm{id}_{L} : (M \otimes L, M^+ \otimesmax \overline{L^+}) 
\to (L^\ast \otimes L ,(L^+)^\ast \otimesmin \overline{L^+})$ is positive.
		\item The map $s_\Phi : (M \otimes L,M^+ \otimesmax L^+) \to \mathbb R$ is positive.
		\item $\phi^\Phi \in (M^+)^* \otimesmin (L^+)^\ast$ 
		\item The map $\Phi \otimes \mathrm{id}_N : (M \otimes N,M^+ \otimesmax N^+) \to (L^\ast \otimes N,(L^+)^\ast \otimesmin N^+)$ is positive, for any preordered vector space $(N,N^+)$ with $N^+$ closed.
	\end{enumerate}
\end{prop}

\begin{proof}
	Since both $(2) \iff (3)$ and $(4) \implies (1)$ are trivial, we only prove $(1) \implies (2)$ and $(3) \implies
(4)$. The first implication follows immediately from  Lemma \ref{lem:chi-in-max}.  For the second
implication, use Lemmas \ref{lem:chi-in-max} and \ref{lem:ABCD-inclusion-cones} to prove that 
	$$\phi^\Phi \otimes \chi_N \in ((M^*)^+ \otimesmin (L^+)^\ast) \otimesmin ((N^*)^+ \otimesmax N^+) \subseteq (M^+ 
\otimesmax N^+)^* \otimesmax ((L^+)^\ast \otimesmin N^+),$$
	which is precisely the desired conclusion, since $\phi^\Phi \otimes \chi_N$ is the tensor corresponding to the map
$\Phi \otimes \mathrm{id}_N$.
\end{proof}

\begin{defi} Let $\Phi:(M,M^+)\to (L^\ast,(L^+)^\ast)$ be a positive linear map between preordered vector spaces. We say that $\Phi$ is \emph{entanglement breaking} (EB) if any of the equivalent conditions in Proposition \ref{prop:EB} holds.
\end{defi}

\begin{remark}
Point (1) of Proposition~\ref{prop:EB} was used as a definition in \cite{jencova2018incompatible}. Note that this definition of entanglement breaking maps agrees with the usual one used in quantum mechanics \cite{Horodecki2003} for $L^\ast$, $M$ being the vector space of Hermitian matrices and their respective cones being the cones of positive semidefinite elements. This can be seen from the point (3) of Proposition~\ref{prop:EB}, which states that the corresponding Choi matrix is separable \cite[Exercise 6.1]{watrous2018theory}.
\end{remark}

We conclude this section with a small lemma connecting the trace of the composition of two maps and the inner product of the corresponding tensors. Recall that the \emph{trace} of a linear operator $\Phi : L \to L$ is defined as
$$\operatorname{Tr} \Phi = \sum_{j=1}^{\dim L} \alpha_j(\Phi(v_j)) = \langle \chi_L, \phi^\Phi \rangle,$$
where $\{v_j\}_{j = 1}^{\dim L}$, $\{\alpha_j\}_{j = 1}^{\dim L}$, are dual bases of $L$, $L^*$.

\begin{lem} \label{lem:trace-of-maps}
Let $\Psi: M \to L$, $\Phi: M^\ast\to L^\ast$ be two linear maps between vector spaces. Then,
\begin{equation*}
    \operatorname{Tr}[\Psi \Phi^\ast] = \langle \varphi^{\Phi}, \varphi^{\Psi} \rangle.
\end{equation*}
\end{lem}
\begin{proof}
    Using the definition of the trace and Equation \eqref{eq:phiPhi-chi}, we have
    \begin{align*}
        \operatorname{Tr}[\Psi \Phi^\ast] &= \langle \chi_L, \phi^{\Psi\Phi^*} \rangle\\
        &= \langle \chi_L, (\Phi \Psi^* \otimes \mathrm{id}_L)(\chi_{L^*})\rangle \\
        &=   \langle (\Phi^\ast \otimes \mathrm{id}_{L^*}) (\chi_{L}), (\Psi^\ast \otimes \mathrm{id}_{L})(\chi_{L^*}) \rangle\\
        &= \langle \varphi^{\Phi}, \varphi^{\Psi} \rangle.
    \end{align*}
\end{proof}

\subsection{Tensor products of Banach spaces} \label{sec:tensor-norms-prelim}
Besides the tensor product of convex cones, we will also need the tensor product of Banach spaces. See \cite{Ryan2002} for a good introduction. Let $X$, $Y$ be two Banach spaces with norms $\norm{\cdot}_X$ and $\norm{\cdot}_Y$, respectively. Then, there are usually infinitely many natural norms that can be used to turn the vector space $X \otimes Y$ into a Banach space. There are two choices of norms which are minimal and maximal in a sense as we shall see next.

\begin{defi}[Projective tensor norm]
The \emph{projective norm} of an element $z \in X \otimes Y$ is defined as
\begin{equation*}
    \norm{z}_{X \otimes_\pi Y}:= \inf \left\{\sum_i \norm{x_i}_X \norm{y_i}_Y: z = \sum_{i} x_i \otimes y_i \right\}.
\end{equation*}
\end{defi}
Let $X^\ast$ and $Y^\ast$ be the dual spaces of $X$ and $Y$, respectively, and let their norms be $\norm{\cdot}_{X^\ast}$ and $\norm{\cdot}_{Y^\ast}$.
\begin{defi}[Injective tensor norm]
Let $z = \sum_i x_i \otimes y_i \in X \otimes Y$. Then, its \emph{injective norm} is 
\begin{equation*}
      \norm{z}_{X \otimes_\epsilon Y}:= \sup \left\{\left| \sum_i \phi(x_i) \psi(y_i)\right|: \norm{\phi}_{X^\ast} \leq 1, \norm{\psi}_{Y^\ast} \leq 1   \right\}.
\end{equation*}
\end{defi}
If $X$ and $Y$ are clear from the context, we will sometimes only write $\norm{\cdot}_\pi$ and $\norm{\cdot}_\epsilon$, respectively. Importantly, $\norm{\cdot}_\epsilon$ and $\norm{\cdot}_\pi$ are dual norms, i.e.
\begin{equation*}
    \norm{z}_{X \otimes_\epsilon Y} = \sup_{\norm{\phi}_{X^\ast \otimes_\pi Y^\ast} \leq 1} |\phi(z)|
\end{equation*}
and vice versa.

In some cases, the projective and injective norms have simpler expressions. The projective norm $\ell^{g}_1 \otimes_\pi X$ of a vector 
$$\mathbb R^g \otimes X \ni z = \sum_{i=1}^g e_i \otimes z_i$$
with $\{e_i \}_{i \in [g]}$ the standard basis is given by (see e.g.\ \cite[Example 2.6]{Ryan2002})
\begin{equation} \label{eq:projective-norm}
    \|z\|_{\ell_1^g \otimes_\pi X} =  \sum_{i=1}^g \|z_i\|_{X}.
\end{equation}
The injective norm $\ell^g_1 \otimes_\epsilon X$ is (see e.g.\ \cite[Example 3.4]{Ryan2002})
\begin{equation} \label{eq:injective-norm}
\|z\|_{\ell_1^g \otimes_\epsilon X} = \sup_{\norm{y}_{X^\ast} \leq 1} \sum_{i = 1}^g |\langle y, z_i \rangle| = \sup_{\epsilon \in \{\pm 1\}^g} \norm{\sum_{i = 1}^g \epsilon_i z_i}_X.
\end{equation}

We will later be interested in the maximal ratio between the projective and the injective norms introduced in \cite[Equation (15)]{aubrun2018universal}:
\begin{equation} \label{eq:rho-quotient}
    \rho(X, Y) = \max_{\|z\|_{\epsilon}\le 1} \|z\|_{\pi}
\end{equation}
where the maximum runs over all $z \in X \otimes Y$.

The injective and the projective norms both belong to a class of tensor norms which are the reasonable crossnorms.

\begin{defi}[\cite{Ryan2002}]
Let $X$ and $Y$ be two Banach spaces. We say that a norm $\norm{\cdot}_\alpha$ on $X \otimes Y$ is a \emph{reasonable crossnorm} if it has the following properties:
\begin{enumerate}
\item $\norm{x \otimes y}_\alpha \leq \norm{x}_X \norm{y}_Y$ for all $x \in X$, $y \in Y$,
\item For all $\phi \in X^\ast$, for all $\psi \in Y^\ast$, $\phi \otimes \psi$ is bounded on $X \otimes Y$ and $\norm{\phi \otimes \psi}_{\alpha^\ast} \leq \norm{\phi}_{X^\ast} \norm{\psi}_{Y^\ast}$,
\end{enumerate} 
where $\norm{\cdot}_{\alpha^\ast}$ is the dual norm to $\norm{\cdot}_\alpha$.
\end{defi}

The injective and projective norms are the smallest and largest reasonable crossnorms we can put on $X \otimes Y$, respectively:
\begin{prop}[{\cite[Proposition 6.1]{Ryan2002}}] \label{prop:reasonable-cross-norms}
Let $X$ and $Y$ be Banach spaces. 
\begin{enumerate}
\item[(a)] A norm $\norm{\cdot}_\alpha$ on $X \otimes Y$ is a reasonable crossnorm if and only if
\begin{equation*}
\norm{z}_{X \otimes_\epsilon Y} \leq \norm{z}_\alpha \leq \norm{z}_{X \otimes_\pi Y}
\end{equation*}
for all $z \in X \otimes Y$.
\item[(b)] If $\norm{\cdot}_\alpha$ is a reasonable crossnorm on $X \otimes Y$, then $\norm{x \otimes y}_\alpha = \norm{x}_X \norm{y}_Y$ for every $x \in X$ and every $y \in Y$. Furthermore, for all $\phi \in X^\ast$ and all $\psi \in Y^\ast$, the norm $\norm{\cdot}_{\alpha^\ast}$ satisfies $\norm{\phi \otimes \psi}_{\alpha^\ast} = \norm{\phi}_{X^\ast} \norm{\psi}_{Y^\ast}$.
\end{enumerate}
\end{prop}

\section{An extension theorem}\label{sec:extension}
In this section, we prove an extension theorem which we will use in Section \ref{sec:maps} to relate the compatibility of measurements in a GPT to properties of an associated map. Along the way, we will use the extension theorem to study tensor products of certain cones.

Let $E\subseteq \mathbb R^d$ be a subspace containing a point with positive coordinates:
\begin{equation}\label{eq:E-strictly-positive}
E \cap \operatorname{ri}(\mathbb R_+^d) \neq \emptyset.
\end{equation}
We set $E^+:=E \cap \mathbb R_+^d$. Note that $E^+$ is proper in $E$  (see \cite[Corollary 6.5.1]{Rockafellar1970} to conclude that $E^+$ is generating).

We shall use the following key extension theorem (see e.g.~\cite[Theorem 1]{Castillo2005}):
\begin{thm}[M.\ Riesz extension theorem] \label{thm:riesz}
Let $(X,X^+)$ be a preordered vector space, $Y \subseteq X$ a linear subspace, and $\phi: Y \to \mathbb R$ a positive linear form on $(Y, Y^+)$, where $Y^+:=Y \cap X^+$. Assume that for every $x \in X$, there exists $y \in Y$ such that $x \leq y$. Then, there exists a positive linear form $\tilde \phi: X \to \mathbb R$ such that $\tilde \phi|_Y = \phi$.
\end{thm}
\begin{remark}
Note that \cite[Theorem 1]{Castillo2005} states the theorem only for ordered vector spaces. However, the same proof works if $X^+$ is an arbitrary convex cone.
\end{remark}

We prove now the main result of this section; see also Remark \ref{rem:subspace-max} for an equivalent formulation. Note that the tensor product of $\mathbb R_+^d$ with a proper cone $L^+$ is unique since $\mathbb R_+^d$ is simplicial. 

\begin{prop} \label{prop:positive-form-extension}
Let $(L,L^+)$ be a proper ordered vector space, and $E \subseteq \mathbb R^d$ as in Equation \eqref{eq:E-strictly-positive}. Any  positive linear form 
\begin{equation}\label{eq:extend-cone}
\phi: (E \otimes L,(E \otimes L) \cap (\mathbb R_+^d \otimes L^+)) \to \mathbb R
\end{equation}
can be extended to a positive linear form $\tilde \phi: \mathbb R^d \otimes L \to \mathbb R$. 
\end{prop}
\begin{proof}
We shall use Theorem \ref{thm:riesz} with $X = \mathbb R^d \otimes L$, $X^+ = \mathbb R^d_+ \otimes L^+$ and $Y = E \otimes L$. We have to show that for any $x \in \mathbb R^d \otimes L$, there is a $y \in E \otimes L$ such that $y-x \in  \mathbb R^d_+ \otimes L^+$. It is enough to consider simple tensors of the form $ x= r \otimes v$, where $r \in \mathbb R^d$ and $v \in L$; the general case will follow by linearity. Since $L^+$ is generating, there are $v_+$, $v_- \in L^+$ such that $v = v_+ -v_-$. Furthermore, from the assumption \eqref{eq:E-strictly-positive}, $E$ contains a vector with strictly positive coordinates $e$, hence there exist $\lambda_\pm > 0$ such that $\lambda_+ e - r \geq 0$ and  $\lambda_- e + r \geq 0$. Then,
\begin{equation*}
\lambda_+ e \otimes v_+ + \lambda_- e \otimes v_- - r \otimes v = \lambda_+ e \otimes v_+ + \lambda_- e\otimes v_- - r \otimes v_+ + r \otimes v_- \in \mathbb R^d_+ \otimes L^+
\end{equation*}
and $\lambda_\pm e \otimes v_\pm \in E \otimes L$. Thus, we can choose $ y = \lambda_+ e \otimes v_+ + \lambda_- e \otimes v_-$.
\end{proof}

 We can now identify the dual proper ordered vector space $(E^*,(E^+)^*)$. 
\begin{prop}\label{prop:E_dual}
Let us identify the dual vector space $E^*\cong E$, with duality given by the standard inner product in $\mathbb R^d$. We then have
$(E^+)^*=J(\mathbb R^d_+)$, where  $J:\mathbb R^d\to E$ is the orthogonal projection onto $E$.
For a proper ordered vector space $(L,L^+)$, we have  
\[
(E^+)^*\otimesmin L^+=(J\otimes \operatorname{id}_L)(\mathbb R^d_+\otimes L^+), 
\]
\end{prop}
 
\begin{proof}
Let $r\in \mathbb R^d_+$, then for $e\in E^+ \subseteq \mathbb R^d_+$
\[
\langle J(r),e\rangle=\langle r, e\rangle \ge 0,
\]
since $\mathbb R^d_+$ is self-dual, so that $J(\mathbb R^d_+)\subseteq (E^+)^*$. 
By Proposition
\ref{prop:positive-form-extension}, any $\varphi\in (E^+)^*$ extends to a positive form $\tilde \varphi: \mathbb
R^d\to \mathbb R$, so that $\varphi=J(\tilde \varphi)$, which implies the reverse inclusion.
If $(L,L^+)$ is a proper ordered vector space, then $(E^+)^*\otimesmin L^+$ is a cone of elements of the form 
\[
\sum_j \varphi_j\otimes v_j= \sum_j J(\tilde \varphi_j)\otimes v_j= (J\otimes \operatorname{id}_L)(\sum_j \tilde\varphi_j\otimes v_j)
\]
for $v_j\in L^+$ and $\tilde \varphi_j\in \mathbb R^d_+$, proving the last statement.
\end{proof}

We now provide a useful characterization of the maximal tensor product of $E^+$ with $L^+$, identifying at the same time the cone appearing in Equation \eqref{eq:extend-cone}.

\begin{prop} \label{prop:rightcone}
For $E$ as above and $(L,L^+)$ a proper ordered vector space, it holds that
\begin{equation*}
E^+ \otimes_{\mathrm{max}} L^+ = (E \otimes L) \cap (\mathbb R_+^d \otimes L^+).
\end{equation*}
\end{prop}
\begin{proof}
The inclusion ``$\subseteq$'' follows from the monotonicity of the $\max$ tensor product with respect to each factor. To show the reverse inclusion ``$\supseteq$'', we have to prove that for any $z \in (E \otimes L) \cap (\mathbb R_+^d \otimes L^+)$, and for any $\beta \in (E^+)^*$, $\alpha \in (L^+)^*$, we have that $\langle \beta \otimes \alpha, z \rangle \geq 0$.

By Proposition \ref{prop:E_dual}, $\beta \in  (E^+)^*$ 
implies that $\beta=J(\tilde \beta)$ for a positive form $\tilde \beta: \mathbb R^d \to \mathbb R$.  Since $z \in \mathbb R^d_+ \otimes L^+$, we have a decomposition $z = \sum_{i = 1}^d r_i \otimes v_i$, where $r_i \in \mathbb R^d_+$ and $v_i \in L^+$. 
This yields 
$$\langle \beta \otimes \alpha, z \rangle = \langle \tilde \beta \otimes \alpha, z \rangle =  \sum_{i = 1}^d \tilde \beta(r_i)\alpha( v_i) \geq 0,$$
finishing the proof.
\end{proof}

\begin{cor} \label{cor:E-factors}
For $E_1$, $E_2$ satisfying Equation~\eqref{eq:E-strictly-positive} 
(but not necessarily of the same dimension), it holds that
\begin{equation*}
E_1^+ \otimes_{\mathrm{max}} E_2^+ = (E_1 \otimes E_2) \cap (\mathbb R_+^{d_1} \otimes \mathbb R_+^{d_2}).
\end{equation*}
\end{cor}
\begin{proof}
The proof is almost the same as for Proposition \ref{prop:rightcone}, only that the functionals on both subspaces need to be extended.
\end{proof}

\begin{remark}\label{rem:subspace-max}
Using Proposition \ref{prop:rightcone}, one can restate Proposition \ref{prop:positive-form-extension} as follows: Any  positive linear form $\phi: (E \otimes L,E^+ \otimes_{\max} L^+) \to \mathbb R$ can be extended to a positive linear form $\tilde \phi: \mathbb R^d \otimes L \to \mathbb R$ for proper $L^+$.
\end{remark}

 We now study extendability of general positive maps on $(E,E^+)$. 
 The proof technique is inspired by the finite dimensional version of Arveson's extension theorem
\cite[Theorem 6.2]{Paulsen2002}.

\begin{prop}\label{prop:positivity-extension}
Let $E \subseteq \mathbb R^d$ be a subspace such that $E \cap \operatorname{ri}(\mathbb R_+^d) \neq \emptyset$ and $E^+
= E \cap \mathbb R^d_+$. Let $J$ be the orthogonal projection onto $E$ and let $(L,L^+)$ be an proper ordered vector space.
 Finally, let $\Phi: E \to L^\ast$ be a linear map. The following are equivalent:
\begin{enumerate}
\item There exists a positive extension $\tilde \Phi: (\mathbb R^d, \mathbb R_+^d) \to (L^\ast, (L^+)^\ast)$ of $\Phi$.
\item The linear map $\Phi \otimes \mathrm{id}_L: (E \otimes L, E^+ \otimes_{\max} L^+) \to (L^\ast \otimes L, (L^+)^\ast \otimes_{\min} L^+)$ is positive.
\item  The form 
$s_\Phi: (E \otimes L, E^+ \otimes_{\max} L^+) \to \mathbb R$
is positive.
\item $\varphi^\Phi\in (J\otimes \operatorname{id}_{L^\ast})(\mathbb R^d_+\otimes (L^+)^\ast)$.
\item $\Phi$ is entanglement breaking.
\
\end{enumerate}
\end{prop}

\begin{proof} By Propositions \ref{prop:EB} and \ref{prop:E_dual}, the statements (2) - (5) are equivalent. It is
therefore enough to show that (1) $\implies$ (2) and (3) $\implies$ (1). 
We start by showing that the existence of the positive extension (1) implies the positivity of the map $\Phi \otimes \mathrm{id}_L$ in (2). Both $\tilde \Phi$ and $\mathrm{id}_L$ are positive maps. Therefore, $\tilde \Phi \otimes \mathrm{id}_L: (\mathbb R^d \otimes L, \mathbb R_+^d \otimes_{\min} L^+ =  \mathbb R_+^d \otimes_{\max} L^+) \to (L^\ast \otimes L, (L^+)^\ast
\otimes_{\min} L^+)$ is positive. The claim follows by Proposition \ref{prop:rightcone},
since $\Phi \otimes \mathrm{id}_L$ is a restriction of this map to $E \otimes L$. 

It remains to show that $(3) \implies (1)$. Using Proposition \ref{prop:positive-form-extension} and Remark
\ref{rem:subspace-max}, we extend the form $s_\Phi$ to $\tilde s_\Phi : \mathbb R^d \otimes L \to \mathbb R$. Let
$\tilde \Phi: (\mathbb R^d, \mathbb R_+^d) \to (L^\ast,(L^+)^\ast)$ be the positive map related to $\tilde s_\Phi$ as in Equation~\eqref{eq:Phi_sPhi}. It remains to check that $\tilde \Phi$ is indeed an extension of $\Phi$. Let $v_j$ and $\alpha_j$ be elements of a basis of $L$ and its dual basis for all $j \in [\dim L$]. For any $e \in E$, we compute
\begin{align*}
\tilde \Phi(e) = \sum_{j=1}^{\dim L} \tilde s_\Phi(e \otimes v_j) \alpha_j= \sum_{j=1}^{\dim L} s_\Phi(e \otimes v_j)
\alpha_j=\Phi(e),
\end{align*}
finishing the proof.
\end{proof}

\begin{remark}
Note that the existence of a positive extension of $\Phi: E \to L^\ast$ can be checked using conic programming, see Section \ref{sec:conic-programming-map} in the Appendix.
\end{remark}

\section{Generalized spectrahedra} \label{sec:gen-spectrahedra}

In this section, we will generalize some of the theory of (free) spectrahedra to the setting of ordered vector spaces by allowing for more general cones than the cone of positive semidefinite matrices. We will reformulate the compatibility of measurements as an inclusion problem for generalized spectrahedra in Section \ref{sec:gen-spectra-and-comp}.

Recall that a spectrahedron \cite{ramana1995some, vinzant2014spectrahedron} is a convex set that can be represented by a linear matrix inequality, that is by positive semidefinite constraints. We define a generalized spectrahedron as a convex subset of some vector space which can be represented by positivity conditions with respect to some abstract cone. 

\begin{defi}
	Let $L,M$ be two finite-dimensional vector spaces and consider a cone $C \subseteq M \otimes L$. For a $g$-tuple of elements $a = (a_1, \ldots, a_g) \in M^g$, we define the \emph{generalized spectrahedron}
	$$\mathcal D_{a}(L,C) := \{(v_1, \ldots, v_g) \in L^g \, : \, \sum_{i=1}^g a_i \otimes v_i  \in C\}.$$
\end{defi}
\begin{remark}\label{rem:spectrahedra}
It is easy to see that any generalized spectrahedron is a convex cone and that the generalized spectrahedron is closed if $C$ is. In fact, note that the $g$-tuple $a\in M^g$ defines a linear map $a: \mathbb{R}^g\to M$, by $x\mapsto \sum_i x_ia_i$. The generalized spectrahedron $\mathcal D_{a}(L,C)$ is the largest cone in $\mathbb R^g\otimes L\cong L^g$ that makes  
the map
$$a\otimes \mathrm{id}_L : (\mathbb R^g \otimes L, \mathcal D_{a}(L,C)) \to (M \otimes L, C)$$
positive. If $a$ is a basis of $M$, then the corresponding map $a\otimes \mathrm{id}_L$ is an isomorphism through which the cones $C$ and $\mathcal D_a(L,C)$ are affinely isomorphic.
\end{remark}

Note that usual spectrahedra correspond to the choice $L = \mathbb R$, $M = \mathcal M^{\mathrm{sa}}_n(\mathbb C)$, and $C$ being the positive semidefinite cone. Free spectrahedra are the union over $d \geq 1$ of generalized spectrahedra for $L = \mathcal M^{\mathrm{sa}}_d(\mathbb C$) and $C$ being the PSD cone of $dn \times dn$ matrices; note that there exists no natural notion of generalized \emph{free} spectrahedra, since there are no canonical sequences of cones $(C_d)_{d \geq 1}$, with $C_1=C$.

We will now consider generalized spectrahedra which are in some sense minimal and maximal. The definitions are inspired by the corresponding notions for matrix convex sets, see \cite{davidson2016dilations, passer2018minimal}.
\begin{defi} \label{def:D_min}
Let $\mathcal C \subset \mathbb  R^{g}$ be a closed convex cone. Let $(L,L^+)$ be a finite-dimensional preordered vector space. Then
\begin{equation*}
    \mathcal D_{\min}(\mathcal C; L, L^+):= \left \{\sum_i x^{(i)}\otimes h_i\in \mathbb{R}^g\otimes L\cong L^g:\  x^{(i)} \in \mathcal C,~h_i \in L^+~\forall i\right\} \cong \mathcal C \otimes_{\min} L^+
\end{equation*}
is the \emph{minimal} generalized spectrahedron corresponding to $\mathcal C$.
\end{defi}

\begin{defi} \label{def:D_max}
Let $\mathcal C \subset \mathbb R^{g}$ be a closed convex cone. Let $(L,L^+)$ be a finite-dimensional preordered vector space with $L^+$ closed. Then,
\begin{equation*}
     \mathcal D_{\max}(\mathcal C; L, L^+):= \left\{(v_1, \ldots, v_g) \in L^g: \sum_{i = 1}^g h_i v_i \in L^+~\forall h \in \mathbb R^{g} \mathrm{~s.t.~} \sum_{i = 1}^g h_i c_i\geq 0~\forall c \in \mathcal C\right\} \cong \mathcal C \otimes_{\max} L^+
\end{equation*}
is the \emph{maximal} generalized spectrahedron corresponding to $\mathcal C$. 
\end{defi}
In order to show that these sets are indeed generalized spectrahedra and to justify the identifications with minimal and maximal cones, we prove a lemma.
\begin{lem}
Let $(L, L^+)$ be a preordered vector space. Let $e = (e_1, \ldots, e_g)$ be the canonical basis of $\mathbb R^g$. Then, 
\begin{equation*}
       \mathcal D_{\min}(\mathcal C; L, L^+) \cong \mathcal D_{e}(L, \mathcal C \otimes_{\min} L^+).
\end{equation*}
If $L^+$ is closed, then
\begin{equation*}
       \mathcal D_{\max}(\mathcal C; L, L^+) \cong \mathcal D_{e}(L, \mathcal C \otimes_{\max} L^+).
\end{equation*}
\end{lem}
\begin{proof}
By definition, $(v_1, \ldots, v_g) \in \mathcal D_{\min}(\mathcal C; L, L^+)$ if and only if $v_j = \sum_{i} x_j^{(i)} h_i$ for some $x^{(i)} \in \mathcal C$ and $h_i \in L^+$. Moreover, $\sum_{j = 1}^g e_j \otimes v_j \in \mathcal C \otimes_{\min} L^+$ if and only if
\begin{equation*}
    \sum_{j = 1}^g e_j \otimes v_j = \sum_{i} x^{(i)} \otimes h_i \qquad  x^{(i)} \in \mathcal C, h_i \in L^+
\end{equation*}
A comparison of coordinates proves the first assertion. Let now $\sum_{i = 1}^g e_i \otimes v_i \in \mathcal C \otimes_{\max} L^+$. By definition, this is true if and only if 
\begin{equation*}
    \sum_{i = 1}^g \phi(e_i) v_i \in L^+  \qquad \forall \phi \in \mathcal C^\ast
\end{equation*}
since $L^+$ is closed. Realizing that for $c = \sum_{i = 1}^g c_i e_i$, we can write $\phi(c) = \sum_{i = 1}^g c_i h_i \geq 0$ with $h_i = \phi(e_i)$ for all $i \in [g]$, proves one inclusion. The second assertion follows by identifying $(h_1, \ldots, h_g)$ as in the assertion with some $\phi \in \mathcal C$. 
\end{proof}

The following proposition justifies that we speak of minimal and maximal cones.
\begin{prop} \label{prop:min-max_sandwich}
    Let $\mathcal C \subset \mathbb R^{g}$ be a closed convex cone. Let $(L,L^+)$, $(M,M^+)$ be two finite-dimensional preordered vector spaces with a tensor cone $C \subset M \otimes L$. Let $M^+$, $L^+$ be closed convex cones.  If $a \in M^g$ such that $\mathcal D_a(\mathbb R, M^+) = \mathcal C$, then
    \begin{equation*}
        \mathcal D_{\min}(\mathcal C; L, L^+) \subseteq \mathcal D_a(L,C) \subseteq  \mathcal D_{\max}(\mathcal C; L, L^+).
    \end{equation*}
Moreover, for a closed convex $\mathcal C^\prime$ such that $\mathcal C^\prime \subseteq \mathcal C$, it holds that
\begin{equation*}
     \mathcal D_{\min}(\mathcal C^\prime; L, L^+) \subseteq  \mathcal D_{\min}(\mathcal C; L, L^+) \quad \mathrm{and} \quad \mathcal D_{\max}(\mathcal C^\prime; L, L^+) \subseteq  \mathcal D_{\max}(\mathcal C; L, L^+).
\end{equation*}
\end{prop}
\begin{proof}
Let $v \in \mathcal D_{\min}(\mathcal C; L, L^+) $. Then, $v = (v_1, \ldots, v_g)$ where
\begin{equation}
    v_j = \sum_i x^{(i)}_j h_i \qquad \forall j \in [g]
\end{equation}
for some $x^{(i)} \in \mathcal C$ and $h_i \in L^+$. Thus, we can rewrite
\begin{equation*}
    \sum_{j = 1}^g a_j \otimes v_j = \sum_{i} \left(\sum_{j = 1}^g x_j^{(i)} a_j \right) \otimes h_i
\end{equation*}
Since $\mathcal D_a(\mathbb R, M^+) = \mathcal C$, it holds that
\begin{equation*}
    \sum_{j = 1}^g x_j^{(i)} a_j \in M^+ \qquad \forall i.
\end{equation*}
Thus, since $C$ is a tensor cone,
\begin{equation*}
   \sum_{j = 1}^g a_j \otimes v_j \in M^+ \otimes_{\min} L^+ \subseteq C.
\end{equation*}
Therefore, we can conclude that $v \in \mathcal D_a(L,C)$, which proves the first assertion. 

For the second assertion, let $v =(v_1, \ldots, v_g) \in D_a(L,C)$. Thus
\begin{equation*}
    \sum_{j = 1}^g a_j \otimes v_j \in C \subseteq M^+ \otimes_{\max} L^+,
\end{equation*}
since $C$ is a tensor cone. Thus,
\begin{equation*}
    \sum_{j = 1}^g \varphi(a_j) \psi(v_j) \geq 0 \qquad \forall \varphi \in (M^+)^\ast, \psi \in (L^+)^\ast.
\end{equation*}
Using that $M^+ \cong (M^+)^{\ast \ast}$ and that $\mathcal D_a(\mathbb R, M^+) = \mathcal C$, it follows that $(\psi(v_1), \ldots, \psi(v_g)) \in \mathcal C$ for all $\psi \in (L^+)^\ast$. Hence, for all $\psi \in (L^+)^\ast$, it holds that
\begin{equation*}
    \sum_{i = 1}^g h_i \psi(v_i)\geq 0~\forall h \in \mathbb R^{g} \mathrm{~s.t.~} \sum_{i = 1}^g h_i c_i\geq 0~\forall c \in \mathcal C.
\end{equation*}
Using that $L^+ \cong (L^+)^{\ast \ast}$, it follows that $v \in \mathcal D_{\max}(\mathcal C; L, L^+)$. The inclusions in the last assertion follow directly from the Definitions \ref{def:D_min} and \ref{def:D_max}.
\end{proof}

\begin{remark} In view of Remark \ref{rem:spectrahedra}, the first inclusion of the above proposition follows directly from the fact that the map $a: (\mathbb{R}^g,\mathcal C)\to (M,M^+)$ is positive and for any positive map, its tensor product with the identity map is positive with respect to the minimal tensor cone. 
Since we have $\mathcal C^*=a^*((M^+)^*)$, the second inclusion is a consequence of the same property of $a^*$.
\end{remark}

\begin{remark} Assume that $M^+$ and $L^+$ are proper cones.
By the results of \cite{Aubrun2019a}, Proposition \ref{prop:min-max_sandwich} implies that the generalized spectrahedron  $\mathcal D_a(L,C)$ does not  depend on the choice of the tensor cone $C$  if and only if at least one of the cones $\mathcal C=\mathcal D_a(\mathbb{R},M^+)$ or $L^+$ is simplicial. 
\end{remark}

As for free spectrahedra, we can connect the inclusion of generalized spectrahedra to the positivity of an associated map \cite{helton_matricial_2013,  davidson2016dilations, helton2019dilations}.

\begin{prop} \label{prop:containment-to-map-inclusion}
Let $L$, $M$ and $N$ be finite-dimensional vector spaces and let $C_M \subset 
M \otimes L$, $C_N \subset N \otimes L$ be two cones. Moreover, let $a \in M^g$, $b \in N^g$ be two tuples, where $a$ consists of linearly independent elements which span the subspace $M^\prime \subseteq M$. We define a map $\Phi: M^\prime \to N$, $\Phi(a_i) = b_i$ for all $i \in [g]$. Then,
\begin{equation*}
\mathcal D_a(L, C_M) \subseteq \mathcal D_b(L, C_N) 
\end{equation*}
if and only if $\Phi \otimes \mathrm{id}_L: (M^\prime \otimes L,  M^\prime \otimes L \cap C_M) \to (N \otimes L, C_N)$ is positive.
\end{prop} 
\begin{proof}
Let us assume the inclusion. Consider $z \in M^\prime \otimes L$. Thus, we can write
\begin{equation*}
z = \sum_{i = 1}^g a_i \otimes z_i,
\end{equation*}
where $z_i \in L$ for all $i \in [g]$. If $z \in C_M$, then $(z_1, \ldots, z_g) \in \mathcal D_a(L, C_M)$ and hence
\begin{equation*}
(\Phi \otimes \mathrm{id}_L)(z) = \sum_{i = 1}^g b_i \otimes z_i .
\end{equation*}
The right hand side is in $C_N$ since $\mathcal D_a(L, C_M) \subseteq \mathcal D_b(L, C_N)$. Conversely, let $\Phi \otimes \mathrm{id}_L$ be positive. Let $(v_1, \ldots, v_g) \in \mathcal D_a(L, C_M)$. Then,
\begin{equation*}
\sum_{i = 1}^g a_i \otimes v_i \in (M^\prime \otimes L) \cap C_M 
\end{equation*}
and the assertion follows from an application of $\Phi \otimes \mathrm{id}_L$ to this element.
\end{proof}

As for free spectrahedra, it is possible to look at what inclusion with respect to some preordered vector space $(L,L^+)$ implies for the inclusion with respect to $(\mathbb R, \mathbb R_+)$.
\begin{prop}
Let $(M,M^+)$, $(N,N^+)$ and $(L,L^+)$ be preordered vector spaces, where $N^+$ is closed and $L^+$ contains at least one element which is not in $\overline{-L^+}$. Moreover, let $C_M \subset M \otimes L$ and $C_N \subset N \otimes L$ be tensor cones and $a \in M^g$, $b \in N^g$ for some $g \in \mathbb N$. Then, 
\begin{equation*}
\mathcal D_a(L, C_M) \subseteq \mathcal D_b(L, C_N) \implies \mathcal D_a(\mathbb R, M^+) \subseteq \mathcal D_b(\mathbb R, N^+) 
\end{equation*}
\end{prop}
\begin{proof}
Let $v \in L^+$, $v \not \in -\overline{L^+}$ and $x \in \mathcal D_a(\mathbb R, M^+)$. Then, $(x_1 v, \ldots, x_g v) \in \mathcal D_a(L, C_M)$, since $C_M$ contains in particular $M^+ \otimes_{\min}L^+$. Thus, 
\begin{equation*}
\left(\sum_{i = 1}^g x_i b_i\right) \otimes v \in C_N.
\end{equation*}
This implies that $\sum_{i = 1}^g x_i b_i \in (N^+)^{\ast \ast} \cong N^+$, since $C_N \subseteq N^+ \otimes_{\max} L^+$ and we can find $\phi \in  (L^+)^\ast$ such that $\phi(v) > 0$ since $\psi(v) = 0$ for all $\psi \in L^*$ implies $-v \in \overline{L^+}$. 
Hence $x \in \mathcal D_b(\mathbb R, N^+)$.
\end{proof}
In general, $\mathcal D_a(\mathbb R, M^+) \subseteq \mathcal D_b(\mathbb R, N^+)$ does not imply $\mathcal D_a(L, C_M) \subseteq \mathcal D_b(L, C_N)$. However, if $M$ and $N$ contain order units, the implication can be made true by shrinking the left hand side.

\begin{defi} \label{def:inclusion-constants}
Let $(M,M^+)$, $(N,N^+)$ and $(L,L^+)$ be preordered vector spaces where $M^+$ and $N^+$ contain order units $\mathds 1_M$ and $\mathds 1_N$, respectively. Moreover, let $C_M \subset M \otimes L$ and $C_N \subset N \otimes L$ be tensor cones and $a \in M^g$, $b \in N^g$ for some $g \in \mathbb N$. The \emph{set of inclusion constants for $\mathcal D_a(L, C_M)$ and $C_N$} is defined as 
\begin{align*}
\Delta_a(L, C_M, C_N) :=
&\{ s\in [0,1]^g: \forall b \in N^g,~ \mathcal D_{(\mathds 1_M, a)}(\mathbb R, M^+) \subseteq \mathcal D_{(\mathds 1_N, b)}(\mathbb R, N^+) \\ &\implies (1,s) \cdot \mathcal D_{(\mathds 1_M, a)}(L, C_M) \subseteq \mathcal D_{(\mathds 1_N, b)}(L, C_N) \}.
\end{align*}
Here, $(1,s) \cdot \mathcal D_{(\mathds 1_M, a)}(L, C_M) := \{(v_0, s_1 v_1, \ldots, s_g v_g): v \in \mathcal D_{(\mathds 1_M, a)}(L, C_M)\}$.
\end{defi}

\begin{question}
Are there natural conditions which would entail
$\{s: \sum_i s_i \leq 1 \} \subseteq \Delta_a(L, C_M, C_N)$? This can be done for free spectrahedra, see \cite[Theorem 1.4]{helton2019dilations} and \cite[Section 8]{davidson2016dilations}.
\end{question}

The following definition is motivated by the matrix range introduced in \cite{Arveson1972} and generalized in \cite{davidson2016dilations}.
\begin{defi}
Let $(L,L^+)$ be a preordered vector space with order unit $\mathds 1_L$ and consider a tuple $a \in L^g$. Then, the \emph{functional range} of $a$ is the set
\begin{equation*}
\mathcal W(a) = \{(\phi(a_1), \ldots, \phi(a_g)): \phi \in (L^+)^\ast, \phi(\mathds 1_L) = 1 \} \subseteq \mathbb R^g.
\end{equation*}
\end{defi}

\begin{prop}
Let $(L,L^+)$ be a preordered vector space with order unit $\mathds 1_L$ and let $a \in L^g$. If $V^+$ is proper, then $\mathcal W(a)$ is compact and convex. The set $\mathcal W(a)$ is non-empty if and only if $L^+ \neq L$.
\end{prop}
\begin{proof}
For any $a_i$, $i \in [g]$, there is a $t_i \geq 0$ such that $a_i \in t_i [-\mathds 1_L, \mathds 1_L]$. Thus, $|\phi(a_i)| \leq t_i$ for all $ \phi \in (L^+)^\ast, \phi(\mathds 1_L) = 1$ and boundedness of $\mathcal W(a)$ follows. Let $x^{(n)}$ be a sequence in $\mathcal W(a)$ converging to $x$. With any $x^{(n)}$, we can associate a $\phi_n \in (L^+)^\ast, \phi_n(\mathds 1_L) = 1$. Since $(L^+)^\ast$ is closed by definition, there is a map $ \phi \in (L^+)^\ast, \phi(\mathds 1_L) = 1$ such that $x = (\phi(a_1), \ldots, \phi(a_n))$ and $x \in \mathcal W(a)$. This follows from the Bolzano-Weierstrass theorem and the fact that the unital positive functionals form a compact set (consider the order unit norm on $L$). Convexity follows from the fact that the set of $\phi$ as in the statement is convex.

The set $\mathcal W(a)$ is empty if and only if there are no unital functionals in $L^\ast$. Let $L = L^+$. Then, it is easy to see that $(L^+)^{\ast} = \{0\}$. Conversely, if $L^+ \neq L$, then $- \mathds 1_L  \not \in L^+$ by Lemma \ref{lem:trivial-cone}. 
Thus, the functional $\phi$ on $\mathbb R\mathds 1_L$ such that $\phi(\mathds 1_L) = 1$ can be extended to an element in $(L^+)^*$ by Theorem \ref{thm:riesz}.
\end{proof}

\begin{prop}\label{prop:matrix-range-duality}
Let $(L,L^+)$ be a proper ordered vector space with order unit $\mathds 1_L$. Furthermore, let $a \in L^{g}$. Let $\mathcal C_a := \{x \in \mathbb R^g: (1,-x) \in \mathcal D_{(\mathds 1_L, a)}(\mathbb R, L^+)\}$. Then, $\mathcal W(a)^\circ = \mathcal C_a$. If $0 \in \mathcal W(a)$, then $\mathcal C_a^\circ = \mathcal W(a)$.
\end{prop}
\begin{proof} 
Since $L^+$ is closed, $L^+ \cong (L^+)^{\ast \ast}$ and 
\begin{equation*}
\mathds 1 + \sum_i x_i a_i \in L^+ \iff 1 + \sum_i x_i \varphi(a_i) \geq 0 \qquad \forall \phi \in (L^+)^\ast, \phi(\mathds 1_L) = 1.
\end{equation*}
Note that the only map in $(L^+)^\ast$ with $\phi(\mathds 1_L) = 0$ is the constant map, since $\mathds 1_L$ is an order unit. Thus, it is enough to check unital maps for the $\Longleftarrow$ implication. This proves the first assertion. The second assertion follows from the bipolar theorem for convex sets (\cite[Equation  (1.10)]{aubrun_alice_2017}) since $\mathcal W(a)$ is closed.
\end{proof}

\begin{prop} \label{prop:zero-interior-range}
Let $(L,L^+)$ be a proper ordered vector space with order unit $\mathds 1_L$. Furthermore, let $a \in L^{g}$. Then, $\mathcal C_a$ is bounded if and only if $0 \in \mathrm{int}~\mathcal W(a)$.
\end{prop}
\begin{proof}
This follows from the fact that for convex sets $K \subseteq \mathbb R^n$, $K^\circ$ is bounded if and only if $0 \in \mathrm{int}~K$ \cite[Exercise 1.14]{aubrun_alice_2017}, combined with Proposition \ref{prop:matrix-range-duality}.
\end{proof}

\begin{prop} \label{prop:direct-sum-of-gen-spectrahedra}
Let $(M,M^+)$ and $(N,N^+)$ be two proper ordered vector spaces containing order units $\mathds 1_M$ and $\mathds 1_N$. Let $a \in M^{k}$ and $b \in N^{l}$ be such that $\mathcal C_a$, $\mathcal C_b$ are polytopes for $k$, $l \in \mathbb N$. Then, for any closed tensor cone $C_{MN}$,
\begin{equation*}
\mathcal D_{(\mathds 1_M \otimes \mathds 1_N, a \otimes \mathds 1_N, \mathds 1_M \otimes b)}(\mathbb R, C_{MN}) = \mathbb R_+\{(1,-z): z\in \mathcal C_a \oplus \mathcal C_b\}.
\end{equation*}
\end{prop}
\begin{proof}
Let $-x \in \mathcal \mathcal C_{a}$, $-y \in \mathcal C_b$. Then, $(1,x,0) \in \mathcal D_{(\mathds 1_M \otimes \mathds 1_N, a \otimes \mathds 1_N, \mathds 1_M \otimes b)}(\mathbb R, C_{MN})$, because $(\mathds 1_M + \sum_i x_i a_i) \otimes \mathds 1_N \in M^+ \otimes_{\min} N^+$, and likewise $(1,0,y) \in \mathcal D_{(\mathds 1_M \otimes \mathds 1_N, a \otimes \mathds 1_N, \mathds 1_M \otimes b)}(\mathbb R,C_{MN})$, such that ``$\supset$'' holds in the assertion. Conversely, let $(c,x,y) \in \mathcal D_{(\mathds 1_M \otimes \mathds 1_N, a \otimes \mathds 1_N, \mathds 1_M \otimes b)}(\mathbb R, C_{MN})$, where $x \in \mathbb R^k$ and $y \in \mathbb R^l$. Then, 
\begin{equation}\label{eq:gen-spectrehedron-direct-sum}
c \mathds 1_M \otimes \mathds 1_N + (\sum_{i=1}^k x_i a_i)\otimes \mathds 1_N + \mathds 1_M \otimes (\sum_{j = 1}^l y_j b_j) \in  C_{MN}.
\end{equation}
Boundedness of $\mathcal C_a$ and $\mathcal C_b$ implies by Proposition \ref{prop:zero-interior-range} that $0 \in \mathrm{int}~\mathcal W(a)$, $0 \in \mathrm{int}~\mathcal W(b)$. Let $\phi$, $\psi$ be positive unital functionals which send $a$ and $b$ to zero, respectively. Then, $\phi \otimes \psi \in (C_{MN})^\ast$. An application of this map to Equation \eqref{eq:gen-spectrehedron-direct-sum} implies $c \geq 0$. 

Let $c = 0$. Then, by applying $\phi \otimes \beta$, $\alpha \otimes \psi$ for $\alpha \in (M^+)^\ast$ and $\beta \in (N^+)^\ast$, it holds that $\sum_i x_i a_i \in M^+$, $\sum_j y_j b_j \in N^+$. Since $\mathcal C_a$ is bounded, $\{a_i\}_i$ is a set of linearly independent elements. The same is true for $\{b_j\}_j$. Without loss of generality, let $x \neq 0$. Then $\sum_i x_i a_i \neq 0$ and $(1,- \lambda x) \in \mathcal C_a$ for all $\lambda \geq 0$, which is a contradiction to $\mathcal C_a$ being bounded. Thus $c = 0$ implies $x = 0 = y$.

Thus, we can set $c = 1$ without loss of generality. Let now for $\phi^\prime \in (M^+)^\ast$, $\psi^\prime \in (N^+)^\ast$, $\phi^\prime(\mathds 1_M) = 1 = \psi^\prime(\mathds 1_N)$. An application of $\phi^\prime \otimes \psi^\prime$ to Equation \eqref{eq:gen-spectrehedron-direct-sum} implies $\mathcal C^\circ \supset W(a) \times \mathcal W(b)$, where $\mathcal C:= \{z:(1,-z) \in \mathcal D_{(\mathds 1_M \otimes \mathds 1_N, a \otimes \mathds 1_N, \mathds 1_M \otimes b)}(\mathbb R, C_{MN})\}$. Since $\mathcal C$ is closed and contains $0$, we obtain $\mathcal C \subset (\mathcal W(a) \times \mathcal W(b))^\circ = (\mathcal C_a^\circ \times \mathcal C_b^\circ)^\circ = \mathcal C_a \oplus \mathcal C_b$ with Proposition \ref{prop:matrix-range-duality}.  
\end{proof}

\begin{cor} \label{cor:direct-sum-splits}
Let $(M,M^+)$ and $(N,N^+)$ be two proper ordered vector spaces with closed cones containing order units $\mathds 1_M$ and $\mathds 1_N$. Let $a \in M^{k}$ and $b \in N^{l}$ be such that $\mathcal C_a$, $\mathcal C_b$ are polytopes for $k$, $l \in \mathbb N$. Let moreover $(L, L^+)$ be another preordered vector space with order unit $\mathds 1_L$ and $h_1 \in L^{k}$, $h_2 \in L^{l}$. Then, for any closed tensor cone $C_{MN}$,
\begin{equation*}
\mathcal D_{(\mathds 1_M \otimes \mathds 1_N, a \otimes \mathds 1_N, \mathds 1_M \otimes b)}(\mathbb R, C_{MN}) \subseteq \mathcal D_{(\mathds 1_L, h_1, h_2)}(\mathbb R, L^+)
\end{equation*}
if and only if
\begin{equation*}
\mathcal D_{(\mathds 1_M, a)}(\mathbb R, M^+) \subseteq \mathcal D_{(\mathds 1_L, h_1)}(\mathbb R, L^+) \quad \land \quad \mathcal D_{(\mathds 1_N, b)}(\mathbb R, N^+) \subseteq \mathcal D_{(\mathds 1_L, h_2)}(\mathbb R, L^+).
\end{equation*}
\end{cor}

\section{General probabilistic theories}\label{sec:gpt}
In this section, we finally introduce the class of physical theories we are interested in, the so-called \emph{general probabilistic theories} (GPTs). They form the framework that is used to describe states and measurement outcomes of arbitrary (physical) theories. It is within this generalized setting that we would like to study questions about the compatibility of measurements. In this brief introduction, we mostly follow the exposition in \cite{lami2018non}.

\subsection{Definitions}
Any GPT corresponds to a triple $(V,V^+,\mathds 1)$, where $V$ is a vector space with  a proper cone $V^+$
and $\mathds 1$
is an order unit in the dual cone $A^+=(V^+)^*\subset V^*=A$. We assume here that $V$ is finite dimensional. The set of
states of the system is identified as the  subset 
\[
K:=\{v\in V^+,\ \langle\mathds 1,v\rangle=1\}.
\]
Note that $K$ is compact and  convex and is a base of the cone $V^+$. 

\begin{ex}\label{ex:CM} Any \emph{classical system} is described by the triple $\mathrm{CM}_d := (\mathbb R^d,\mathbb R^d_+,1_d)$, $d \in \mathbb N$, where $\mathbb R^d_+$
denotes the set of elements with non-negative coordinates and $1_d=(1,1,\dots,1)=\sum_ie_i\in \mathbb R^d$; here
$e_1,\dots, e_d$ denotes the standard basis on $\mathbb R^d$.
Then $(\mathbb R^d)^*=\mathbb R^d$ with duality given by the standard inner product and the simplicial cone
$\mathbb R^d_+$ is self-dual.  The classical state space is the probability simplex
\begin{align*}
\Delta_d=\left\{x=(x_1,\dots,x_d)\in \mathbb R^d,\ x_i\ge 0,\ \sum_{i=1}^d x_i=1\right\}= \{x\in \mathbb R^d_+,\ \langle x, 1_d\rangle =1\}.
\end{align*}
\end{ex}

\begin{ex}\label{ex:QM} \emph{Quantum mechanics} corresponds the triple $\mathrm{QM}_d:=(\mathcal M_d^{\mathrm{sa}}(\mathbb C), \mathrm{PSD}_d ,\operatorname{Tr})$, $d \in \mathbb N$, where  $\mathrm{PSD}_d$ is the cone of $d \times d$ positive semidefinite complex, self-adjoint matrices, and $\operatorname{Tr}$ is the usual, un-normalized, trace. As in the case of classical systems described above, the $\mathrm{PSD}_d$ cone is self-dual.  The quantum state space is the set of density matrices 
\begin{align*}
\mathcal S_d := \{\rho \in \mathcal M_d^{\mathrm{sa}}(\mathbb C)\, : \,  \rho \geq 0, \ \operatorname{Tr} \rho = 1\}.
\end{align*}
\end{ex}

\begin{ex} The \emph{hypercube} GPT has the unit ball of $\ell_\infty^n$ as a state space, $n \in \mathbb N$. It is described by the triple $\mathrm{HC}_n := (\mathbb R^{n+1}, C_n, \mathds 1)$, where 
$$C_n = \{(x_0, x_1, \ldots, x_n) \in \mathbb R^{n+1} \, : \, x_0 \geq \max_{i \in [n]} |x_i|\} \quad \text{ and } \quad \mathds 1(x_0, x_1, \ldots, x_n) = x_0.$$
\end{ex}

Note also that any compact convex subset in a (finite-dimensional) vector space can be represented as a base of some proper
cone. Indeed, let $K$ be such a set and let $A=A(K)$ be the set of affine functions $K\to \mathbb R$. Then $A$ is a
finite dimensional vector space and the subset $A^+=A(K)^+$ is a proper cone in $A$. Let $\mathds 1=1_K$ be the constant
function, then $\mathds 1$ is an order unit. Put $V=A^*$, $V^+=(A^+)^*$, then $V^+$ is a proper cone in $V$  and 
$K$ is affinely isomorphic to the base of $V^+$, determined by $\mathds 1$.

\subsection{Base norms and order unit norms}

In GPTs, we have natural norms induced by the cones in the state space and the space of effects. They make the space $V$ into a base norm space and the space $A$ into an order unit space. For more details about the following, see the excellent \cite[Chapter 1.6]{lami2018non}.

\begin{defi}
	Given a GPT $(V, V^+, \mathds 1)$, define the following norm on $V$, called a \emph{base norm}
	$$\|x\|_V = \inf\{\mathds 1(y) + \mathds 1(z) \, : \, y,z \in V^+ \text{ s.t.~ } x = y - z\}$$
	as well as a norm on $A := V^*$, called an \emph{order unit norm}
	$$\|\alpha\|_A = \inf \{t \geq 0\, : \, \alpha \in t [-\mathds 1, \mathds 1]\}.$$
\end{defi}
Using the above definition, one can characterize positivity in $V$ using metric properties: 
$$x \in V^+ \iff \|x\|_V = \mathds 1(x).$$
The base norm and the order unit norm are dual to each other, i.e.\
\begin{equation*}
    \norm{x}_V = \sup_{\norm{\alpha}_A \leq 1} |\langle \alpha, x \rangle|.
\end{equation*}

\begin{ex}
For the classical GPT $\mathrm{CM}_d$, the base norm and the order unit norm correspond to the $\ell_1$ and $\ell_\infty$ norms, respectively: 
\begin{align*}
\|x\|_V &= \|x\|_1 = \sum_{i=1}^d |x_i| \quad \text{ and } \quad  \|\alpha\|_A = \|\alpha\|_\infty = \max_{i\in [d]} |\alpha_i|.
\end{align*}
For quantum mechanics, we obtain the Schatten $1$ and $\infty$ norms, respectively: 
\begin{align*}
\|x\|_V &= \|x\|_1 = \operatorname{Tr} \sqrt{x^2}\\
\|\alpha\|_A &= \|\alpha\|_\infty,
\end{align*}
where the $\|\cdot\|_\infty$ norm is the usual operator norm. 
\end{ex}

\subsection{Centrally symmetric GPTs} \label{sec:centrally-symmetric-GPTs}
For some GPTs such as the hypercube $\mathrm{HC}_n$, we have more structure we can use. We review in this section what it means for a GPT $(V,V^+,\mathds 1)$ to be \emph{centrally symmetric} in the sense of \cite[Definition 25]{lami2018ultimate}. In this case, the vector space $V$ admits a decomposition $V = \mathbb R v_0 \oplus \bar V$, and we shall write $x = (x_0, \bar x)$ for a vector $x = x_0v_0 \oplus \bar x$. We then have
$$x = (x_0, \bar x) \in V^+ \iff \|\bar x\|_{\bar V} \leq x_0,$$
where $\|\cdot\|_{\bar V}$ denotes the norm on $\bar V$. The decomposition $V = \mathbb R v_0 \oplus \bar V$ is such that the form $\mathds 1$ is given by $\mathds 1(x) = x_0$. The base norm is given as $\norm{x}_V = \max(|x_0|,\|\bar x\|_{\bar V})$. The dual ordered vector space $(A,A^+) = (V^*,(V^+)^\ast)$ admits a similar description in terms of the dual norm $\|\cdot\|_{\bar A}$ in the dual space $\bar A=\bar V^\ast$: $A = \mathbb R \mathds 1 \oplus \bar A$ and $a = (a_0, \bar a) \in A^+ \iff \|\bar a\|_{\bar A} \leq a_0$. Note that for centrally symmetric spaces, both $A$ and $V$ possess distinguished order units.

\begin{lem}\label{lem:order-unit-norm-cs-GPT}
Let $(V,V^+,\mathds 1)$ be a centrally symmetric GPT. The order unit norm on $A$ satisfies
\begin{equation*}
\norm{\phi}_{A} = |\phi_0| + \|(\phi_1, \ldots, \phi_g)\|_{\bar A}
\end{equation*}
for all $\phi \in A$ with coordinates $\phi = \phi_0 \mathds 1 + \sum_{i = 1}^g \phi_i \alpha_i$, where $\phi_0$, $\phi_i \in \mathbb R$ for all $i \in [g]$ and $\{\mathds 1, \alpha_i\}_{i \in [g]}$ is a basis of $A$.
\end{lem}
\begin{proof}
By definition of the order unit norm of a vector $\phi \in A$, $\|\phi\|_A$ is the minimal $t \geq 0$ such that $\phi \in t[-\mathds 1, \mathds 1]$. 
\begin{align*}
\phi \in t[-\mathds 1, \mathds 1] & \iff t \mathds 1 - \phi \in A^+ \land t \mathds 1 + \phi \in A^+ \\
& \iff (t -\phi_0) \mathds 1 - \sum_{i = 1}^g \phi_i \alpha_i \in A^+ \land (t +\phi_0) \mathds 1 + \sum_{i = 1}^g \phi_i \alpha_i\in A^+ \\
 & \iff (t -\phi_0) \geq \|(\phi_1, \ldots, \phi_g)\|_{\bar A} \land (t +\phi_0) \geq \|(\phi_1, \ldots, \phi_g)\|_{\bar A} \\
&\iff t \geq |\phi_0| + \|(\phi_1, \ldots, \phi_g)\|_{\bar A}.
\end{align*}
This proves the assertion.
\end{proof}

\begin{ex} \label{ex:hypercube}
The hypercube GPT $\mathrm{HC}_n$  is centrally symmetric, with the base norm and the order unit norm given by
\begin{align*}
\|x\|_V &= \max(|x_0|, \|(x_1, \ldots, x_n)\|_\infty) = \|x\|_\infty = \max_{0 \leq i \leq n} |x_i| \\
\|\alpha\|_A &= |\alpha_0| + \|(\alpha_1, \ldots, \alpha_n)\|_1 = \|\alpha\|_1 =  \sum_{i=0}^n |\alpha_i|.
\end{align*}
Such GPTs can be used to model theories containing PR boxes \cite{Popescu1994}.
\end{ex}

\begin{ex}\label{ex:qubits}
Qubits (i.e.~quantum mechanics for two-level systems) form a centrally symmetric GPT, with the unit ball of $\mathbb R^3$ as a state space: the \emph{Bloch ball}. One can decompose any qubit state $\rho$ as 
$$\rho = \frac 1 2 \left( I_2 + r \cdot \sigma \right),$$
where $I_2$ is the identity matrix and $\sigma = (\sigma_X, \sigma_Y, \sigma_Z)$ is the vector of Pauli matrices and $r \in \mathbb R^3$ is a vector of norm at most one.
\end{ex}

\subsection{Measurements and compatibility}

We have already introduced states and norms in our GPT. In this section, we show how the elements of $A$ can represent  measurements in this theory and we study their compatibility. We work under the no-restriction hypothesis \cite{Janotta2013}, which states that all affine functions $K \to [0,1]$ correspond to physical effects.

Let $(V,V^+,\mathds 1)$ be a GPT.
Any  measurement of the system is determined by assigning the corresponding outcome probabilities to any state. We will
only consider measurements with a finite number of outcomes, usually  labeled by the set $[k]$, $k \in \mathbb N$.

There are several equivalent ways to describe a measurement. First of all, to each outcome $i\in [k]$ and $\rho\in K$,
let $f_i(\rho)$ denote the probability of 
obtaining the outcome $i$ if the system is in the state $\rho$. Then, $f_i$ is an affine map 
 $K\to [0,1]$, $i\in [k]$, and corresponds to  some  $f_i\in A^+$, $f_i\le \mathds 1$. Such 
elements are called \emph{effects}. Thus, the measurement is described by a tuple of effects $f=(f_1,\dots,f_k)$ which must satisfy
  $\sum_{i = 1}^k f_i=\mathds 1$. The measurement can be also described as an
affine map into the probability simplex $f: K\to \Delta_k$,  given by 
\[
\rho \mapsto \sum_{i = 1}^k f_i(\rho)e_i.
\] 
Here, $e_i$ are the vertices (or extremal points) of the simplex. We will use the same notation also for the unique extension of this map to $(V,V^+)\to (\mathbb R^k,\mathbb R^k_+)$ (where the vertices correspond to elements of the standard basis). There is also an associated positive unital map $\Phi_f: (\mathbb R^{k},\mathbb R^{k}_+)\to (A,A^+)$, determined by
\[
\Phi_f(e_i)=f_i,
\]
which is clearly the adjoint map of  $f$. A measurement $f=(f_1, \ldots, f_k)$ in which the effects are multiples of the order unit $f_i = p_i \mathds 1$ is called \emph{trivial}; the vector $(p_1, \ldots, p_k) \in \mathbb R^k$ is then a probability vector.

We turn now to the analysis of the notion of compatibility of several measurements. We shall consider sets of $g$ measurements, having possibly different numbers of outcomes $k_1, \ldots, k_g$. We shall write $\mathbf k := (k_1, \ldots, k_g) \in \mathbb N^g$.

\begin{defi}\label{def:compatibility}
Let $f^{(1)},\dots, f^{(g)}$ be measurements, $f^{(i)}=(f^{(i)}_1,\dots, f^{(i)}_{k_i})$, $i=1,\dots,g$. We  will say
that the collection $(f^{(1)},\dots, f^{(g)})$ is \emph{ compatible} if all $f^{(i)}$ can be obtained as marginals of a
single \emph{joint measurement}. More precisely, there is a measurement $h$ with outcomes labeled by
$[k_1]\times\dots\times[k_g]$ such that 
\[
f^{(i)}_j=\sum_{m_1,\dots,m_{i-1},m_{i+1},\dots, m_g} h_{m_1,\dots,m_{i-1},j,m_{i+1},\dots,\dots,m_g},\qquad j\in[k_i],\
i\in [g].
\]
\end{defi}

The \emph{noise robustness of incompatibility} with respect to \emph{white noise} can be described as the amount of white noise that has to be mixed with the measurements to make the collection $f=(f^{(1)},\dots, f^{(g)})$ compatible, see \cite{Heinosaari2015}. 
This leads to the following definition of the compatibility region for $f$:
\[
\Gamma(f):=\{s \in [0,1]^g \, : s_i f^{(i)} + (1-s_i)\mathds{1}/k_i \text{ are compatible measurements}\}.
\]
By considering these regions for all collections of measurements with $\mathbf{k}$ outcomes, we obtain a characterization of the amount of incompatibility available in the given GPT.  

\begin{defi}\label{def:compreg}
Given a GPT $(V,V^+,\mathds{1})$, $g \in \mathbb N$ and $\mathbf{k} \in \mathbb N^g$, we define the \emph{compatibility region for $\mathbf{k}$ outcomes} as 
\begin{align*}\Gamma(\mathbf k; V,V^+):=&\{s \in [0,1]^g \, : s_i f^{(i)} + (1-s_i)\mathds{1}/k_i \text{ are compatible measurements for all} \\ & \text{ collections } f^{(i)} \in A^{k_i}, i \in [g], \text{ of $g$ measurements with $\mathbf k$ outcomes} \}\\
=& \bigcap_{f\text{ with }\mathbf{k}\text{ outcomes }} \Gamma(f).
\end{align*}
If $\mathbf k = 2^{\times g}$, we will just write $\Gamma(g; V,V^+)$. 
\end{defi}

 It is easy to see that $\Gamma(\mathbf k; V, V^+)$ is convex and that $\Gamma(\mathbf k; V, V^+) = [0,1]^g$ if and only if all $\mathbf k$ outcome measurements are compatible. If $k_i \geq 2$ for all $i \in [g]$, then \cite{Plavala2016} shows that this is the case if and only if the state space is isomorphic to a simplex. 
 
 The compatibility region is always non-empty as the next proposition shows. The intuitive argument is that we can add as many trivial measurements as we want without affecting the compatibility of a set of measurements and that every measurement is compatible with itself. 
 \begin{prop}
 Let $(V,V^+,\mathds{1})$ be a GPT. Then, for any $g \in \mathbb N$ and $\mathbf{k} \in \mathbb N^g$, it holds that
 \begin{equation*}
     \left\{s\in [0,1]^g: \sum_{i = 1}^g s_i \leq 1\right\} \subseteq \Gamma(\mathbf k; V, V^+).
 \end{equation*}
 \end{prop}
 \begin{proof}
 We need to prove that $e_i \in \Gamma(\mathbf k; V, V^+)$ for all $i \in [g]$, where $\{e_i\}_{i =1}^g$ is the standard basis in $\mathbb R^g$. The statement then follows by convexity. Without loss of generality, let us consider $e_1$. Let $f \in A^{k_1}$ be a measurement. Then, $h$ with effects
 \begin{equation*}
     h_{i_1, \ldots, i_g} = \frac{1}{k_2 \cdots k_g} f_{i_1} \qquad \forall i_j \in [k_i],~ j \in [g]
 \end{equation*}
 is a joint measurement for $f$ and the trivial measurements with effects $\mathds 1/k_i$ for all $i \in [g]\setminus \{1\}$. This proves the assertion. 
 \end{proof}

The elements in the compatibility region with all coefficients equal corresponds to the notion of a compatibility degree for $f$:
\[
\gamma(f):=\max\{s\in [0,1], (s,\dots,s)\in \Gamma(f)\}.
\]

\begin{defi}\label{def:compdeg} Let $(V,V^+,\mathds 1)$ be a GPT, $g\in \mathbb N$, $\kk\in \mathbb N^g$. We define
the \emph{compatibility degree} for $\kk$ outcomes as
\begin{align*}
\cdeg(\kk; V,V^+)&:=\max\{s\in [0,1]:\ s f^{(i)} + (1-s)\mathds{1}/k_i \text{ are compatible measurements for all} \\ & \qquad\qquad\qquad \text{ collections } f^{(i)} \in A^{k_i}, i \in [g], \text{ of $g$ measurements with $\mathbf k$ outcomes} \}\\
&=\min_{f\text{ with }\mathbf{k}\text{ outcomes }} \gamma(f).
\end{align*}
If $\mathbf k = 2^{\times g}$, we will just write $\cdeg(g; V,V^+)$. 
\end{defi}

For dichotomic measurements, the compatibility degree decreases with the number of measurements. The intuitive argument is that subsets of compatible measurements remain compatible. 
\begin{prop} \label{prop:gamma-decreases}
Let $(V,V^+,\mathds 1)$ be a GPT and let $g$, $g^\prime \in \mathbb N$ be such that $g \leq g^\prime$. Let $\mathbf k^\prime \in \mathbb N^{g^\prime}$ and $\mathbf k = (k_1^\prime, \ldots, k_g^\prime)$. Then,
\begin{equation*}
    \gamma(\mathbf k; V, V^+) \geq  \gamma(\mathbf k^\prime; V, V^+).
\end{equation*}
\end{prop}
\begin{proof}
Let $f^{(i)} \in A^{k_i^\prime}$, $i \in [g^\prime]$, be a collection of measurements. Let $s = \gamma(\mathbf k^\prime; V, V^+)$ and let $\tilde f^{(i)}$ be the corresponding noisy measurements defined by
\begin{equation*}
    \tilde f^{(i)}_j = s f^{(i)}_j + (1-s) \frac{\mathds{1}}{k_i^\prime} \qquad \forall j \in [k_i^\prime],~i \in [g^\prime].
\end{equation*}
Then, the $\tilde f^{(i)}$, $i \in [g^\prime]$,  are compatible with joint measurement $h^\prime$. Setting \begin{equation*}
    h_{i_1, \ldots, i_g} = \sum_{i_j \in [k_j^\prime], j \in [g^\prime] \setminus [g]} h^\prime_{i_1, \ldots, i_{g^\prime}}
\end{equation*}
yields a joint measurement for $\bar f^{(i)}$, $i \in [g]$. Since the measurements were arbitrary, it follows that $s \leq \gamma(\mathbf k; V, V^+)$.
\end{proof}

Intuitively, the more outcomes we have, the more difficult it becomes  for all measurements to be compatible. This is the content of the next proposition, which is the GPT version of \cite[Proposition 3.35]{bluhm2020compatibility}. The proof is very similar. 

\begin{prop}
    Let $g \in \mathbb N$, $\kk^\prime$, $\kk \in \mathbb N^g$ and let $(V, V^+, \mathds 1)$ be a GPT. Let $\kk^\prime \geq \kk$, where the inequality is meant to hold entrywise. Then
    \begin{equation*}
        \Gamma(\mathbf k^\prime; V, V^+) \subseteq  \Gamma(\mathbf k; V, V^+).
    \end{equation*}
    In particular, $\gamma(\mathbf k^\prime; V, V^+) \leq  \gamma(\mathbf k; V, V^+)$.
\end{prop}
\begin{proof}
Let $(f_1^{(i)}, \ldots, f_{k_i}^{(i)}) \in A^{k_i}$ be measurements for all $i \in [g]$. Let $s \in \Gamma(\mathbf k^\prime; V, V^+)$. Then, the measurements
\begin{equation*}
    F^{(i)} := \Bigg(s_if_1^{(i)}+(1-s_i)\frac{\mathds 1}{k_i^\prime}, \ldots, s_i f_{k_i}^{(i)}+(1-s_i)\frac{\mathds 1}{k_i^\prime}, \underbrace{(1-s_i)\frac{\mathds 1}{k_i^\prime}, \ldots, (1-s_i)\frac{\mathds 1}{k_i^\prime}}_{k_i^\prime - k_i}\Bigg)\in A^{k'_i}
\end{equation*}
are compatible. We will show that this is still true if we replace, for some fixed $l \in [g]$, the $l$-th measurement $F^{(l)}$ by 
\begin{equation*}
     G^{(l)} := \left(s_l f_1^{(l)}+(1-s_l)\frac{\mathds 1}{k_l}, \ldots, s_l f_{k_l}^{(l)}+(1-s_l)\frac{\mathds 1}{k_l}\right)\in A^{k_l}.
\end{equation*}
An iterative application of this procedure then shows that $s \in \Gamma(\mathbf k; V, V^+)$ and the assertion follows.

Let $H_{i_1, \ldots, i_g}$ be the effects of the joint measurement for the $F^{(i)}$, where $i_j \in [k_j^\prime]$, $j \in [g]$. Let us define
\begin{equation*}
    R_{i_1, \ldots, i_{l-1}, k_l + 1, i_{l+1}, \ldots, i_g} := \sum_{i_l = k_l + 1}^{k_l^\prime}H_{i_1, \ldots, i_g}.
\end{equation*}
Then, let us define
\begin{equation*}
    h_{i_1, \ldots, i_g} = H_{i_1, \ldots, i_g} + \frac{1}{k_l} R_{i_1, \ldots, i_{l-1}, k_l + 1, i_{l+1}, \ldots, i_g} \qquad i_l \in [k_l],~i_j \in [k_j^\prime]~\forall j \in [g] \setminus \{l\}.
\end{equation*}
It is easy to verify that $h$ is a measurement. Moreover, for $m \in [g]$, $m \neq l$, $q \in [k_m^\prime]$,
\begin{equation*}
    \sum_{\substack{i_j \in [k_j^\prime], j \in [g]\setminus \{m,l\},\\ i_l \in [k_l], i_m = q}} h_{i_1, \ldots, i_g}= F^{(m)}_q
\end{equation*}
and for $p \in [k_l]$,
\begin{align*}
    \sum_{i_j \in [k_j^\prime], j \in [g]\setminus \{l\}, i_l = p} h_{i_1, \ldots, i_g}&= \sum_{i_j \in [k_j^\prime], j \in [g]\setminus \{l\}, i_l = p} H_{i_1, \ldots, i_g} + \frac{1}{k_l} \sum_{i_j \in [k_j^\prime], j \in [g]\setminus \{l\}}  R_{i_1, \ldots, i_{l-1}, k_l + 1, i_{l+1}, \ldots, i_g} \\
    &= F_p^{(l)} + \frac{1}{k_l}\sum_{j = k_l +1}^{k_l^\prime}F_j^{(l)} \\
    & = s_l f_p^{(l)} + \frac{(1-s_l)}{k_l^\prime} \mathds 1 + \frac{(1-s_l)(k_l^\prime - k_l)}{k_l k_l^\prime} \mathds 1 \\
    &= s_l f_p^{(l)} +\frac{(1-s_l)}{k_l} \mathds 1 =G^{(l)}_p.
\end{align*}
Thus, $h$ is the desired joint measurement for the $F^{(i)}$, $i \in [g] \setminus \{l\}$ and $G^{(l)}$.
\end{proof}

We have seen in the last proposition how the compatibility regions are related for measurements with a different number of outcomes within the same theory. We will now show that sometimes the compatibility regions of different GPTs can be related.

\begin{prop}
Let $g \in \mathbb N$, $\mathbf k \in \mathbb N^g$. Let $(V_1, V_1^+, \mathds 1_{A_1})$ and $(V_2, V_2^+, \mathds 1_{A_2})$ be two GPTs. Let $\Psi:(V_1, V_1^+) \to (V_2,V_2^+) $ be a positive map such that $\Psi^\ast(\mathds 1_{A_2}) = \mathds 1_{A_1}$. If $\Psi$ is a \emph{retraction}, i.e.~if there is a positive map $\Theta: (V_2, V_2^+) \to (V_1,V_1^+)$ such that $\Psi \circ \Theta = \mathrm{id}_{V_{2}}$, then the GPT $(V_1, V_1^+, \mathds 1_{A_1})$ is ``less compatible'' than $(V_2, V_2^+, \mathds 1_{A_2})$ in the measurement setting $\mathbf k$: $\Gamma(\mathbf k; V_1, V_1^+) \subseteq \Gamma(\mathbf k; V_2, V_2^+)$.
\end{prop}
\begin{proof}
Let $s \in \Gamma(\mathbf k; V_1, V_1^+)$. Let $f^{(i)} \in A_2^{k_i}$, $i \in [g]$, be a collection of measurements. Then, the $\Psi^\ast(f_j^{(i)}) \in A_1^+$, $j \in [k_i]$, form a collection of measurements as well. Moreover, the measurements given by
\begin{equation*}
         h_j^{(i)} = s_i \Psi^\ast(f_j^{(i)}) + (1-s_i) \frac{\mathds 1_{A_1}}{k_i} \qquad \forall i \in [g], ~j \in [k_i]
\end{equation*}
are compatible by the assumption on $s$. Since, $\Psi$ is a retraction, it holds that $\Theta^\ast \circ \Psi^\ast = \mathrm{id}_{A_2}$. Since $\Psi^\ast(\mathds 1_{A_2}) = \mathds 1_{A_1}$, also $\Theta^\ast(\mathds 1_{A_1}) = \mathds 1_{A_2}$. Thus, the image of the joint measurement for the $h^{(i)}$ under $\Theta^\ast$ is again measurement in $(V_2, V_2^+, \mathds 1_{A_2})$. Furthermore, it is a joint measurement for the noisy measurements 
\begin{equation*}
     s_i f_j^{(i)} + (1-s_i) \frac{\mathds 1_{A_2}}{k_i} \qquad \forall i \in [g], ~j \in [k_i],
\end{equation*}
since $ \Theta^\ast(\Psi^\ast(f_j^{(i)})) = f_j^{(i)} $ for all $ i \in [g]$, $j \in [k_i]$.
\end{proof}

\begin{remark}
An example for this situation is $(V_1, V_1^+, \mathds 1_{A_1}) = \mathrm{QM_d}$ for some $d \in \mathbb N$ and $(V_2, V_2^+, \mathds 1_{A_2}) =\mathrm{CM}_d $, where $\Psi^\ast$ embeds the probability distributions in $\mathrm{CM}_d$ as diagonal matrices. Then, $\Theta^\ast$ projects onto the diagonal entries of the matrix. 
\end{remark}

\section{Compatible  measurements and positive maps}\label{sec:maps}

In this section, we prove that the compatibility of a collection of measurements can be characterized by the properties of a certain associated positive map. 

Let $f=(f^{(1)},\dots, f^{(g)})$ be a collection of  measurements for a GPT  $(V,V^+,\mathds 1)$. Assume that for $i\in [g]$, $g \in \mathbb N$,  the measurement $f^{(i)}$ has $k_i$ outcomes and let us denote $\kk:=(k_1,\dots,k_g)\in
\mathbb N^g$. The collection of $g$ affine maps $f^{(i)}: K\to \Delta_{k_i}$ can be naturally identified with an affine  map of the state space into the \emph{polysimplex} $P_{\kk}:= \Delta_{k_1}\times\dots \times \Delta_{k_g}$,  defined as 
\[
\rho\mapsto (f^{(1)}(\rho), \dots,f^{(g)}(\rho))\in P_\kk.
\]
Conversely, we can obtain a collection of $g$ measurements from any such map by composition with the projections
$m_i: P_\kk \to \Delta_{k_i}$ to the $i$-th component of $P_\kk$, for $i \in [g]$.

We next find a suitable representation of $P_\kk$ as a base of the positive cone in an ordered vector space. As we have seen 
in Section \ref{sec:gpt}, we can find such a representation by considering affine functions $P_\kk\to \mathbb R$.
Since $P_\kk$ is a convex polytope, any affine function on $P_\kk$ is determined by its value on each of the $k_1 \cdots k_g$
vertices of $P_\kk$.
Therefore,  $A(P_\kk)$ is isomorphic to a subspace 
\[
E_{\mathbf k}\subseteq \mathbb R^\kk:=\mathbb R^{k_1\cdots k_g}\cong \bigotimes_{i\in [g]} \mathbb R^{k_i}
\]
and it is clear that the positive cone $A(P_\kk)^+\cong E_\kk^+:=E_\kk\cap \mathbb R^{\kk}_+$.
The unit $1_{P_\kk}$ corresponds to the vector $ 1_\kk=1_{k_1}\otimes\dots\otimes 1_{k_g}$ of all ones, which is
clearly contained in $E_\kk\cap \operatorname{int}(\mathbb R^{\kk}_+)$.

 Let $J_\kk$ denote the orthogonal projection $\mathbb R^{\kk}\to
E_\kk$. By Proposition \ref{prop:E_dual}, we  may put $E_\kk^*\cong E_\kk$ and $(E_\kk^+)^*\cong J_\kk(\mathbb
R^{\kk}_+)$. Then $P_\kk$ is represented as the base of  $(E_\kk^+)^*$ determined by the unit functional $1_\kk$,
that is
\[
P_\kk \cong \{ x\in J_\kk(\mathbb R^{\kk}_+),\  \langle 1_\kk, x\rangle=1\}=\{J_\kk y,\  y\in \mathbb R^{\kk}_+,\
\langle1_\kk,y\rangle=1\}=J_\kk(\Delta_{\kk}), 
\]
where 
\[
\Delta_\kk:=\Delta_{k_1\dots k_g}\cong \otimes_i \Delta_{k_i}.
\]
The vertex $(e^{(k_1)}_{i_1},\dots, e^{(k_g)}_{i_g})$ of $P_\kk$ corresponds precisely to 
$J_\kk(e^{(k_1)}_{i_1}\otimes\dots\otimes e^{(k_g)}_{i_g})$, where $e^{(k)}_j$ denote  elements of the standard basis in
$\mathbb R^k$. 

We have therefore constructed a linear and order isomorphism between the order unit spaces 
$(A(P_\kk),A(P_\kk)^+,1_{P_\kk})$ and $(E_\kk, E_\kk^+, 1_\kk)$, respectively the base norm spaces $(V(P_\kk),V(P_\kk)^+)$ and $(E_\kk^*, (E_\kk^+)^* = J_\kk(\mathbb R_+^{\kk}))$.
We will use this identification throughout this section. 
The following result is easily checked.

\begin{lem}\label{lem:proj_marginals}
The restriction 
\[
J_\kk|_{\Delta_\kk}:\Delta_\kk\to J_\kk(\Delta_{\kk})\cong P_\kk
\]
defines a collection $(j^{(1)},\dots,j^{(g)})$ of  measurements on $\Delta_{\kk}$, where 
$j^{(i)}(p)$ is the $i$-th marginal of $p$, for any $p\in \Delta_\kk\cong \otimes_i \Delta_{k_i}$. In other words, for a probability distribution $p \in \Delta_\kk$, we have $J_\kk(p) = (p^{(1)}, \ldots, p^{(k)}) \in P_\kk$, where $p^{(i)} = j^{(i)}(p)$ is the $i$-th marginal of $p$.
 
\end{lem}

To give  a more explicit description of $E_\kk$, observe first that each projection 
$m^i: P_\kk\to \Delta_{k_i}$ onto the $i$-th component is a measurement on $P_\kk$ and the corresponding effects are given by 
the elements of $E_\kk^+$ of the form 
\[
\eta^{(i)}_j:=1_{k_1}\otimes\dots \otimes 1_{k_{i-1}}\otimes e^{(k_i)}_j\otimes 1_{k_{i+1}}\otimes\dots\otimes 1_{k_g},\qquad j\in
[k_i],\ i\in [g].
\] 
By the results of \cite{jencova2018incompatible}, these elements  generate the extremal rays of  $E_\kk^+$ and consequently $E_\kk$ 
is spanned by the set  $\{\eta^{(i)}_j\}$.  We choose a basis of $E_\kk$ given as
\[
w:=\{1_\kk, w^{(i)}_j, \ j\in [k_i-1],\ i\in [g]\},
\]
where 
\begin{equation*}
w^{(i)}_j := \underbrace{1_{k_1} \otimes \cdots \otimes 1_{k_{i-1}}}_{i-1 \text{ factors}} \otimes
v^{(k_i)}_j \otimes \underbrace{1_{k_{i+1}} \otimes \cdots \otimes 1_{k_g}}_{g-i \text{ times}} \in
\mathbb R^{\kk}
\end{equation*}
and  $v^{(k)}_1, \ldots, v^{(k)}_{k-1}$ are vectors  in $\mathbb R^{k}$ defined as
\begin{equation*}
v^{(k)}_j(l) := -\frac{2}{k} + 2\delta_{l, j}, \qquad \forall j \in [k-1],\, \forall l \in [k].
\end{equation*}
The set $\{1_k, v^{(k)}_j, j\in [k-1]\}$ is a basis of $\mathbb R^k$ and  the dual basis 
of $\mathbb R^k\cong (\mathbb R^k)^*$ is
$\{\frac1k 1_k, v^{(*k)}_j,  j\in [k-1]\}$, where
\[
v^{(*k)}_j:= \frac12(e^{(k)}_j-e^{(k)}_k),\qquad  j\in [k-1].
\]
It is then clear that 
\[
w^*:=\{\frac 1{k_1\dots k_g} 1_\kk, w^{(*i)}_j,\  \ j\in [k_i-1],\ i\in [g]\},
\]
where
\[
w^{(*i)}_j:=\frac{k_i}{k_1\dots k_g} 1_{k_1}\otimes\dots\otimes 1_{k_{i-1}}\otimes v^{(*k_i)}_j\otimes
1_{k_{i+1}}\otimes\dots\otimes 1_{k_g},
\]
is the basis of $E_\kk^*\cong E_\kk$ which is  dual to $w$. This choice of basis may seem strange at first, but it will be convenient in Section \ref{sec:gen-spectra-and-comp}.

Let $\Phi: E_\kk\to A$ be a linear map, then $\Phi$ is determined by elements $p_0$, $p^{(i)}_j\in A$, such that
\[
\Phi(1_\kk)=p_0,\quad \Phi(w^{(i)}_j)= p^{(i)}_j,\qquad j\in [k_i-1],\ i\in[g].
\]
Then the corresponding element $\phi^\Phi\in E^*_\kk\otimes A$ (see Section \ref{sec:positive_maps}) is
\[
\phi^\Phi = \frac 1{k_1\dots k_g} 1_\kk\otimes p_0+\sum_{i=1}^g\sum_{j=1}^{k_i-1} w^{(i*)}_j\otimes p^{(i)}_j.
\]

\begin{prop}\label{prop:pu_map}
Let $\Phi$ and $p_0$, $p^{(i)}_j\in A$ be as above.  Then $\Phi: (E_\kk,E_\kk^+)\to (A,A^+)$ is positive and unital if and
only if  $p_0=\mathds 1$ and there is a collection $(f^{(1)},\dots, f^{(g)})$ of measurements, $f^{(i)}: K\to
\Delta_{k_i}$,  such that 
\[
p^{(i)}_j = 2f^{(i)}_j-\frac2{k_i} \mathds 1,\qquad j\in [k_i-1],\ i\in [g].
\]
The adjoint map $\Phi^* : (V,V^+) \to (E_\kk^*, (E_\kk^+)^*)$ satisfies 
\begin{equation}\label{eq:adjoint}
\Phi^*(\rho)=(f^{(1)}(\rho),\dots, f^{(g)}(\rho)),\qquad \rho\in K.
\end{equation}

\end{prop}

\begin{proof}
It is clear that $\Phi$ is unital if and only if $p_0=\mathds 1$. Moreover, $\Phi$ is positive if and only if
$\varphi^\Phi\in (E_\kk^+)^*\otimesmax A^+$.
Since the elements $\eta^{(i)}_j$ generate extremal rays of $E^+_\kk$, it is enough to verify that
\[
\langle \phi^\Phi, \eta^{(i)}_j\otimes v\rangle \ge 0,\qquad \forall v\in V^+,\  j\in [k_i],\ i\in [g].
\]
From this condition and 
\[
\langle w^{(*i')}_{j'},\eta^{(i)}_j\rangle = \frac 12\delta_{i,i'}(\delta_{j,j'}-\delta_{j,k_i}),
\]
we get that for $i\in [g]$,
\[
\frac1{k_i}\mathds 1+ \frac12 p^{(i)}_j\in A^+,\quad j\in [k_i-1],\qquad  \frac 1{k_i} \mathds 1-\frac12
\sum_{j=1}^{k_i-1} p^{(i)}_j\in A^+.
\]
We now put
\[
f^{(i)}_j:=\frac1{k_i}\mathds 1+\frac12 p^{(i)}_j,\ j\in [k_i-1],\qquad f^{(i)}_{k_i}:=\mathds 1-\sum_{j=1}^{k_i-1} f^{(i)}_j,
\]
then $(f^{(1)},\dots, f^{(g)})$ is a collection of measurements of the required form. 
Since the $\eta^{(i)}_j$ correspond to the effects of the projection map onto the $i$-th component of $P_\kk$, the last statement 
follows from 
\[
\langle \Phi^*(v),\eta^{(i)}_j \rangle = \langle \phi^\Phi, \eta^{(i)}_j\otimes v\rangle= \langle  f^{(i)}_j,v\rangle,\qquad j\in [k_i],\ i\in [g].
\]
\end{proof}

Let $f=(f^{(1)},\dots,f^{(g)})$ be a collection of measurements with outcome spaces specified by the vector
$\kk=(k_1,\dots,k_g)$. 
By the above result, we clearly have that there is an associated positive unital map $\Phi^{(f)}:E_\kk\to A$, defined as
\begin{align}
\Phi^{(f)}:E_{\mathbf k} &\to A \nonumber\\
1_{\kk} &\mapsto \mathds{1} \label{eq:def-Phi-f}\\
w_j^{(i)} &\mapsto 2f_j^{(i)} - \frac{2}{k_i}\mathds{1}\qquad \forall j \in [k_i-1],~ i \in [g].\nonumber
\end{align}

The corresponding element in
$(E_\kk^+)^*\otimesmax A^+$ is given by
\begin{equation}\label{eqq:def_phi_f}
\varphi^{(f)}=\frac{1}{k_1\dots k_g}1_\kk \otimes \mathds 1  + \sum_{i\in [g]}\sum_{j\in [k_i-1]} w^{(*i)}_j\otimes 
(2f_j^{(i)} - \frac{2}{k_i}\mathds{1}).
\end{equation}

Note that the action of the map $\Phi^{(f)}$ can alternatively be defined as
\begin{align*}
    \Phi^{(f)}: E_{\mathbf k} &\to A \\
    \eta_j^{(i)} &\mapsto f_j^{(i)} \qquad \forall j \in [k_i],~i \in [g]
\end{align*}
at the cost that this is only well-defined for measurements since the $\eta_j^{(i)}$ are not all linearly independent.

We now prove the main result of this section; see \cite[Theorem 1 and Proposition A1]{jencova2018incompatible} for a result similar to the equivalence of (2) - (5).

\begin{thm}\label{thm:etb_ext}
Let $f=(f^{(1)},\dots, f^{(g)})$, $\Phi^{(f)}$ and 
$\varphi^{(f)}$ be as above. The following are equivalent.
\begin{enumerate}
\item There is a positive map $\tilde \Phi : (\mathbb R^{\kk},\mathbb R^{\kk}_+)\to (A,A^+)$
extending $\Phi^{(f)}$.
\item $\Phi^{(f)}$ is entanglement breaking.
\item $\varphi^{(f)}\in (J_\kk\otimes \mathrm{id})(\mathbb R^{\kk}_+\otimes A^+)$.

\item There is some measurement $h: K\to \Delta_\kk$ such that $(\Phi^{(f)})^*|_K= J_\kk\circ h$.
\item The collection $f$ is compatible.

\end{enumerate}

\end{thm}

\begin{proof} The statements (1) - (3) are equivalent by Proposition \ref{prop:positivity-extension}.
 Moreover, we have (4) $\iff$ (5) by Lemma \ref{lem:proj_marginals} and Equation~\eqref{eq:adjoint}. 

Assume (1), then since $\Phi^{(f)}$ is unital and $1_\kk\in E_\kk$, $\tilde \Phi$ is unital as well. Hence there is
some measurement $h:K\to \Delta_{\kk}$ such that $\tilde \Phi^*|_K=h$.  For $\rho\in K$ and $e\in E_\kk$, we have
\[
\langle (\Phi^{(f)})^*(\rho),e\rangle = \langle \rho, \Phi^{(f)}(e)\rangle=\langle \rho, \tilde \Phi(e)\rangle=
\langle h(\rho), e\rangle= \langle J_\kk(h(\rho)), e\rangle,
\]
which proves (4).
Finally, assume (4), then we have
\[
\langle \Phi^{(f)}(e),v\rangle= \langle e, (\Phi^{(f)})^*(v)\rangle=\langle e, J_\kk \circ h(v)\>=
\langle e, h(v)\rangle= \langle \Phi_{(h)}(e),v\rangle,\qquad e\in E_\kk,\ v\in V,
\]
so that $\Phi_{(h)}$ is a positive extension of $\Phi^{(f)}$. Here, we have identified the affine map $h$ on $K$ with its unique extension to a linear map on $V$. This finishes the proof.
\end{proof}

\begin{remark}
We already noticed that map extension can be verified using conic programming (see Section \ref{sec:conic-programming-map} in the Appendix). The same is true for the existence of a joint measurement (see \cite{wolf2009measurements, Plavala2016}).
\end{remark}

\section{Inclusion of generalized spectrahedra and compatibility of measurements in GPTs}
\label{sec:gen-spectra-and-comp}

In this section, we will show that the compatibility of measurements in a GPT can be phrased as an inclusion problem of 
generalized spectrahedra.

\subsection{The GPT jewel} \label{sec:jewel-ordered-vector-spaces}
In this section, we define the universal generalized spectrahedron we will consider to relate the compatibility of measurements in GPTs to inclusion problems of generalized spectrahedra. It plays the same role as the matrix diamond in \cite{bluhm2018joint} and the matrix jewel in \cite{bluhm2020compatibility} and is in fact a generalization of both.

Recall from the previous section that, for a positive integer $g$ and a $g$-tuple of positive integers $\kk \in \mathbb N^g$, we define the vector subspace $E_\kk \subseteq \mathbb R^{k_1 \cdots k_g}$ as the linear span of the basis $w=\{1_\kk, w^{(i)}_j, \, i \in [g], \, j \in [k_i-1]\}$, where $1_\kk$ is the all-ones vector and 
$$w^{(i)}_j = 1_{k_1} \otimes \cdots \otimes 1_{k_{i-1}} \otimes \underbrace{\left( 2 e_j^{(k_i)} - \frac{2}{k_i}1_{k_i} \right)}_{v^{(k_i)}_j} \otimes 1_{k_{i+1}} \otimes \cdots \otimes 1_{k_g}$$
for $i \in [g]$ and $j \in [k_i-1]$. In particular, $\dim E_\kk = 1-g+\sum_{i-1}^g k_i$.

\begin{defi}
Let $g \in \mathbb N$, $\mathbf k \in \mathbb N^g$ and let $(L, L^+)$ be a proper ordered vector space. 
The \emph{$(\mathbf k; L,L^+)$-jewel} is defined as
\begin{equation*}
\mathcal D_{\mathrm{GPT}\jewel}(\mathbf k; L,L^+) := \mathcal D_w(L,E_{\mathbf k}^+ \otimes_{\max} L^+).
\end{equation*}
For $k_1=\ldots=k_g=2$, the above has an especially easy form and we write 
\begin{equation*}
\mathcal D_{\mathrm{GPT}\diamond}(g; L,L^+):= \mathcal D_{\mathrm{GPT}\jewel}(2^{\times g};L,L^+)
\end{equation*}
for short. In this case, $w^{(i)}_1 =c_i$ for all $i \in [g]$, with 
\[
c_i = (1,1) \otimes \cdots \otimes (1,1) \otimes
(+1,-1) \otimes (1,1) \otimes \cdots \otimes (1,1) \in \mathbb R^{2^g}.
\]
 We call this object the \emph{$( g;
L,L^+)$-diamond}.
\end{defi}

Since $w$ is a basis of $E_\kk$, we see that the $(\kk, L,L^+)$-jewel fully describes the tensor product
$E_\kk^+\otimesmax L^+$. The next lemma gives a more explicit description. First, fix $\{e_\kappa\}_{\kappa \in
[k_1]\times \ldots \times [k_g]}$, the standard basis of $\mathbb R^{\kk}$. One has the following decompositions 
$$\forall i \in [g], j \in [k_i-1], \qquad w_j^{(i)} = \sum_{\kappa \in [k_1]\times \ldots \times [k_g]}
w_j^{(i)}(\kappa) e_\kappa,$$
 where by the definition of $w^{(i)}_j$,
\begin{equation}\label{eq:coefficients}
w^{(i)}_j(\kappa)=v^{(k_i)}_j(\kappa_i)=-\frac 2{k_i}+ 2\delta_{\kappa_i,j}.
\end{equation}

\begin{lem} \label{lem:elements-in-the-max}
Let $(L,L^+)$ be a proper ordered vector space.
Moreover, let $\mathbf k \in \mathbb N^{g}$, $z_0$, $z_j^{(i)} \in L$, where $i \in [g]$, $j \in [k_i-1]$. 
Then,
\begin{equation} \label{eq:entrywise-condition}
z_0 + \sum_{i = 1}^g \sum_{j = 1}^{k_i-1} w_j^{(i)}(\kappa) z_j^{(i)} \in L^+ \qquad \forall \kappa \in [k_1] \times \ldots \times [k_g]
\end{equation}
if and only if
\begin{equation*}
1_\kk\otimes z_0 + \sum_{i = 1}^g \sum_{j = 1}^{k_i-1} w_j^{(i)} \otimes z^{(i)}_j \in E_{\mathbf k}^+ \otimes_{\mathrm{max}} L^+.
\end{equation*}
In particular, for $\mathbf k = 2^{\times g}$, Equation~\eqref{eq:entrywise-condition} has the form
\begin{equation*}
z_0 + \sum_{i = 1}^g \epsilon_i z_i \in L^+ \qquad \forall \epsilon \in \{\pm 1\}^g.
\end{equation*}

\end{lem}
\begin{proof}
Let $z_0,  z_j^{(i)}\in L$, $i \in [g]$, $j \in [k_i-1]$ and put
\begin{equation*}
y:= 1_\kk \otimes z_0 + \sum_{i = 1}^g \sum_{j = 1}^{k_i-1} w_j^{(i)}\otimes z_j^{(i)} = \sum_{\kappa \in [k_1]\times \ldots \times [k_g]} e_\kappa \otimes (z_0 + \sum_{i = 1}^g \sum_{j = 1}^{k_i-1} w_j^{(i)}(\kappa)z_j^{(i)}).
\end{equation*}
The standard basis $\{e_\kappa\}_{\epsilon \in [k_1]\times \ldots \times [k_g]}$ of $\mathbb R^\kk\cong (\mathbb
R^\kk)^*$ is self-dual and the elements $e_\kappa$ define the extremal rays of the cone $\mathbb R^\kk_+\cong (\mathbb
R^\kk_+)^*$. Therefore, $y \in E_{\mathbf k}^+
\otimesmax L^+$ if and only if for all $\beta \in (L^+)^\ast$ and $\kappa \in [k_1]\times \ldots \times [k_g]$, 
\begin{equation*}
\langle e_\kappa \otimes \beta, y \rangle = \langle \beta, z_0 + \sum_{i = 1}^g \sum_{j = 1}^{k_i-1} w_j^{(i)}(\kappa)z_j^{(i)} \rangle \geq 0
\end{equation*}
by Proposition \ref{prop:rightcone}. Equivalently,  $z_0 + \sum_{i = 1}^g \sum_{j = 1}^{k_i-1} w_j^{(i)}(\kappa)z_j^{(i)} \in (L^+)^{\ast\ast}$ for all $\kappa \in [k_1]\times \ldots \times [k_g]$. Since $L$ is finite dimensional and $L^+$ is closed, we have that $(L^+)^{\ast\ast} = L^+$ by the bipolar theorem.
\end{proof}

\begin{remark}\label{rem:GPT-diamond}
The above lemma yields a more appealing form of the $(g;  L,L^+)$-diamond, namely
\begin{equation*}
\mathcal D_{\mathrm{GPT}\diamond}(g;L,L^+) = \left \{ (z_0, \ldots, z_g) \in L^{g+1}:  \forall \epsilon \in \{\pm 1\}^g, \, z_0 + \sum_{i=1}^g \epsilon_i z_i \in L^+ \right\}.
\end{equation*}
\end{remark}

\begin{remark}\label{rem:jewel-is-max}
Lemma \ref{lem:elements-in-the-max} shows that $\mathcal D_{\mathrm{GPT}\jewel}(\mathbf k; L, L^+)$ is a maximal generalized spectrahedron, in the sense of Definition \ref{def:D_max}. If $\mathbf k = 2^{\times g}$, the corresponding closed convex cone is generated by the unit ball of $\ell_1^g$, as can be seen from Remark \ref{rem:GPT-diamond}.
\end{remark}

\begin{remark}\label{rem:normalized_jewel} It can be seen by Lemma \ref{lem:elements-in-the-max} that for any 
 $z=(z_0,z^{(i)}_j)_{i,j}\in \mathcal D_{\mathrm{GPT}\jewel}(\mathbf k; L,L^+)$ we must have $z_0\in L^+$. Indeed, 
let
\begin{equation} \label{eq:zkappa}
    z_\kappa:= z_0 + \sum_{i = 1}^g \sum_{j = 1}^{k_i-1} w_j^{(i)}(\kappa) z_j^{(i)}  \qquad  \kappa \in [k_1]
\times \ldots \times [k_g].
\end{equation}
Then $z_\kappa\in L^+$ and one can see using Equation \eqref{eq:coefficients} that $z_0$ is the barycenter of $\{z_\kappa\}_\kappa$:
\begin{equation}\label{eq:barycenter}
z_0=\frac{1}{k_1\cdots k_g} \sum_{\kappa\in [k_1]\times\dots\times[k_g]} z_\kappa.
\end{equation}
Note also that we obtain $z^{(i)}_j$ from $z_\kappa$ by
\begin{equation}\label{eq:z_from_kappa}
z^{(i)}_j=\frac12( z_{k_1,\dots,k_{i-1},j,k_{i+1},\dots,k_g}-z_{k_1,\dots,k_g}),\quad i\in [g],\ j\in [k_i-1].
\end{equation}
In the case of the GPT $(V,V^+,\mathds 1)$ (or if we have  fixed an order unit in $(L^+)^*$), the $(\kk; V,V^+)$
jewel can be normalized by assuming that $z_0\in K$ (that is, $\mathds 1(z_0)=1$). Indeed, this follows from Proposition \ref{prop:zero-implies-zero} below. The resulting set is then convex, closed and bounded.
Indeed, boundedness can be seen easily from
\[
\frac 1{k_1\cdots k_g} \sum_\kappa \|z_\kappa\|_V =\frac 1{k_1\cdots k_g} \sum_\kappa \mathds 1(z_\kappa)\le
\mathds 1(z_0)=1,
\] 
so that the norm $\|z_\kappa\|_V$ is bounded. The normalized jewel can be used in all the results below.
\end{remark}

\begin{prop} \label{prop:zero-implies-zero}
Let $(V, V^+, \mathds 1)$ be a GPT and $g \in \mathbb N$, $\mathbf{k} \in \mathbb N^g$. Moreover, let $z=(0,z^{(i)}_j)_{i,j}\in \mathcal D_{\mathrm{GPT}\jewel}(\mathbf k; V,V^+)$. Then, $z_{j}^{(i)} = 0$ for all $j \in [k_i-1]$, $i \in [g]$.
\end{prop}

\begin{proof} Let $z$ be as in the statement and let $z_\kappa$ be as in Equation \eqref{eq:zkappa}. Then since $z_0=0$ and all $z_\kappa\in V^+$, Equation \eqref{eq:barycenter} implies that $z_\kappa=0$ for all $\kappa\in [k_1]\times\dots \times[k_g]$. The statement now follows from Equation \eqref{eq:z_from_kappa}.
\end{proof}

\begin{ex}\label{ex:jewel}
Let us discuss now the particular forms of the GPT diamond in the case of classical, quantum, and hypercubic GPTs. 

In the case of the classical GPT $\mathrm{CM}_d$, a $(g+1)$-tuple $(z_0,z_1, \ldots, z_g)$ of vectors from $\mathbb R^d$ is an element of $\mathcal D_{\mathrm{GPT}\diamond}(g;\mathbb R^d,\mathbb R^d_+)$ if, for all $j \in [d]$, and for all sign choices $\epsilon \in \{\pm 1\}^g$,
$$z_0(j) + \sum_{i=1}^g \epsilon_i z_i(j) \geq 0 \iff z_0(j) \geq \sum_{i=1}^g |z_i(j)| \iff z_0(j) \geq \|z(j)\|_1.$$

For quantum mechanics $\mathrm{QM}_d$, the condition reads, for self-adjoint matrices $z_0,\ldots, z_g \in \mathcal M_d^{\mathrm{sa}}(\mathbb C)$, 
$$z_0 + \sum_{i=1}^g \epsilon_i z_i \in \mathrm{PSD}_d,$$
for all sign vectors $\epsilon \in \{\pm 1\}^g$. From Remark \ref{rem:normalized_jewel}, we know that $z_0$ is positive semidefinite. We claim that 
\begin{equation}\label{eq:QM-diamond}
    \mathcal D_{\mathrm{GPT}\diamond}(g; \mathcal M_d^{\mathrm{sa}}(\mathbb C), \mathrm{PSD}_d) = \{ (z_0, z_0^{1/2} \tilde z z_0^{1/2}) \, : \, z_0 \in \mathrm{PSD}_d \text{ and } \tilde z \in \mathcal D_{\diamond,g}(d)\},
\end{equation}where $\mathcal D_{\diamond,g}(d)$ is the ``standard'', quantum mechanical matrix diamond, used in \cite{bluhm2018joint} (see also \cite{davidson2016dilations}):
$$\mathcal D_{\diamond,g}(d) = \left\{X \in (\mathcal M_d^{\mathrm{sa}}(\mathbb C))^g \, : \, \forall \epsilon \in \{\pm 1\}^g, \, \sum_{i=1}^g \epsilon_i X_i \leq I_d\right\}.$$
To show that Equation~\eqref{eq:QM-diamond} holds, we prove the two inclusions. Let us start by showing ``$\supseteq$''. For $\tilde z \in \mathcal D_{\diamond,g}(d)$, we have 
$$I_d + \sum_{i=1}^g \epsilon_i \tilde z_i \geq 0 \implies z_0^{1/2} (I_d + \sum_{i=1}^g \epsilon_i \tilde z_i) z_0^{1/2} \geq 0 \iff z_0 + \sum_{i=1}^g z_0^{1/2} \tilde z_i z_0^{1/2} \geq 0.$$

For the reverse inclusion, note that, given $(z_0, \ldots, z_g) \in \mathcal D_{\mathrm{GPT}\diamond}(g; \mathcal M_d^{\mathrm{sa}}(\mathbb C), \mathrm{PSD}_d)$, we have, for all $\epsilon$,
$$z_0 \geq \pm \sum_{i=1}^g \epsilon_i z_i,$$
hence the support of the self-adjoint matrix $\sum_{i=1}^g \epsilon_i z_i$ is included in the support of $z_0$. By restricting all the matrices to the support of $z_0$, we have, for all $\epsilon$,
$$ 0 \leq z_0 + \sum_{i=1}^g \epsilon_i z_i = z_0^{1/2} \left( I + \sum_{i=1}^g \epsilon_i \tilde z_i \right) z_0^{1/2} \implies  I + \sum_{i=1}^g \epsilon_i \tilde z_i \geq 0,$$
where $\tilde z_i := z_0^{-1/2} z_i z_0^{-1/2}$. Hence, the usual matrix diamond and the GPT diamond for $\mathrm{QM}_d$ differ only in the choice of the free term $z_0$: in the former case, one fixes $z_0 = I_d$, while in the latter case $z_0$ is free. 

In the case of the hypercubic GPT $\mathrm{HC}_n$, the situation is similar to the classical GPT: 
$$(z_0,\ldots, z_g) \in \mathcal D_{\mathrm{GPT}\diamond}(g; \mathbb R^{n+1}, C_n) \iff \forall j \in [n], \quad z_0(0) \pm z_0(j) \geq \sum_{i=1}^g|z_i(0) \pm z_i(j)|.$$
\end{ex}

\subsection{Spectrahedral inclusion and compatibility}

Having introduced the universal GPT jewel and diamond, we use them to characterize compatibility of measurements in GPTs via generalized spectrahedral inclusion.

\begin{defi}\label{def:sh_f}
Let $(V, V^+, \mathds 1)$ be a GPT. We define the following generalized spectrahedron: for the preordered vector space $(L,L^+)$ and $\mathbf k \in \mathbb N^g$, $f_j^{(i)} \in A$, $i \in [g]$, $j \in [k_i-1]$,
\begin{equation}\label{eq:def-D-effects-GPT}
\mathcal D_{f}(\mathbf k; L, L^+) := \left \{\left(z_0, z_j^{(i)}\right)_{ij} \in L^{1- g+\sum_i k_i} \, : \,  \mathds{1} \otimes z_0 + \sum_{i = 1}^g \sum_{j = 1}^{k_i-1} \left(2f_j^{(i)} - \frac{2}{k_i}\mathds{1}\right) \otimes z_j^{(i)} \in A^+ \otimesmin L^+ \right \}.
\end{equation}
In the case where $\mathbf k = 2^{\times g}$, we will write $\mathcal D_{f}(g; L, L^+)$ for simplicity.
\end{defi}

The following key result connects the inclusion of the GPT jewel $\mathcal
D_{\mathrm{GPT}\jewel}$ inside a given $\mathcal D_f$ defined above to the positivity of a linear map between two tensor cones. This establishes a bridge between generalized spectrahedral inclusion and measurement compatibility in GPTs.

\begin{prop}\label{prop:inclusion-positivity}
Let $\mathbf k \in \mathbb N^g$ and let $(L, L^+)$ be a proper ordered vector space. Then, the inclusion $\mathcal
D_{\mathrm{GPT}\jewel}(\mathbf k; L, L^+) \subseteq \mathcal D_{f}(\mathbf k; L, L^+)$ holds if and only if $\Phi^{(f)}
\otimes \mathrm{id}_L: (E_{\mathbf k} \otimes L, E_{\mathbf k}^+ \otimes_{\mathrm{max}} L^+) \to (A \otimes L, A^+\otimes_{\mathrm{min}} L^+)$ is positive, where $\Phi^{(f)}$ is defined as in Equation~\eqref{eq:def-Phi-f}.
\end{prop}

\begin{proof}
This follows directly from Propositions \ref{prop:rightcone} and \ref{prop:containment-to-map-inclusion}.
\end{proof}

We have two important special cases.

\begin{prop} \label{prop:level_1}
Let $\mathbf k \in \mathbb N^g$ and let $(V,V^+,\mathds 1)$ be a GPT. Then, $\mathcal D_{\mathrm{GPT}\jewel}(\mathbf k; \mathbb R, \mathbb R_+) \subseteq \mathcal D_f(\mathbf k; \mathbb R, \mathbb R_+)$ if and only if $\{f_1^{(i)}, \ldots, f_{k_i}^{(i)}\}$ form measurements for all $i \in [g]$, where $f_{k_i}^{(i)} := \mathds 1 - \sum_{j = 1}^{k_i-1} f_j^{(i)}$.
\end{prop}
\begin{proof} This follows immediately from Propositions \ref{prop:inclusion-positivity} and \ref{prop:pu_map}. 

An alternative proof is as follows. Let $E_i:=\mathrm{span}\{ 1_{k_i},\ v^{(k_i)}_j,\ j\in [k_i-1]\}$.
Corollaries \ref{cor:E-factors} and \ref{cor:direct-sum-splits} yield that the inclusion is true if and only if
\begin{equation*}
\mathcal D_{(1_{k_i}, v^{(k_i)})}(\mathbb R, E_i^+) \subseteq \mathcal D_{(\mathds 1, f^{(i)})}(\mathbb R, A^+) \qquad \forall i \in [g].
\end{equation*}
From \cite[Lemma 4.3]{bluhm2020compatibility}, it follows that the extreme rays of $\mathcal D_{(1_{k_i}, v^{(k_i)})}(\mathbb R, E_i^+)$ are 
\begin{equation*}
\mathbb R_+(1, \frac{k_i}{2}e_j) \quad \forall j \in [k_i-1], \qquad \mathbb R_+(1,-\frac{k_i}{2}(1, \ldots 1)).
\end{equation*}
Here, $\{e_j\}_{j \in [k_i-1]}$ is the standard basis in $\mathbb R^{k_i-1}$. Thus, the inclusion is equivalent to 
\begin{equation*}
k_i f_j^{(i)} \in A^+ \quad \forall j \in [k_i-1] \quad \land \quad  k_i \mathds 1 - \sum_{j = 1}^{k_i-1} k_i f_j^{(i)} \in A^+
\end{equation*}
for all $i \in [g]$. Dividing by $k_i$ proves the assertion. 
\end{proof}

\begin{thm}\label{thm:inclusion_compatible}
Let $\mathbf k \in \mathbb N^g$ and let $(V,V^+,\mathds 1)$ be a GPT. Then, $\mathcal D_{\mathrm{GPT}\jewel}(\mathbf k; V, V^+) \subseteq \mathcal D_f(\mathbf k; V, V^+)$ if and only if $\{f_1^{(i)}, \ldots, f_{k_i}^{(i)}\}_{i \in [g]}$ are \emph{compatible} measurements, where $f_{k_i}^{(i)} := \mathds 1 - \sum_{j = 1}^{k_i-1} f_j^{(i)}$ for all $i \in [g]$.
\end{thm}

\begin{proof} Follows 
by Propositions \ref{prop:inclusion-positivity}, \ref{prop:positivity-extension} and Theorem \ref{thm:etb_ext}.
\end{proof}

The characterization of extendable maps in Proposition \ref{prop:positivity-extension} gives us another  condition
for  compatibility.

\begin{cor}\label{cor:compatible-witness}
Let $(V, V^+, \mathds 1)$ be a GPT. Elements $\{f_1^{(i)}, \ldots, f_{k_i}^{(i)}\} \in A^{k_i}$ form compatible measurements for all $i \in [g]$ if and only if 
\begin{equation}\label{eq:compatibility-with-witnesses}
\forall z \in \mathcal D_{\mathrm{GPT} \jewel}(\mathbf k; V, V^+), \qquad \mathds 1(z_0) \geq \sum_{i=1}^g 
\sum_{j=1}^{k_i-1} \langle   \frac{2}{k_i}\mathds 1- 2f_j^{(i)}, z_j^{(i)} \rangle.
\end{equation}
\end{cor}

\begin{proof}
The equation above is equivalent to the positivity of the linear form $s_\Phi$ from Proposition \ref{prop:positivity-extension}.
\end{proof}

\begin{remark}\label{rem:cones_in_sh}
Note that the last condition  appears  much weaker than the requirement from Theorem \ref{thm:inclusion_compatible}.
Indeed, the condition in the above theorem reads 
$$
y_{f,z}:=\mathds 1 \otimes z_0 + \sum_{i=1}^g \sum_{j=1}^{k_i-1}  \left[2f_1^{(i)} - \frac{2}{k_i}\mathds 1\right]
\otimes z_j^{(i)} \in A^+ \otimesmin V^+,\quad \forall z=(z_0,z^{(i)}_j)_{ij}\in \mathcal D_{\mathrm{GPT} \jewel}(\mathbf
k; V, V^+),$$
whereas  the condition in Equation \eqref{eq:compatibility-with-witnesses} requires only evaluation of $y_{f,z}$
 against the tensor $\chi_V$ from Equation~\eqref{eq:def-chi} (instead of all elements of $V^+\otimesmax A^+$). 
Note that this means that if $C\subset A\otimes V$ is a tensor cone such that 
$\chi_V \in  C^*$, then we may replace the cone $A^+\otimesmin V^+$ in Definition \ref{def:sh_f}  (with $L=V$) by
$ C$. If we have a reasonable family of tensor cones  $ C_L\subset A\otimes L$ 
for a family of proper ordered vector spaces $(L,L^+)$
 containing $(V,V^+)$ and such that $ C_V= C$, then we may replace the cones in Definition \ref{def:sh_f} for all such $L$ (cf.\ Proposition \ref{prop:EB}). This would lead
 to somewhat different definitions, but equivalent in the two extreme cases ($L=\mathbb R$ and $L=V$). Compare this to 
 \cite[Theorem 5.3]{bluhm2018joint}.
\end{remark}

 The ''intermediate'' cases of Proposition \ref{prop:inclusion-positivity} can be also related to compatibility, as follows. Let $(L,L^+)$ be an ordered vector space with $L^+$ closed and let $\Psi: (A,A^+)\to (L^*,(L^+)^*)$ be a positive map such that  $\Psi(\mathds{1})\in \operatorname{int}((L^+)^*)$, then  $\Psi(\mathds 1)$ is an order unit in $(L^*,(L^+)^*)$
and we can think of the triple $(L,L^+,\Psi(\mathds 1))$ as describing a GPT.  Clearly, if $f=\{f_1,\dots,f_k\}$ is a
measurement for $(V,V^+,\mathds 1)$, then $\Psi(f)=\{\Psi(f_1),\dots,\Psi(f_k)\}$ is a measurement for
$(L,L^+,\Psi(\mathds 1))$.

\begin{prop}\label{prop:inclusion_intermediate} Let $(L,L^+)$ be a proper ordered vector space. Then
$\mathcal
D_{\mathrm{GPT}\jewel}(\mathbf k; L, L^+) \subseteq \mathcal D_{f}(\mathbf k; L, L^+)$ if and only if for any 
positive map $\Psi: (A,A^+)\to (L^*, (L^+)^*)$ such that $\Psi(\mathds 1) \in \operatorname{int} (L^+)^\ast$, the elements $\Psi(f^{(i)}_j)$,
$j \in [k_i]$, $i\in [g]$ form a collection of compatible measurements on $(L,L^+,\Psi(\mathds 1))$.

\end{prop}

\begin{proof}
We start by observing that for a proper cone $C$, $\phi \in \operatorname{int} C^\ast$ if and only if $\phi(c) > 0$ for all $c \in C \setminus \{0\}$ \cite[Theorem 3.5]{Aliprantis2007}. Let $\psi \in \operatorname{int} \left(V^+\otimesmax (L^+)^*\right)$. Then, $\langle \psi, \alpha \otimes v \rangle > 0$ for all $\alpha \in A^+$, $v \in L^+$ such that $\alpha \otimes v \neq 0$. Going to the associated positive map $\Psi: (A,A^+)\to (L^*, (L^+)^*)$, we have
\begin{equation*}
    \langle \psi, \alpha \otimes v \rangle = \langle \Psi(\alpha), v \rangle.
\end{equation*}
Thus, $\Psi(\mathds 1) \in \operatorname{int} (L^+)^\ast$ for all  $\psi \in \operatorname{int} \left(V^+\otimesmax (L^+)^*\right)$.

For  $z=(z_0,z^{(i)}_j)_{ij}\in   \mathcal D_{\mathrm{GPT} \jewel}(\mathbf k; L, L^+)$ let $y_z$ be the corresponding
element in $E_\kk^+\otimesmax L^+$. Then the inclusion is equivalent to
\begin{equation} \label{eq:intermediate-subspaces}
\langle \psi, (\Phi^{(f)}\otimes \mathrm{id})(y_z)\rangle \ge 0,\quad \forall z\in   \mathcal D_{\mathrm{GPT} \jewel}(\mathbf k;
L, L^+),\ \forall \psi\in V^+\otimesmax (L^+)^*.
\end{equation}
Let $\Psi: (A,A^+)\to (L^*,(L^+)^*)$ be the positive map corresponding to $\psi$. Then 
\[
\langle \psi, (\Phi^{(f)}\otimes \mathrm{id})(y_z)=\langle \Psi(\mathds 1),z_0\rangle + \sum_{i=1}^g \sum_{j=1}^{k_i-1} \langle 2\Psi(f_j^{(i)}) -
\frac{2}{k_i}\Psi(\mathds 1), z_j^{(i)}
\rangle.
\]
 By density, we can restrict to $\psi$ such that $\Psi(\mathds 1) \in \operatorname{int} (L^+)^\ast$ in Equation \eqref{eq:intermediate-subspaces} since the set of such elements contains $\operatorname{int}\left( V^+\otimesmax (L^+)^* \right)$ as argued above. The statement now follows by Corollary \ref{cor:compatible-witness}.
\end{proof}

\begin{remark} \label{rem:question1} The set of positive maps $\Psi$, or vectors $\psi\in V^+\otimesmax (L^+)^*$, can be restricted  to
extremal elements. In general, however, extremal positive maps are difficult to characterize. 
If the generalized spectrahedra $\mathcal D_f(\kk, L,L^+)$ are defined as in Remark \ref{rem:cones_in_sh}  by a
family of cones $\mathcal C_L\subset  A\otimes L$, then it is enough to consider maps such that $\Psi \otimes \mathrm{id}$ is positive with respect to this family and using extremal maps with this property might be more convenient.

 As an example, consider the case of quantum systems $\mathrm{QM}_d$ from Example \ref{ex:QM}. Here, 
the canonical tensor $\chi_V$ is the maximally entangled state which is
clearly an element in $\mathrm{PSD}_{d^2}$. So if we restrict the intermediate cases to $(L,L^+)=(\mathcal M_k^{\mathrm{sa}}(\mathbb C),\mathrm{PSD}_k)$, we can replace the minimal tensor product $\mathrm{PSD}_d\otimesmin \mathrm{PSD}_k$ by the larger cone $\mathrm{PSD}_{dk}$. Accordingly, we can reduce to extremal completely positive maps $\mathcal M_d(\mathbb C)\to \mathcal M_k(\mathbb C)$, 
 which are of the form $X\mapsto A^*XA$ for some 
 $A: \mathbb C^k\to \mathbb C^d$. The requirement $\Psi(\mathds 1) \in \operatorname{int} \mathrm{PSD}_k$ is met if $A^\ast A > 0$. If $k\le d$, it is easy to see that $A^*f^{(i)}A$ are compatible for all $A$ if and only if they
are compatible for all isometries $A:\mathbb C^k\to \mathbb C^d$, see 
\cite[Theorem 5.3]{bluhm2018joint}. 
\end{remark}

\subsection{Compatibility region and inclusion constants}

We can now relate the compatibility region  $\Gamma(\kk, V,V^+)$ (Definition \ref{def:compreg})
 to the inclusion constants defined in Section \ref{sec:gen-spectrahedra} (Definition \ref{def:inclusion-constants}).
The following is a restriction of the set of inclusion constants, where we require the coefficients by which we scale to be the same on some elements
\begin{defi}
Given a GPT $(V,V^+,\mathds{1})$, $g \in \mathbb N$ and $\mathbf{k} \in \mathbb N^g$, we define the set of \emph{inclusion constants for the $(\mathbf k; V,V^+)$-jewel} as
\begin{align*}\Delta(\mathbf{k}; V,V^+):=&\{s \in [0,1]^g \, : \, \forall f_j^{(i)} \in A,~j \in [k_i-1],~ i \in [g],~ \, \mathcal D_{\mathrm{GPT}\jewel}(\mathbf k; \mathbb R, \mathbb R^+) \subseteq \mathcal D_f(\mathbf k; \mathbb R, \mathbb R^+) \\ &\implies 
(1,s_1^{\times(k_1-1)},\ldots, s_g^{\times(k_g-1)} )\cdot \mathcal D_{\mathrm{GPT}\jewel}(\mathbf k; V,V^+) \subseteq \mathcal D_f(\mathbf k; V,V^+) \}.
\end{align*}
If $\mathbf k = 2^{\times g}$, we will just write $\Delta(g; V,V^+)$.
\end{defi}
\begin{remark}
The notation $\Delta(\mathbf{k}; V,V^+)$ introduced above should not be confused with the notation $\Delta_k$ for the $(k-1)$-dimensional probability simplex, used extensively in Section \ref{sec:maps}.
\end{remark}

We can now prove that the set of inclusion constants for the GPT jewel is precisely the compatibility region of the GPT introduced in Definition \ref{def:compreg}. The result below connects the operationally defined compatibility region with the geometrical set of inclusion constants. 

\begin{thm} \label{thm:delta-is-gamma}
Given a GPT $(V,V^+,\mathds{1})$, it holds that $\Gamma(\mathbf k; V,V^+) = \Delta(\mathbf k; V,V^+)$.
\end{thm}
\begin{proof}
From Proposition \ref{prop:level_1} we infer that the inclusion $\mathcal D_{\mathrm{GPT}\jewel}(\mathbf k; \mathbb R, \mathbb R^+) \subseteq \mathcal D_f(\mathbf k; \mathbb R, \mathbb R^+)$ holds if and only if the $\{f_1^{(i)}, \ldots f_{k_i}^{(i)}\}$ are measurements for all $i \in [g]$. Here, $f_{k_i}^{(i)} = \mathds 1 - \sum_{j=1}^{k_i-1}f_j^{(i)}$ for all $i \in [g]$. The statement then is an easy consequence of Theorem \ref{thm:inclusion_compatible} and the following equivalence:
\begin{align*}
&(1,s_1^{\times(k_1-1)}.\ldots, s_g^{\times(k_g-1)} )\cdot \mathcal D_{\mathrm{GPT}\jewel}(\mathbf k; V,V^+)\subseteq \mathcal D_f(\mathbf k; V,V^+) \\ &\qquad
\iff \mathds{1} \otimes z_0 + \sum_{i = 1}^g \sum_{j = 1}^{k_i-1}s_i\left(2f_j^{(i)} - \frac{2}{k_i}\mathds{1}\right) \otimes z_j^{(i)} \in A^+ \otimes_{\mathrm{min}}V^+ \quad\forall z \in \mathcal D_{\mathrm{GPT}\jewel}(\mathbf k; V,V^+) \\
&\qquad\iff \mathcal D_{\mathrm{GPT}\jewel}(\mathbf k; V,V^+)\subseteq \mathcal D_{f'}(\mathbf k; V,V^+),
\end{align*}
where $(f')^{(i)}_j = s_i f^{(i)}_j + (1-s_i)\frac 1{k_i} \mathds{1}$ for $j \in [k_i-1]$, $i \in [g]$.
\end{proof}

\subsection{Inclusion constants from symmetrization} \label{sec:symmetrization}
The aim of this section is to show that we can obtain bounds on the inclusion set of the GPT jewel from consideration of the inclusion set of the GPT diamond. As the latter has more symmetries, it is much easier to work with. Using Theorem \ref{thm:delta-is-gamma}, we obtain bounds on measurements with $\mathbf k$ outcomes derived from the compatibility region for dichotomic measurements. This section is inspired by \cite[Section 7]{bluhm2020compatibility}.

The following is the GPT analogue of Theorem 7.2 of \cite{bluhm2020compatibility}. Its proof is very similar. 
\begin{thm}\label{thm:symmetrization}
Let $g \in \mathbb N$, $k_j \in \mathbb N$ for all $j \in [g]$. Let $(V,V^+,\mathds 1)$ be a GPT. Then,
\begin{equation*}
    \left((k_1-1)^{-2}, \ldots,  (k_g-1)^{-2}\right) \cdot \Delta\left(\sum_{i = 1}^g(k_i-1);V,V^+\right) \subseteq \Delta(\mathbf k; V, V^+).
\end{equation*}
\end{thm}
\begin{proof}
Let $\bar k = \sum_{i = 1}^g(k_i-1)$. 
First, we shall find conditions for $\lambda_i \in [0,1]$ and $\mu_i \in [0,1]$, $i \in [g]$, such that 
\begin{equation*}
    \lambda \cdot \mathcal D_{\mathrm{GPT}\diamond}(\bar k; \mathbb R,\mathbb R^+) \subseteq \mathcal D_{\mathrm{GPT}\jewel}(\mathbf k; \mathbb R,\mathbb R^+) 
\end{equation*}
and
\begin{equation*}
      \mu \cdot \mathcal D_{\mathrm{GPT}\jewel}(\mathbf k; \mathbb R,\mathbb R^+)  \subseteq \mathcal D_{\mathrm{GPT}\diamond}(\bar k; \mathbb R,\mathbb R^+),
\end{equation*}
where $\lambda = (1,\lambda_1^{\times(k_1-1)}, \ldots, \lambda_g^{\times(k_g-1)})$ and $\mu = (1,\mu_1^{\times(k_1-1)}, \ldots, \mu_g^{\times(k_g-1)})$.

By Corollary \ref{cor:direct-sum-splits}, it is enough to consider 
\begin{equation*}
    (1,\lambda_i^{\times(k_i-1)}) \cdot \mathcal D_{\mathrm{GPT}\diamond}((k_i-1); \mathbb R,\mathbb R^+) \subseteq \mathcal D_{\mathrm{GPT}\jewel}(k_i; \mathbb R,\mathbb R^+) 
\end{equation*}
and 
\begin{equation*}
      (1,\mu_i^{\times(k_i-1)}) \cdot \mathcal D_{\mathrm{GPT}\jewel}(k_i; \mathbb R,\mathbb R^+)  \subseteq \mathcal D_{\mathrm{GPT}\diamond}(k_i-1; \mathbb R,\mathbb R^+).
\end{equation*}
Note that $\mathcal D_{1_{\mathbf k} \otimes a}(L,E_{\mathbf k} \otimes_{\max} N^+) = \mathcal D_{a}(L, N^+)$ by Lemma \ref{lem:elements-in-the-max}, where $(L, L^+)$ and $(M,M^+)$ are proper ordered vector, $N^+$ a tensor cone for $M^+$ and $L^+$ and $1_{\mathbf k} \otimes a = (1_{\mathbf k} \otimes a_1, \ldots, 1_{\mathbf k} \otimes a_n) \in (E_{\mathbf k} \otimes M)^n$, $n \in \mathbb N$. From Lemma 4.3 of \cite{bluhm2020compatibility}, we know that the extreme rays of $\mathcal D_{\mathrm{GPT}\diamond}(k_i-1; \mathbb R,\mathbb R^+)$ are
\begin{equation*}
    \mathbb R_+ (1, \pm e_j), \qquad j \in [k_i-1]
\end{equation*}
while the extreme rays of $\mathcal D_{\mathrm{GPT}\jewel}(k_i; \mathbb R,\mathbb R^+)$ are
\begin{equation*}
    \mathbb R_+ \left(1,\frac{k_i}{2} e_j\right),~j \in [k_i-1] \quad \mathrm{and} \quad \mathbb R_+ \left(1,-\frac{k_i}{2}\underbrace{(1, \ldots, 1)}_{k_i-1}\right).
\end{equation*}
We find thus that $(1,\pm \lambda_i e_j) \in \mathcal D_{\mathrm{GPT}\jewel}(k_i; \mathbb R,\mathbb R^+)$ if 
\begin{equation*}
    \pm \lambda_i e_j \in \left[-\frac{k_i}{2(k_i-1)}e_j, \frac{k_i}{2}e_j  \right]
\end{equation*}
since
\begin{equation*}
    -\frac{k_i}{2(k_i-1)}e_j = \frac{1}{k_i-1}\left(\sum_{l \in[k_i-1]\setminus j} \frac{k_i}{2}e_l - \frac{k_i}{2}(1, \ldots, 1)\right)
\end{equation*}
Thus, $\lambda_i \leq k_i/(2(k_i-1))$. Considering the $\ell_1$ norm of the extreme points of $D_{\mathrm{GPT}\jewel}(k_i; \mathbb R,\mathbb R^+) \cap\{(x_0, \ldots, x_{k_i-1})~:~x_0 = 1\}$, we infer moreover that $\mu_i \leq 2/(k_i(k_i-1))$.
Proposition \ref{prop:min-max_sandwich} and Remark \ref{rem:jewel-is-max} together imply furthermore that from
\begin{equation*}
      \mu \cdot \mathcal D_{\mathrm{GPT}\jewel}(\mathbf k; \mathbb R,\mathbb R^+)  \subseteq \mathcal D_{\mathrm{GPT}\diamond}(\bar k; \mathbb R,\mathbb R^+),
\end{equation*}
follows
\begin{equation*}
      \mu \cdot \mathcal D_{\mathrm{GPT}\jewel}(\mathbf k; V,V^+)  \subseteq \mathcal D_{\mathrm{GPT}\diamond}(\bar k; V,V^+),
\end{equation*}
Let $s \in \Delta(\bar k; V, V^+)$ and let $\circ$ denote the entrywise product of vectors. Let $f_j^{(i)} \in A$ for all $i \in [g]$ and $j \in [k_i-1]$. Then, by the above, we have the implication
\begin{align*}
        &\lambda \cdot \mathcal D_{\mathrm{GPT}\diamond}(\bar k; \mathbb R,\mathbb R^+) \subseteq \mathcal D_{\mathrm{GPT}\jewel}(\mathbf k; \mathbb R,\mathbb R^+) \subseteq \mathcal D_{f}(\mathbf k; \mathbb R, \mathbb R^+) \\
        \implies & \mu \circ \lambda \circ s \cdot \mathcal D_{\mathrm{GPT}\jewel}(\mathbf k; V,V^+)  \subseteq \lambda \circ s \cdot \mathcal D_{\mathrm{GPT}\diamond}(\bar k; V,V^+) \subseteq \mathcal D_{f}(\mathbf k; V, V^+),
\end{align*}
from which $\mu \circ \lambda \circ s \in \Delta(\mathbf k, V, V^+)$ and the assertion follows since $\lambda_i \cdot \mu_i \leq (k_i-1)^{-2}$.
\end{proof}

\section{Tensor crossnorms and compatibility of effects} \label{sec:tensor-norms}

In this section, we will only consider dichotomic measurements, that is, with outcomes in $\{1,2\}$; in other words, we set $k_1=\cdots=k_g = 2$.  Any such measurement is determined by an effect in $A^+$ related to  the probability of obtaining the outcome $1$. We will say that a collection of effects is compatible if the corresponding
collection of binary
measurements is compatible. For simplicity, we will write $E_g$ for $E_{\kk}$ with  $\kk=2^{\times g}$.

\subsection{\texorpdfstring{$E_g$}{Eg} as a centrally symmetric GPT}\label{sec:Eg-sym-GPT}

Let us begin by studying more in-depth the space $E_g$ and see that it gives itself rise to a GPT. Recalling the notations of Sections \ref{sec:maps} and \ref{sec:jewel-ordered-vector-spaces}, the distinguished basis of
 $E_g$ is $w=\{1_g,c_i\}$, where \begin{align*}
	1_g &= (1,1, \ldots 1)\\
	c_i &= (1,1)^{\otimes (i-1)} \otimes (1,-1) \otimes (1,1)^{\otimes (g-i)}.
\end{align*}
and the  dual basis $w^*$ in the space $E_g^* \cong E_g$ is $w^*=\{\check{1}_g, \check c_i\}$, where
\begin{align*}
\check{1}_g:=2^{-g}1_g,\qquad \check c_i &:= 2^{-g} c_i.
\end{align*}

 Let us point out that $(E_g,E_g^+, \check{ 1}_g)$ is a centrally symmetric GPT in the sense of \cite[Definition
25]{lami2018ultimate} corresponding to the $\ell_1$ norm in $\mathbb R^g$ (see also Section
\ref{sec:centrally-symmetric-GPTs}). Indeed, we clearly have $E_g=\mathbb R 1_g\oplus \bar E_g$, where
\[
\bar E_g=\mathrm{span}\{c_i,\ i\in [g]\}=\mathrm{span}\{\check c_i,\ i\in [g]\}
\]
and an element $a1_g+\sum_i x_ic_i\in E_g$ is in $E_g^+$ if and only if 
\[
a+ \sum_{i = 1}^g \epsilon_i x_i\ge0 ,\forall \epsilon \in \{\pm 1\}^g, 
\]
as can be seen e.g. from Remark \ref{rem:GPT-diamond} (with $L=\mathbb R$). This is equivalent to
\[
a\ge \sum_{i=1}^g |x_i|= \left\|\sum_i x_ie^{(g)}_i\right\|_{1},
\]
here $\{e^{(g)}_i\}$ is the canonical basis in $\mathbb R^g$ and $\|\cdot\|_1$ is the $\ell_1$ norm.
It follows that 
the subspace  $\bar E_g$ endowed with the restriction of the  base norm $\max\{a, \sum_i|x_i|\}$ can be identified with $\ell^g_1$, via the
isometry
\[
c_i\mapsto e^{(g)}_i,\qquad i\in [g].
\]
By duality, positivity in  $(E_g^*,(E_g^+)^*)$ is characterized by the $\ell_\infty$ norm  and 
the map 
\[
\check c_i\mapsto e^{(g)}_i,\qquad i\in [g]
\]
is an isometry of  $(\bar E_g, \|\cdot\|_{\check{1}_g})$ onto  $\ell_\infty^g$. These identifications will be used
throughout.

Since the $\ell_1$ and $\ell_\infty$ norms are invariant under sign changes of the coordinates, the map
\begin{equation}\label{eq:sign-change-map}
    \sigma_\epsilon: a1_g + \sum_i x_i c_i\mapsto a1_g + \sum_i \epsilon_i x_ic_i
\end{equation}
is an order isomorphism of both $(E_g, E_g^+)$ and $(E_g^*,(E_g^+)^*)$, for any sign  vector
$\epsilon \in \{\pm 1\}^g$. 
This implies that for any proper ordered vector space $(L,L^+)$ and  $z\in E_g\otimes L$ we have
\begin{align*}
z \in E_g^+ \otimesmin L^+ &\iff (\sigma_\epsilon\otimes \mathrm{id})(z) \in E_g^+ \otimesmin L^+\\
z \in E_g^+ \otimesmax L^+ &\iff (\sigma_\epsilon\otimes \mathrm{id})(z)\in E_g^+ \otimesmax L^+
\end{align*}
for any $\epsilon\in \{\pm1\}^g$, and similarly for $(E_g^+)^*$.

\subsection{Effects and tensor crossnorms}

Now we can make the connection between the compatibility of effects and norms on their corresponding tensors. Let $(V, V^+, \mathds 1)$ be a GPT. In addition, let  $f=(f_1, \ldots, f_g)$ be  $g$-tuple of elements in $A$ and let 
\begin{equation}\label{eq:def-phi}
\phi^{(f)} = \check{1}_g  \otimes \mathds 1 + \underbrace{\sum_{1=1}^g \check c_i \otimes (2f_i-\mathds 1)}_{\bar \phi^{(f)}}.
\end{equation}
By the results of Section \ref{sec:maps}, $f$ is a collection of  effects if and only if $\varphi^{(f)}\in (E_g^+)^*\otimesmax A^+$
and $f$ is compatible if and only if $\varphi^{(f)}\in (E_g^+)^*\otimesmin A^+$.
Let $\sigma_\epsilon$ for $\epsilon\in \{\pm1\}^g$ be as in Equation~\eqref{eq:sign-change-map}.
Then
\begin{equation}\label{eq:f_eps}
(\sigma_\epsilon\otimes \mathrm{id})(\phi^{(f)})=\check{1}_g  \otimes \mathds 1 + \sum_{1=1}^g \check c_i \otimes
\epsilon_i(2f_i-\mathds 1)=\phi^{(f_\epsilon)},
\end{equation}
where $(f_\epsilon)_i=f_i$ if $\epsilon_i=1$ and $(f_\epsilon)_i=1-f_i$ if $\epsilon_i=-1$. 
The invariance of the maximal and minimal tensor products under $\sigma_\epsilon \otimes \mathrm{id}$ is a  manifestation of
the fact that a relabelling of  measurement outcomes defines  again a  measurement and that such a relabelling  does not change compatibility  of the effects under study.

Our aim in this section is to make a connection between compatibility and reasonable crossnorms of Banach spaces (see Section \ref{sec:tensor-norms-prelim}). The results of \cite{lami2018ultimate} imply that any base norm on a bipartite GPT is a reasonable crossnorm.
\begin{prop} \label{prop:base-norms-are-crossnorms}
Consider the GPT $(V_A \otimes V_B, V_A^+ \otimes V_B^+, \mathds 1_A \otimes \mathds 1_B)$, where $(V_{\#}, V_{\#}^+, \mathds 1_{\#})$ are GPTs for ${\#} \in \{A, B\}$ and $V_A^+ \otimes V_B^+$ is some proper tensor cone for $V_A^+$, $V_B^+$. Then, the base norm $\norm{\cdot}_{V_A^+ \otimes V_B^+}$ is a reasonable crossnorm.
\end{prop}
\begin{proof}
Let $x \in V_A \otimes V_B$. From \cite[Equation (26)]{lami2018ultimate}, it holds that $\norm{x}_{V_A^+ \otimes V_B^+} \leq \norm{x}_{V_A^+ \otimes_{\min} V_B^+}$. Together with \cite[Proposition 22]{lami2018ultimate}, this implies that $\norm{x}_{V_A^+ \otimes V_B^+} \leq \norm{x}_{\pi}$. Moreover,  \cite[Proposition 22]{lami2018ultimate} implies $\norm{x}_{\epsilon} \leq \norm{x}_{\mathrm{LO}}$, where we refer to \cite{lami2018ultimate} for the definition of the latter norm. From the discussion after \cite[Definition 6]{lami2018ultimate}, it follows that $\norm{x}_{\mathrm{LO}} \leq \norm{x}_{V_A^+ \otimes V_B^+}$. Thus,
\begin{equation*}
\norm{x}_{\epsilon} \leq \norm{x}_{V_A^+ \otimes V_B^+} \leq \norm{x}_{\pi}.
\end{equation*}
Therefore, the assertion follows from point (a) of Proposition \ref{prop:reasonable-cross-norms}.
\end{proof}

We now return to compatible tuples of effects. In this paragraph, $A$ and $\bar E_g$ will always be endowed with the
order unit norms $\|\cdot\|_{\mathds 1}$ and $\|\cdot\|_{\check{1}_g}$.

\begin{thm} \label{thm:effect-tensors} Let $f=(f_1,\dots,f_g)$ be a $g$-tuple of elements in $A$ and let 
\[
\bar \phi^{(f)}= \sum_{i = 1}^g \check c_i \otimes (2f_i-\mathds 1). 
\]
Then we have
\begin{enumerate}
\item Let $\|\cdot\|_\epsilon$ be the injective crossnorm in $\ell_\infty^g\otimes A$. Then $f$ is a collection of effects  if and only if 
\[
\|\bar \phi^{(f)}\|_{\epsilon}\leq 1.
\]
\item There is a reasonable crossnorm $\|\cdot\|_c$ in $\ell_\infty^g \otimes A$ such that $f$ is a compatible
collection of effects if and only if
 \[
\|\bar \phi^{(f)}\|_{c}\leq 1.
\]

\end{enumerate}
\end{thm}
\begin{proof}
(1) can be easily observed directly from  
\[
\|\bar \phi^{(f)}\|_{\epsilon}=
\|\sum_{i = 1}^g \check c_i\otimes (2f_i-\mathds 1)\|_{\epsilon }= \max_{i \in [g]} \|2f_i-\mathds 1\|_{\mathds 1},
\]
since $\ell_\infty \otimes_\epsilon X \cong \ell_\infty(X)$ as Banach spaces \cite[Example 3.3]{Ryan2002}. Recall that $\norm{x}_\mathds{1} = \inf \{t \geq 0: x \in t[-\mathds 1, \mathds 1]\}$.
For (2), note that we have by Theorem \ref{thm:etb_ext} and Equation~\eqref{eq:f_eps} that $f$ is a compatible  
collection of effects if and only if
\[
\check{1}_g\otimes \mathds 1 \pm \bar\phi^{(f)}\in (E_g^*)^+ \otimesmin A^+,
\]
which is equivalent to
\[
\|\bar\phi^{(f)}\|_{\check{1}_g\otimes \mathds 1}\le 1.
\] 
Here $\|\cdot\|_{\check{1}_g\otimes \mathds 1}$ is the order unit norm for the tensor product GPT 
 $(E_g\otimes V,  E_g^+\otimesmax V^+, \check{1}_g\otimes\mathds 1)$. As the dual to the base norm which is a reasonable
crossnorm by Proposition \ref{prop:base-norms-are-crossnorms}, $\|\cdot\|_{\check{1}_g\otimes \mathds 1}$ is a reasonable crossnorm 
 in $E_g^*\otimes A$. 
Hence we may define the norm $\|\cdot\|_c$
as the restriction of the order unit norm  to $\bar E_g \otimes A$. 

Note that part (1) can be proved similarly using the order unit norm for the 
GPT $(E_g\otimes V,  E_g^+\otimesmin V^+, \check{1}_g\otimes\mathds 1)$. The result follows from the fact that by \cite[Proposition 22]{lami2018ultimate} the base norm in this case is equal to the projective norm of the two base norm spaces. 
\end{proof}

\begin{question}
The symmetry of $E_g$ is crucial in the above proof. For general $E_{\mathbf{k}}$, however, such symmetry is not present, so the proof does not immediately extend. It is an interesting question whether there is a characterization similar to Theorem \ref{thm:effect-tensors} also for measurements with more than two outcomes.
\end{question}

We derive now a more explicit form of the crossnorm $\|\cdot\|_c$.

\begin{prop}\label{prop:cross_norms}
Let $\bar \phi\in \ell^g_\infty\otimes A$, then
\[ 
\|\bar \phi\|_{c}=\inf\left\{\left\|\sum_j h_j\right\|_{\mathds 1}, \ \bar \phi=\sum_j
z_j\otimes h_j,\ \|z_j\|_\infty=1,\ h_j\in A^+\right\}.
\]

\end{prop}

\begin{proof}  Note that $\bar\phi=\bar\phi^{(f)}$ for some $g$-tuple of elements in $A$. 
By definition, 
\begin{align*}
\|\bar \phi^{(f)}\|_{c }&=\inf\{\lambda>0, \lambda\check{1}_g\otimes\mathds 1\pm \bar \phi^{(f)}\in
(E_g^+)^*\otimesmin A^+\}\\
&=\inf\{\lambda>0, \lambda\check{1}_g\otimes\mathds 1- \bar \phi^{(f)}\in
(E_g^+)^*\otimesmin A^+\}\\
&= \inf\{\lambda >0,\ \lambda\check{1}_g\otimes\mathds 1- \bar \phi^{(f)} =\sum_j x_j\otimes h_j,\ x_j\in (E_g^+)^*, h_j\in A^+\},
\end{align*}
where we have used Equation~\eqref{eq:f_eps} in the second equality. Assume that 
$\lambda\check{1}_g\otimes\mathds 1- \bar \phi^{(f)} =\sum_j x_j\otimes h_j$ for some $\lambda>0$ and $x_j\in (E_g^+)^*$,
$h_j\in A^+$. We may identify any $x_j\in (E_g^+)^*$ with a pair $(a_j,z_j)$ with $\|z_j\|_\infty\le a_j$, and then
\[
\sum_j x_j\otimes h_j=\check{1}_g\otimes \sum_j a_jh_j+ \sum_j z_j\otimes h_j.
\]
By some easy reshuffling, we may always assume that $a_j=1$. Hence
\[
\bar \phi^{(f)}= \check{1}_g\otimes \left(\lambda\mathds
1-\sum_j h_j\right)+ \sum_j z_j\otimes h_j
\]
and since $\bar \phi^{(f)}\in \bar E_g\otimes A$, we must have $\lambda\mathds 1=\sum_j h_j$. Since $\|z_j\|_\infty\le 1$, we may normalize $z_j$ and obtain
\[
\bar \phi^{(f)}=\sum_j \|z_j\|_\infty^{-1}z_j\otimes \|z_j\|_\infty h_j,\quad \|z_j\|_\infty h_j\in A^+,\ \sum_j
\|z_j\|_\infty h_j\le \sum_j h_j=\lambda\mathds 1.
\]
This finishes the proof.
\end{proof}

\begin{remark}\label{rem:unitball_rho} Note that the unit ball of $\|\cdot\|_c$ can be written as the set
\[
\left\{\sum_i z_i\otimes h_i, \ \|z_i\|_\infty\le 1,\ h_i\in A^+,\ \sum_i h_i=\mathds 1\right\}.
\]
Alternatively, we can identify the unit ball of $\norm{\cdot}_c$ with
\begin{equation*}
    \left\{(v_1, \ldots, v_g) \in A^g: (\mathds 1, v_1, \ldots, v_g) \in \mathcal D_{\min}(\mathcal C; A, A^+) \right\},
\end{equation*}
where 
\begin{equation*}
    \mathcal C := \{(x, \bar x): x \in \mathbb R, \bar x \in \mathbb R^g, x \geq \norm{\bar x}_\infty\}.
\end{equation*}
\end{remark}

It is almost immediate from the definition that $\norm{\cdot}_c$ can be computed using a conic program (see Section \ref{sec:computing-rho} in the Appendix for a nicer dual formulation).

\bigskip

We next describe the compatibility region as the set  of inclusion constants for the crossnorms $\|\cdot\|_\epsilon$ and 
$\|\cdot\|_c$.
For any element $\bar \phi=\sum_i \check c_i\otimes p_i$ and any $s\in \mathbb R^g$, we
define
\[
s.\phi:=\sum_i \check c_i\otimes s_i p_i.
\]
Given a $g$-tuple of elements  $f=(f_1,\dots, f_g)\in A^g$ and $\bar \varphi^{(f)}=\sum_i\check c_i \otimes (2f_i-\mathds
1)$, we have
\[
s. \bar\phi^{(f)}=\sum_i \check c_i\otimes s_i(2f_i-\mathds 1)=\sum_i \check c_i \otimes (2(s_if_i + (1-s_i)\mathds
1/2)-\mathds 1)=\bar\phi^{(f_s)},
\]
where $f_s=(s_1f_1+(1-s_1)\mathds 1/2,\dots, s_gf_g+(1-s_g)\mathds 1/2)$. If $f$ is a $g$-tuple of effects, then $f_s$
is a $g$-tuple of effects as well, obtained by mixing each $f_i$ with the trivial effect $\mathds 1/2$. 
From the definition of the compatibility region and Theorem \ref{thm:effect-tensors}, we obtain
\begin{equation}\label{eq:Gammaf}
\Gamma(f)=\{s\in[0,1]^g,\ \|s.\bar\varphi^{(f)}\|_c\le 1\},
\end{equation}
in particular
\begin{equation}\label{eq:gammaf}
\gamma(f)=1/\|\bar\varphi^{(f)}\|_c.    
\end{equation}
The following result is now immediate. 

\begin{thm}\label{thm:inclusion_Gamma} We have
\[
\Gamma(g;V,V^+)=\{s\in [0,1]^g,\ \|s.\varphi\|_c\le 1,\ \forall \varphi\in \ell^g_\infty\otimes A,\
\|\varphi\|_\epsilon \le 1\}.
\]
In particular, the compatibility degree satisfies
\[
\gamma(g; V,V^+)= 1/\max_{\|\varphi\|_{\epsilon}\le 1} \|\varphi\|_{c}.
\]

\end{thm}

\begin{ex}\label{ex:QM_rho}
For quantum mechanics $\mathrm{QM_d}$ ($A^+ = \mathrm{PSD}_d$, $\mathds 1 = \mathrm{Tr}$), the expression for the 
unit ball of $\|\cdot\|_c$ from Remark \ref{rem:unitball_rho} becomes the minimal matrix convex set (of a fixed
dimension) $\mathcal W_d^{\min}(B_{\ell^g_\infty})$ (over the unit ball $B_{\ell^g_\infty}$, cf.~\cite[Section 4]{fritz2017spectrahedral} or \cite[Eq.~(1.4)]{passer2018minimal}). Since the unit ball of the 
 injective norm $\|\cdot\|_\epsilon$ in $\ell^g_\infty\otimes (\mathcal M^{\mathrm{sa}}_d(\mathbb C), \norm{\cdot}_\infty)$ is the maximal matrix convex set $\mathcal W_d^{\max}(B_{\ell^g_\infty})$, the expressions  for $\Gamma$ and $\gamma$ in Theorem \ref{thm:inclusion_Gamma} correspond
precisely to the inclusion constants for minimal and maximal matrix convex sets. 

\end{ex}

This gives the following lower bound on the compatibility degree, which uses that $\rho$ is a reasonable crossnorm, thus being upper bounded by the projective tensor norm.

\begin{cor}\label{cor:norm-ratio} We have
\[
\gamma(g; V,V^+)\ge 1/\rho(\ell^g_\infty, A),
\]
where $\rho(\ell^g_\infty, A)$ as in Equation \eqref{eq:rho-quotient}.
\end{cor}

\begin{question} It is not clear whether this bound is tight in general. 
In the case when $g\le \dim(V)$, it is known that the bound $1/g$ is attained by the hypercube GPT, see \cite{jencova2018incompatible}. We will find below a tight bound in the special case of centrally symmetric GPTs (see Section \ref{sec:centrally-symmetric-tensor}), which is larger than the above bound.
\end{question}

\subsection{Incompatibility witnesses} \label{sec:witnesses}

 In this section, we will consider different notions of witnesses. We will start with objects certifying that the elements of $A$ under study are effects before considering objects certifying compatibility. This viewpoint is dual to the one we have hitherto adopted in this section.

We introduce the set of \emph{effect witnesses}
\begin{equation}
\mathcal Q_{\mathrm{GPT} \diamond}(g;V,V^+) := \left\{ z \in V^g : \sum_{i=1}^g \|z_i\|_V \leq 1 \right\}.
\end{equation}
Note that the set above is the unit ball of $\ell_1^g \otimes_\pi (V, \|\cdot\|_V)$ (see Equation~\eqref{eq:projective-norm}). 

\begin{prop} \label{prop:effect-witness}
Elements $f_1, \ldots, f_g \in A$ are effects (i.e.~$f_i, \mathds 1 - f_i \in A^+$ for all $i \in [g]$) if and only if
\begin{equation*}
\sum_{i = 1}^g \langle 2 f_i - \mathds 1, z_i \rangle \leq 1 \qquad \forall (z_1, \ldots, z_g) \in \mathcal Q_{\mathrm{GPT} \diamond}(g;V,V^+).
\end{equation*}
\end{prop}
\begin{proof}
This is straightforward from Theorem \ref{thm:effect-tensors}, the duality of the injective and projective norms and the fact that 
\[
\langle z,\bar\phi^{(f)}\rangle =\sum_{i = 1}^g \langle 2 f_i - \mathds 1, z_i \rangle.
\]
\end{proof}

Similarly, using the results of the previous section, we introduce the following definition:

\begin{defi}[Strict incompatibility witnesses]
Let $\|\cdot\|_{c^*}$ be the dual norm to $\|\cdot\|_c$ in $\mathbb R^g\otimes V$ and let $\mathcal P_{\mathrm{GPT} \diamond}(g;V,V^+)$ denote the unit ball of $\|\cdot\|_{c^*}$. 
An element $z=(z_1,\dots,z_g)\in \mathcal P_{\mathrm{GPT}
\diamond}(g;V,V^+)$ is called a  \emph{strict incompatibility witness} if 
\[
\|z\|_\pi=\sum_i\|z_i\|_V>1.
\]

\end{defi}

The idea behind this definition is clear:  a strict incompatibility witness is a functional such that 
the value $\langle z,\bar\phi^{(f)}\rangle\le 1$ for all compatible $g$-tuples of effects but there exists some
  $f$ with $\langle z,\bar\phi^{(f)}\rangle >1$, so that $z$ witnesses incompatibility of $f$.

\begin{prop}
Effects $f_1, \ldots, f_g \in A^+$ are compatible if and only if $\langle z,\bar\phi^{(f)}\rangle\le 1$ for all strict incompatibility witnesses $z \in \mathcal P_{\mathrm{GPT} \diamond}(g;V,V^+)$.
\end{prop}  
\begin{proof}
It is clear from Theorem \ref{thm:effect-tensors} that $\langle z,\bar\phi^{(f)}\rangle \leq 1$ for all $z \in \mathcal P_{\mathrm{GPT} \diamond}(g;V,V^+)$ if and only if the $f_i$ are compatible effects. The assertion follows since for effects $f_i$, $\langle z,\bar\phi^{(f)}\rangle \leq 1$ for all $z \in \mathcal Q_{\mathrm{GPT} \diamond}(g;V,V^+)$.
\end{proof}

A notion of incompatibility witness was also introduced in \cite{bluhm2020compatibility}.
 Translated to the GPT setting,  this would be an element $z=(z_0,z_1,\dots,z_g)$ from $\mathcal
D_{\mathrm{GPT}\diamond}(g;V,V^+)$ and  $z$ detects incompatibility for some $g$-tuple of effects $f$ if 
\[
z \not \in  \mathcal D_f(g;V,V^+).
\] 
We next show that this notion is closely related to the one  introduced here.

\begin{prop} \label{prop:incomp_witness}
Let $z=(z_1,\dots,z_g)\in V^g$. Then $z\in\mathcal P_{\mathrm{GPT}\diamond}(g;V,V^+)$ if and only if 
 there is some $z_0\in K$
 such that 
\[
(z_0,z_1,\dots,z_g)\in \mathcal D_{\mathrm{GPT}\diamond}(g;V,V^+).
\]
Moreover, in this case, 
\[
(z_0,z_1,\dots,z_g)\not \in  \mathcal D_f(g;V,V^+)
\] 
for a tuple of effects $f$ if and only if 
there exists a positive map $Y:(V,V^+) \to (V,V^+)$ such that  $\langle (\mathrm{id} \otimes Y) (z) ,\bar\phi^{(f)}\rangle >1$.
\end{prop}

\begin{proof}
Note that  the set 
\[
\{z\in V^g,\ \exists z_0\in K,\ (z_0,z)\in \mathcal D_{\mathrm{GPT}\diamond}(g;V,V^+)\}
\]
is convex, closed and contains 0, therefore,  by Corollary \ref{cor:compatible-witness} and Theorem
\ref{thm:effect-tensors}, it is the polar of the unit ball $B_{\|\cdot\|_c}$. This implies the first
statement.

Furthermore, we have that 
\begin{equation*}
0>s_{\Phi^{(f)}}\left(1_g \otimes z_0 + \sum_{i = 1}^g c_i \otimes z_i \right) =\langle \chi, \mathds 1 \otimes z_0 + \sum_{i = 1}^g
(2f_i-\mathds 1) \otimes z_i \rangle =1-\langle z, \bar\phi^{(f)}\rangle. 
\end{equation*}
implies that $(z_0,z)\not\in \mathcal D_f(g;V,V^+)$. Conversely, $(z_0,z)\not\in \mathcal D_f(g;V,V^+)$ implies that there is a $y \in V^+ \otimes_{\max} A^+$ such that
\begin{equation*}
    0 > \langle y, \mathds 1 \otimes z_0 + \sum_{i = 1}^g
(2f_i-\mathds 1) \otimes z_i \rangle.
\end{equation*}
We can find a positive map $Y^\prime: (V,V^+) \to (V,V^+)$ such that $y = (Y^\prime \otimes \mathrm{id}) (\chi_V)=(\mathrm{id}\otimes Y^{'*})(\chi_V)$. 
Observe that we must have $\mathds{1}(Y'(z_0))>0$. Indeed, assume that $\mathds{1}(Y'(z_0))=0$, then since $Y'$ is positive,  we obtain from \eqref{eq:barycenter} and \eqref{eq:z_from_kappa} that $Y'(z_i)=0$ for all $i$, which would imply $\langle y,  \mathds 1 \otimes z_0 + \sum_{i = 1}^g
(2f_i-\mathds 1) \otimes z_i \rangle =0$. By normalizing $Y = (1/\mathds{1}(Y^{\prime}(z_0))) Y^{\prime}$ we obtain that $Y(z_0) \in K$. It follows thus that
\begin{align*}
    0 >& \langle (\mathrm{id} \otimes Y^\ast)(\chi_V), \mathds 1 \otimes z_0 + \sum_{i = 1}^g
(2f_i-\mathds 1) \otimes z_i \rangle \\&= s_{\Phi^{(f)}}\left(1_g \otimes Y(z_0) + \sum_{i = 1}^g c_i \otimes Y(z_i) \right) = 1 - \langle (\mathrm{id} \otimes Y)z,\bar\phi^{(f)}\rangle.
\end{align*}
\end{proof}

A third notion of entanglement witness has been introduced in \cite{jencova2018incompatible}. There, an entanglement witness is a positive map $W:(V(
P_{\mathbf k}), V(P_{\mathbf k})^+) \to (V, V^+)$ such that $\mathrm{Tr}(
(\Phi^{(f)})^{\ast} W) < 0$; recall that $P_\kk = \Delta_{k_1} \times \cdots \times \Delta_{k_g}$ is the polysimplex from Section \ref{sec:maps}. Using the order isomorphism between $(V(P_{\mathbf k}), V(P_{\mathbf k})^+)$ and $(E_{\mathbf k}^\ast,(E_{\mathbf k}^\ast)^+)$, these correspond to positive maps $W:(E^\ast_\mathbf{k}, (E^\ast_\mathbf{k})^+) \to (V,V^+)$ such that $\mathrm{Tr}[(\Phi^{(f)})^{\ast} W] < 0$. We make the connection to elements of $\mathcal P_{\mathrm{GPT}\diamond}(g;V,V^+)$ below.

\begin{prop} \label{prop:incomp_witness_2}
There is a one-to-one correspondence between $z \in \mathcal D_{\mathrm{GPT}\diamond}(g;V,V^+)$ and positive maps $W:(E^\ast_g, (E^\ast_g)^+) \to (V,V^+)$. Moreover, $z_0 \in K$ if and only if $\mathds 1(W(\check{1}_{g})) = 1$. Restricting to such $W$,
\begin{equation*}
\mathrm{Tr}[(\Phi^{(f)})^{\ast} W] < 0 \qquad \iff \qquad  \langle z, \bar \varphi^f \rangle > 1
\end{equation*}
for a collection of effects $f$.
\end{prop}
\begin{proof}
The map $W:(E^\ast_g, (E^\ast_g)^+) \to (V,V^+)$ is positive if and only if $\varphi^W \in  E_g^+ \otimes_{\max} V^+ $. We can decompose
\begin{equation*}
\varphi^W = 1_g\otimes z_0 + \sum_{i = 1}^g  c_i \otimes z_i
\end{equation*}
and $\varphi^W \in  E_g^+ \otimes_{\max} V^+ $ if and only if $z \in \mathcal D_{\mathrm{GPT}\diamond}(g;V,V^+)$. This proves the first assertion. We have $\mathds 1(W(\check{1}_{g})) = \langle \varphi^W, \check{1}_{g} \otimes \mathds{1} \rangle$ from Equation \eqref{eq:def_phiPhi}, which proves the second assertion. Moreover, $\mathrm{Tr}[(\Phi^{(f)})^{\ast} W]=\mathrm{Tr}[W (\Phi^{(f)})^{\ast}] < 0$ if and only $\langle \varphi^{W}, \varphi^{\Phi^{(f)}} \rangle < 0$ by Lemma \ref{lem:trace-of-maps}. Furthermore,
\begin{equation*}
\langle \varphi^{W}, \varphi^{\Phi^{(f)}} \rangle  = \mathds{1}(z_0) + \sum_{i \in [g]} \langle 2 f_i - \mathds{1}, z_i \rangle = 1 - \langle z,\bar \varphi^f \rangle 
\end{equation*}
This proves the last assertion.
\end{proof}
Lastly, we can also connect compatibility witnesses of \cite{jencova2018incompatible} and \cite{bluhm2020compatibility} directly with each other. To clarify this point was one of the motivations for this work. 
\begin{cor}
    There is a one-to-one correspondence between $z \in \mathcal D_{\mathrm{GPT}\jewel}(\mathbf k;V,V^+)$ and positive maps $W:(E^\ast_\mathbf{k}, (E^\ast_\mathbf{k})^+) \to (V,V^+)$. Moreover,
    \begin{equation*}
\exists\text{ positive map }Y:(V,V^+) \to (V,V^+)  \text{ s.t. } 
\mathrm{Tr}[(\Phi^{(f)})^{\ast} Y W] < 0 \quad \iff \quad  z \neq \mathcal D_f(g;V,V^+)
\end{equation*}
for a collection of measurements $f$.
\end{cor}
\begin{proof}
    This follows from Propositions \ref{prop:incomp_witness} and \ref{prop:incomp_witness_2}, realizing that similar statements can be obtained for $\mathcal D_{\mathrm{GPT}\jewel}(\mathbf k;V,V^+)$ instead of $\mathcal D_{\mathrm{GPT}\diamond}(g;V,V^+)$.
\end{proof}

The following characterization of elements of the dual unit ball  follows easily from Remark \ref{rem:GPT-diamond}:
\begin{cor}
Let $z\in V^g$, then $z\in \mathcal P_{\mathrm{GPT}\diamond}(g;V,V^+)$ if and only if there is some $z_0\in K$ such that
\[
z_0+\sum_i \epsilon_iz_i\in V^+,\quad \forall \epsilon \in \{\pm1\}^g.
\]
\end{cor}

\begin{defi}
The set 
\begin{equation}\label{eq:def-Pi}
\Pi(g; V, V^+) := \left\{(s_1, \ldots, s_g) \in [0,1]^g: \sum_{i = 1}^g s_i \norm{z_i}_V \leq 1 ~\forall (z_1, \ldots, z_g) \in \mathcal P_{\mathrm{GPT} \diamond, g} \right\}
\end{equation}
is called the \emph{$(g;V,V^+)$-blind region}.
\end{defi}

\begin{thm} \label{thm:pi=gamma=delta} It holds that
\begin{equation*}
\Pi(g; V, V^+) = \Gamma(g; V, V^+)
\end{equation*}
\end{thm}

\begin{proof}
The statement follows from Theorem \ref{thm:inclusion_Gamma} and duality of the norms. 
\end{proof}

\begin{ex} In quantum mechanics $\mathrm{QM}_d$, we see by the relation  (described in  Example \ref{ex:jewel}) of the corresponding  GPT diamond $\mathcal D_{\mathrm{GPT} \diamond, g}$ to the
matrix diamond $\mathcal D_{\diamond,g}(d)$ that we have
\[
\mathcal P_{\mathrm{GPT} \diamond, g}=\{(\rho^{1/2}X_1\rho^{1/2},\dots, \rho^{1/2}X_g\rho^{1/2}), \ X\in \mathcal
D_{\diamond,g}(d),\ \rho\in \mathcal{S}_d\}.
\]
It follows that any element $X$ in the matrix diamond defines a family of incompatibility witnesses. Moreover, some of
these witnesses are strict if and only if we have
\[
\|X\|_\sigma:=\sup_{\rho\in \mathcal S_d} \sum_i\|\rho^{1/2}X_i\rho^{1/2}\|_1>1.
\]
It follows that we can write the  Equation~\eqref{eq:def-Pi} as
\[
\Pi(g, \mathcal M_d^{\mathrm{sa}}, \mathrm{PSD}_d)=\{ (s_1,\dots,s_g)\in [0,1]^g,\ \|s.X\|_\sigma \le 1 \ \forall X\in \mathcal D_{\diamond,g}(d)\}.
\]
Note that  the matrix diamond is precisely the unit ball of the injective norm 
 $\|\cdot\|_\epsilon$ in $\ell^g_1\otimes (\mathcal M^{\mathrm{sa}}_d(\mathbb C), \norm{\cdot}_\infty)$, so that  similarly as in Example \ref{ex:QM_rho} it is 
the maximal matrix convex set $\mathcal W_d^{\max}(B_{\ell^g_1})$. From the results of  Example \ref{ex:QM_rho} and 
duality of $\ell^g_\infty$ and $\ell^g_1$, one would expect that the unit ball of the norm $\|\cdot\|_\sigma$ is
precisely the minimal matrix convex set $\mathcal W^{\min}_d(B_{\ell^g_1})$.
One can prove directly that this is indeed the case.

\end{ex}

\section{Compatibility of  effects in centrally symmetric GPTs } \label{sec:centrally-symmetric-tensor}

While Theorem \ref{thm:effect-tensors} nicely characterizes compatible effects in terms of their associated tensors, it is in general hard to compute the crossnorm $\norm{\cdot}_c$. In this section, we will show that in centrally symmetric GPTs, one can substitute, in certain situations, the norm $\norm{\cdot}_c$ by the \emph{projective tensor norm}. This can be done, for example, in the case of the most incompatible effects, see Remark \ref{rem:min-comp-deg}.

In this section we will assume that $(V, V^+,\mathds 1)$ is a \emph{centrally symmetric} GPT.  In this setting, the vector space of
un-normalized states decomposes as $V = \mathbb R v_0 \oplus \bar V$, and we shall write $V \ni z = (z^\circ, \bar z)$
to denote the vector $z = z^\circ v_0 + \bar z$, with $z^\circ \in \mathbb R$ and $\bar z \in \bar V$. The decomposition is such that $\mathds 1(z) = z^\circ$. Similarly, any
element of $A$ has the form $(\alpha,\bar f)$ for $\alpha=\langle
f,v_0\rangle\in \mathbb R$ and $\bar f\in \bar A$. See also Section \ref{sec:centrally-symmetric-GPTs}.

Let $f_1,\dots,f_g\in A$,  $f_i=(\alpha_i,\bar f_i)$. Then each  $f_i$ is
an effect if and only if $\|\bar f_i\|_{\bar A}\le \min\{\alpha_i,(1-\alpha_i)\}$, which implies that we must have 
$\|\bar f_i\|_{\bar A}\le 1/2$. 
Hence if $f_i$ is an effect, then also $(1/2,\bar f_i)$ is an effect. Effects of the form $(1/2,\bar f)$  are called
\emph{unbiased}. See also Remark \ref{rmk:unbiased-symmetric-effects} below.

For $f=(f_1,\dots, f_g)$, $f_i=(a_i,\bar f_i)$,  we have
\[
\bar\phi^{(f)}= y_f\otimes \mathds 1+ \bar\xi^{(f)},\qquad y_f:=\sum_i(2\alpha_i-1)\check c_i,\quad \bar\xi^{(f)}:=2\sum_i
\check c_i\otimes \bar f_i.
\]
Let us denote by $\|\cdot\|_{\bar\epsilon}$ and $\|\cdot\|_{\bar\pi}$ the injective and projective norm on the space $\ell_\infty^g\otimes (\bar A, \norm{\cdot}_{\bar A})$; importantly, note that $\dim \bar A = \dim A - 1$.

\begin{prop}\label{prop:rho_symmetric} Let $\bar \phi=y\otimes \mathds 1+ \bar \xi$, for some $y\in
\ell_\infty^g$ and $\bar \xi\in \ell_\infty^g\otimes \bar A$. Then 
\begin{align*}
\max\{\|y\|_\infty, \|\bar \xi\|_{\bar\epsilon}\}&\le \|\bar\phi\|_\epsilon\le \|y\|_\infty+ \|\bar \xi\|_{\bar\epsilon}\\
\max\{\|y\|_\infty, \|\bar \xi\|_{\bar \pi}\}&\le \|\bar\phi\|_c \le \|y\|_\infty+ \|\bar \xi\|_{\bar\pi}
\end{align*}

\end{prop}

\begin{proof} Let $y=\sum_i y_i\check c_i$ and $\bar\xi=\sum_i \check c_i\otimes \bar g_i$. Then 
\[
\|\bar \phi\|_\epsilon=\max_i \|(y_i, \bar g_i)\|_{\mathds 1}=\max_i(|y_i|+\|\bar
g_i\|_{\bar A})
\]
since $\ell_\infty^g \otimes_\epsilon X \cong \ell_\infty^g(X)$ \cite[Example 3.3]{Ryan2002} and the inequalities for the injective norm are immediate.

Next, assume  that $\|\bar\phi\|_c \le 1$. By Proposition \ref{prop:cross_norms} that there are some $z_j\in \ell_\infty^g$,
$\|z_j\|_\infty=1$
and $h_j=a_j\mathds 1+ \bar h_j$ such that $\|\bar h_j\|_{\bar A}\le a_j$, $\|\sum_j \bar h_j\|_{\bar A}\le 1-\sum_j a_j$ and 
\[
\bar\phi= \sum_j z_j \otimes h_j=\sum_j a_j z_j \otimes \mathds 1+ \sum_j z_j\otimes \bar h_j.
\]
It follows that 
\[
y=\sum_j a_jz_j,\qquad \bar\xi=\sum_j z_j\otimes \bar h_j.
\]
Hence $\|y\|_\infty\le \sum_j a_j\le 1$, moreover, $\sum_j \|\bar h_j\|_{\bar A}\le \sum_j a_j\le 1$, so that  
$\|\bar \xi\|_{\bar\pi}\le 1$. This proves the first inequality for $\|\cdot\|_c$. For the second inequality, note that 
 since $\rho$ is a crossnorm, we have
\[
\|\bar \phi\|_c\le \|y\otimes \mathds 1\|_c+ 
\|\bar \xi\|_c\le \|y\|_\infty+ \|\bar\xi\|_{\bar\pi}.
\]
In the last inequality we also used the fact that
\[
\|\bar\xi\|_{\pi}\le \|\bar\xi\|_{\bar\pi}
\]
 since the inclusion $\bar A \to A$ is an isometry.
 \end{proof}

\begin{cor}\label{cor:unbiased}
We have, for a tuple $f \in A^g$:
\begin{equation}\label{eq:effect-tensor-norm}
	f \text{ effects } \implies \|\bar\xi^{(f)}\|_{\bar\epsilon}  = 
2 \max_{i=1}^g \|\bar f_i\|_{\bar A} \leq 1
\end{equation}
and 
$$f \text{ compatible effects } \implies \|\bar\xi^{(f)}\|_{\bar\pi} \leq 1.$$
If, moreover, the effects $f_i$ are {unbiased}, the reverse implications hold.
\end{cor}
\begin{proof} The implications follows easily by Theorem \ref{thm:effect-tensors} and Proposition \ref{prop:rho_symmetric}. If the $f_i$ are unbiased, $y_f = 0$. We moreover have
for $\bar \xi\in \ell_\infty^g\otimes \bar A$ that
\[
\|\bar\xi\|_\epsilon=\|\bar\xi\|_{\bar\epsilon}\ \text{ and } \ \|\bar\xi\|_c=\|\bar\xi\|_\pi=\|\bar\xi\|_{\bar\pi}.
\]
Thus, Proposition \ref{prop:rho_symmetric} implies the last statement. See also
\cite[Proposition 2.25]{lami2018non} or \cite[Proposition 27]{lami2018ultimate}. 

\end{proof}

\begin{remark} \label{rmk:unbiased-symmetric-effects}
	Note that the reverse implication in Equation~\eqref{eq:effect-tensor-norm} is true if and only if the effect $f$ is unbiased (we are considering here $g=1$). Indeed, an element  $f = f_0 \mathds 1 + \bar f \in A$ is an effect if and only if 
	\begin{align*}
		f \in A^+ &\iff \|\bar f\|_{\bar A} \leq f_0 \qquad \qquad \text{and}\\
		\mathds 1 - f \in A^+ &\iff \|\bar f\|_{\bar A} \leq 1-f_0,
	\end{align*}
	which give together $\|\bar f\|_{\bar A} \leq \min(f_0, 1-f_0)$. However, Equation~\eqref{eq:effect-tensor-norm} gives $\|\bar f\|_{\bar A} \leq 1/2$ which is identical to the previous condition if and only if $f_0 = \langle f, v_0 \rangle = 1/2$, i.e.~if the effect is unbiased. 
\end{remark}

\bigskip

We now turn to incompatibility witnesses. Let us denote
\begin{align*}
\mathcal P'_{\mathrm{GPT}\diamond}(g;V,V^+)&:=\{ (z_1,\dots,z_g)\in \mathcal P_{\mathrm{GPT}\diamond}(g;V,V^+),\
z_i=(0,\bar z_i),\ i\in [g]\}\\
&= (0\oplus \bar V)^g\cap\mathcal P_{\mathrm{GPT}\diamond}(g;V,V^+).
\end{align*}

\begin{lem}\label{lem:witness_central} Let $\bar z=(\bar z_1,\dots,\bar z_g)\in \bar V^g$ and let $z=((0,\bar
z_1),\dots,(0,\bar z_g))$. Then 
\[
\|z\|_{c^*}=\|z\|_\epsilon=\|\bar z\|_{\bar\epsilon},\quad \|z\|_\pi=\|\bar z\|_{\bar\pi}
\]
here $\epsilon$, resp.~$\bar \epsilon$, is the injective crossnorm in $\ell_1^g\otimes (V, \norm{\cdot}_V)$, resp.~$\ell_1^g\otimes(\bar V, \norm{\cdot}_{\bar V})$, similarly for $\pi$ and $\bar\pi$. In particular, 
\[
z\in \mathcal P'_{\mathrm{GPT}\diamond}(g;V,V^+) \iff
(v_0,(0,\bar z_1),\dots,(0,\bar z_g))\in \mathcal D_{\mathrm{GPT}\diamond}(g;V,V^+)\iff \|\bar z\|_{\bar\epsilon}\le 1.
\]

\end{lem}

\begin{proof} We have
\[
\|z\|_\pi=\sum_{i=1}^g\|z_i\|_V=\sum_{i=1}^g\|\bar z_i\|_{\bar V}=\|\bar z\|_{\bar \pi}.
\]
Further, 
\begin{align*}
\|z\|_\epsilon\le 1 &\iff \left\|\sum_{i=1}^g\epsilon_i (0,\bar z_i)\right\|_V=\left\|\sum_{i=1}^g\epsilon_i\bar z_i\right\|_{\bar V}\le 1,\quad \forall
\epsilon\in \{\pm1\}^g\\
&\iff \|\bar z\|_{\bar\epsilon}\le1.
\end{align*}
We have used Equation \eqref{eq:injective-norm}. On the other hand, using Proposition \ref{prop:incomp_witness}, for any $\epsilon\in \{\pm1\}^g$,
\begin{align*}
\left\|\sum_{i=1}^g\epsilon_i\bar z_i\right\|_{\bar V}\le 1,\quad \forall \epsilon\in \{\pm1\}^g
 &\iff v_0+\sum_i \epsilon_i (0,\bar z_i)\in V^+, \quad \forall \epsilon\in \{\pm1\}^g\\
&\iff (v_0,(0,\bar z_1),\dots,(0,\bar z_g))\in \mathcal D_{\mathrm{GPT}\diamond}(g;V,V^+)\\
&\implies z\in \mathcal P'_{\mathrm{GPT}\diamond}(g;V,V^+)\iff \|z\|_{c^*}\le 1.
\end{align*}
Since $\|\cdot\|_{c^*}$ is a crossnorm, we have 
\[
z\in \mathcal P'_{\mathrm{GPT}\diamond}(g;V,V^+)\iff \|z\|_{c^*}\le 1\implies \|z\|_\epsilon \le 1,
\]
this finishes the proof.
\end{proof}

Remarkably, we have the following result where we write $\bar A$ and $\bar V$ for the Banach spaces $(\bar A, \norm{\cdot}_{\bar A})$ and $(\bar V, \norm{\cdot}_{\bar V})$.

\begin{thm} \label{thm:pi=piprime}
For a centrally symmetric GPT $(V, V^+, \mathds 1)$, we have 
\begin{align*}
\Gamma(g;V,V^+)=\Pi'(g; V, V^+) &:= \{s\in [0,1]^g: \|s.\bar z\|_{\ell_1^g\otimes_\pi \bar V} \leq 1, \ \forall 
 \|\bar z\|_{\ell_1^g\otimes_\epsilon\bar V}\le 1 \}\\
&=\{s\in [0,1]^g: \|s.\bar \phi\|_{\ell^g_\infty\otimes_\pi\bar A}\le 1,\ \forall 
\|\bar\phi\|_{\ell_\infty^g\otimes_\epsilon\bar A}\le 1\}.
\end{align*}
In particular, the compatibility degree is
\[
\gamma(g;V,V^+)=1/\rho(\ell_\infty^g, \bar A) = 1/\rho(\ell_1^g, \bar V),
\]
where the quantity $\rho$ is as in Equation \eqref{eq:rho-quotient}:
\[
\rho(\ell_\infty^g,\bar A)= \max_{\bar\phi}\frac{\|\bar\phi\|_{\ell_{\infty}^g\otimes_\pi\bar
A}}{\|\bar\phi\|_{\ell_{\infty}^g\otimes_\epsilon\bar A}}
\]
\end{thm}

\begin{proof} By Theorem \ref{thm:pi=gamma=delta}, to prove the first equality it is enough to show that for a centrally symmetric GPT we have
 \[
\Pi(g;V,V^+)=\Pi'(g;V,V^+).
\]
 The inclusion $\Pi\subseteq \Pi'$ is clear from Lemma \ref{lem:witness_central}. 
To prove the reverse inclusion, let us define, for a tuple $s \in [0,1]^g$, 
\begin{align*}
\beta_s &:=  \max_{z \in \mathcal P_{\mathrm{GPT} \diamond}(g;V,V^+)} \|s.z\|_{\ell_1^g\otimes_\pi V}=
\max_{z \in \mathcal P_{\mathrm{GPT} \diamond}(g;V,V^+)} \sum_{i=1}^g s_i \|z_i\|_V\\
\beta'_s &:= \max_{z \in \mathcal P'_{\mathrm{GPT} \diamond}(g;V,V^+)}\|s.z\|_{\ell_1^g\otimes_\pi V}= \max_{\|\bar
z\|_{\ell_1^g\otimes_\epsilon \bar V} \leq 1}  \sum_{i=1}^g s_i \|\bar z_i\|_{\bar V}.
\end{align*}
Obviously, $s \in \Pi(g; V, V^+)$ if and only if $\beta_s \leq 1$, the same being true for the primed variants. We shall
prove that $\beta_s \leq \beta'_s$, for all $s \in [0,1]^g$. To this end, fix some $s$ and let $z \in \mathcal
P_{\mathrm{GPT} \diamond}(g;V,V^+)$, $z_i=(z_i^\circ,\bar z_i)$,  achieve the maximum in $\beta_s$. We consider the set partition $[g] = I \sqcup J$, where $i \in I \iff z_i^\circ > \|\bar z_i \|_{\bar V}$. We have thus
\begin{equation}\label{eq:beta-s}
\beta_s = \sum_{i \in I} s_i|z_i^\circ| + \sum_{j \in J} s_j\|\bar z_j\|_{\bar V}.
\end{equation}
Let us denote by $s_I$, resp.~$s_J$, the restriction of the $g$-tuple $s$ to the index set $I$, resp.~$J$. Putting 
$$\lambda:= \sum_{i \in I} |z_i^\circ|,$$
we have $\sum_{i \in I} s_i|z_i^\circ| \leq \lambda \|s_I\|_\infty\leq \lambda \beta'_{s_I}$, where the last inequality
follows from the fact that for any $\bar v\in \bar V$ with $\|\bar v\|_{\bar V}\leq1$, we have $\|(0,\dots,0,\bar v,0,\dots,0)\|_{\bar\epsilon}\le 1$ .

 Note that $\lambda \leq 1$, since $z \in \mathcal P_{\mathrm{GPT} \diamond}(g;V,V^+)$ implies using Proposition \ref{prop:incomp_witness} that $z_0 + \sum_{i = 1}^g \epsilon_i z_i \in V^+$ for some $z_0 \in K$ and all $\epsilon \in \{\pm 1\}^g$. Therefore, an application of $\mathds 1$ yields $\sum_{i = 1}^g |z_i^\circ| \leq 1$. 

Let us now focus on the second term in the RHS of Equation \eqref{eq:beta-s}. First, note that if the set $J$ is empty we are
done: $\beta_s \leq \lambda \beta'_s \leq \beta'_s$. Assume now $\lambda \in [0,1)$. Using Lemma \ref{lem:P-subset-0}
below, we have that 
$$\frac{1}{1-\lambda}\left( (0,\bar z_j) \right)_{j \in J} \in \mathcal P'_{\mathrm{GPT} \diamond}(|J|;V,V^+),$$
and thus 
$$\sum_{j \in J} s_j\|\bar z_j\|_{\bar V} \leq (1-\lambda)\beta'_{s_J}$$
by the definition of $\beta'_{s_J}$. We have proven, up to this point, that $\beta_s \leq \lambda \beta'_{s_I} + (1-\lambda)\beta'_{s_J}$. To conclude, we need to show that the function $\beta'_\cdot$ is ``concave'', which we do next. Let $\bar x_0, \bar x_1, \ldots, \bar x_{|I|} \in \bar V$ be such that $\|\bar x_0\|_{\bar V} \leq 1$ and
\begin{align*}
&\left( (1,\bar x_0), (0, \bar x_1), \ldots, (0, \bar x_{|I|}) \right) \in \mathcal D_{\mathrm{GPT}\diamond}(|I|; V, V^+) \qquad \text{and}\\
&\beta'_{s_I} = \sum_{i = 1}^{|I|} s_i \|\bar x_i\|_{\bar V}.
\end{align*}
Consider similar elements $\bar y_0, \bar y_1, \ldots, \bar y_{|J|} \in \bar V$. The claim follows from Proposition \ref{prop:incomp_witness} and the following fact: 
$$\left( (1,\lambda \bar x_0 + (1-\lambda)\bar y_0), \lambda(0, \bar x_1), \ldots, \lambda(0, \bar x_{|I|}), (1-\lambda)(0, \bar y_1), \ldots, (1-\lambda)(0, \bar y_{|J|}) \right) \in \mathcal D_{\mathrm{GPT}\diamond}(g;V, V^+).$$
\end{proof}

\begin{lem}\label{lem:P-subset-0}
Let $z \in \mathcal P_{\mathrm{GPT} \diamond}(g;V,V^+)$, and consider a subset $J \subseteq [g]$ such that $\sum_{i \notin J} |z_i^\circ| < 1$. Then, 
$$\frac{1}{1-\sum_{i \notin J} |z_i^\circ|}\left( (0,\bar z_j) \right)_{j \in J} \in \mathcal P^\prime_{\mathrm{GPT} \diamond}(|J|;V,V^+).$$
\end{lem}
\begin{proof}
Let us put $I := [g] \setminus J$. From the hypothesis $z \in \mathcal P_{\mathrm{GPT} \diamond}(g;V,V^+)$, using Proposition \ref{prop:incomp_witness} we have that, for any signs $\epsilon_i$, $\eta_j$, 
$$1 + \sum_{i \in I} \epsilon_i z_i^\circ + \sum_{j \in J} \eta_j z_j^\circ \geq \left\| \bar z_0 + \sum_{i \in I} \epsilon_i \bar z_i + \sum_{j \in J} \eta_j \bar z_j\right\|_{\bar V},$$
where $\bar z_0 \in \bar V$ is such that 
$(z_0,z_1, \ldots, z_g) \in \mathcal D_{\mathrm{GPT}\diamond}(g;V,V^+)$.
Writing 
$$\sum_{j \in J} \eta_j \bar z_j = \frac 1 2 \left[ \left( \bar z_0 + \sum_{i \in I} \epsilon_i \bar z_i + \sum_{j \in J} \eta_j \bar z_j \right) - \left( \bar z_0 + \sum_{i \in I} \epsilon_i \bar z_i + \sum_{j \in J} (-\eta_j) \bar z_j \right) \right],$$
we get
\begin{align*}
\left\|\sum_{j \in J} \eta_j \bar z_j \right\|_{\bar V} &\leq \frac 1 2 \left[ \left( 1 + \sum_{i \in I} \epsilon_i z_i^\circ + \sum_{j \in J} \eta_j z_j^\circ \right) + \left( 1 + \sum_{i \in I} \epsilon_i z_i^\circ + \sum_{j \in J} (-\eta_j) z_j^\circ \right) \right] \\
&= 1 + \sum_{i \in I} \epsilon_i z_i^\circ.
\end{align*}
The conclusion follows from the expression for the injective norm in Equation \eqref{eq:injective-norm}, Lemma \ref{lem:witness_central} and by setting, for all $i \in I$, $\epsilon_i = - \operatorname{sign}(z_i^\circ)$.
\end{proof}

\begin{remark}\label{rem:min-comp-deg}
Note that Theorem \ref{thm:pi=piprime} and Proposition \ref{prop:rho_symmetric} imply that minimal
compatibility degree is attained on tuples of unbiased effects.
\end{remark}

\section{Inclusion constants} \label{sec:incl-const}
In this section, we use the connection to tensor crossnorms exhibited in Sections \ref{sec:tensor-norms} and \ref{sec:centrally-symmetric-tensor} to give bounds on the compatibility region of various GPTs of interest. We have seen that, in order to find the most incompatible $g$-tuples of effects, one needs to compare the injective with the $c$-tensor norms for an arbitrary GPT, and the injective with the projective tensor norms for centrally symmetric GPTs.
We are thus interested in the ratio
\begin{equation*}
  \frac{1}{\rho(\ell_1^g, X)} =  \min_{z \in \ell_1^g \otimes X}\frac{\|z\|_{\ell_1^g \otimes_\epsilon X}}{\|z\|_{\ell_1^g \otimes_\pi X}} = \min_{z \in \ell_\infty^g \otimes X^\ast}\frac{\|z\|_{\ell_\infty^g \otimes_\epsilon X^\ast}}{\|z\|_{\ell_\infty^g \otimes_\pi X^\ast}}.
\end{equation*}

Proposition 12 of \cite{aubrun2018universal} implies the generic bound
\begin{equation*}
    \rho(\ell_1^g, X) \leq \min\{g, \dim{X}\}.
\end{equation*}
Thus,
\begin{equation} \label{eq:gamma-generic-lower-bound}
    \gamma(g;V,V^+) \geq \frac{1}{\min\{g,\dim{V}\}}
\end{equation}
for a GPT $(V,V^+,\mathds 1)$ by Corollary \ref{cor:norm-ratio}.
If $(V,V^+,\mathds 1)$ is a centrally symmetric GPT of dimension $n+1$, we find by Theorem \ref{thm:pi=piprime} that
\begin{equation}\label{eq:trivial-lower-bounds-csGPTs}
    \gamma(g;V,V^+) \geq \frac{1}{\min\{g,n\}}.
\end{equation}
For concrete GPTs, however, we can often find better bounds using Theorem \ref{thm:pi=piprime}. This is the topic of the remainder of this section. If $(V, V^+, \mathds 1)$ is the centrally symmetric GPT giving rise to the Banach space $X$ as in Section \ref{sec:centrally-symmetric-GPTs}, we will write $\Gamma(g; X)$ and $\gamma(g; X)$ for brevity.

\subsection{Hypercubes}

Let us consider the hypercube GPT $\mathrm{HC}_n$ from Example \ref{ex:hypercube}, where the state space $K$ is the $n$-dimensional hypercube, the unit ball of the Banach space $\ell_\infty^n$. We have $\dim V = n+1$. 

Let us compute the injective norm $\ell_1^g \otimes_\epsilon \ell_\infty^n$ of a vector $z$ using duality:
$$\|z\|_{\ell_1^g \otimes_\epsilon \ell_\infty^n} = \sup_{y \, : \, \|y\|_{\ell_\infty^g \otimes_\pi \ell_1^n} \leq 1} |\langle y, z \rangle|.$$
But the extremal points of the unit ball of the Banach space $\ell_\infty^g \otimes_\pi \ell_1^n$ are the vectors $\epsilon \otimes e_j$ \cite[Proposition 2.2]{Ryan2002}, where $\epsilon \in \{\pm 1\}^g$ is an extremal point of the unit ball of $\ell_\infty^g$ and the $e_j$ are the standard basis vectors, $j \in [n]$. We have thus 
$$\|z\|_{\ell_1^g \otimes_\epsilon \ell_\infty^n} = \sup_{j, \epsilon}  \left |\sum_{i=1}^g \epsilon_i z_{ij}  \right |,$$
where $z_i = \sum_{j=1}^n z_{ij} e_j$ (see also Equation \eqref{eq:injective-norm}). Choosing the appropriate $\epsilon$ solves easily the maximization problem at fixed $j$. We have thus, in the case of the hypercube, 
\begin{align*}
\|z\|_{\ell_1^g \otimes_\epsilon \ell_\infty^n} &= \sup_{j \in [n]} \sum_{i=1}^g |z_{ij}|\\
\|z\|_{\ell_1^g \otimes_\pi \ell_\infty^n} &=  \sum_{i=1}^g \sup_{j \in [n]} |z_{ij}|,
\end{align*}
where the second expression follows from Equation \eqref{eq:projective-norm}.

\begin{prop} \label{prop:comp-region-hypercube}
	For the centrally symmetric GPT with hypercube state space, the compatibility region is the \emph{enlarged} simplex
	$$\Gamma(g; \ell_\infty^n) = \{s \in [0,1]^g \, : \, \forall I \subseteq [g] \text{ s.t. } |I| \leq n, \ \sum_{i \in I} s_i \leq 1\}.$$
\end{prop}
Before proving the result, let us point out that when $g \leq n$, the enlarged simplex above is just the usual probability simplex $\Delta_g$; when $g>n$, the set is strictly larger than $\Delta_g$, see Figure \ref{fig:enlarged-simplex}.
\begin{proof}
	Let us first show the ``$\supseteq$'' inclusion. Take $s \in [0,1]^g$ with the property in the statement. For $i \in [g]$, choose
	$$J(i) \in \operatorname{argmax}_{j\in [n]}|z_{ij}|.$$
	Let us write $J([g]) = \{J(i): i \in [g]\}$ for the range of $J$. We also choose, for every $j \in J([g])$, 
		$$i_*(j) \in \operatorname{argmax}_{i \in J^{-1}(j)} s_i;$$
	note that $i_*$ is an injective function, since $J^{-1}(j) \cap J^{-1}(j^\prime) = \emptyset$ for $j \neq j^\prime$. We have then, for $z$ such that $\|z\|_{\ell_1^g \otimes_\epsilon \ell_\infty^n} \leq 1$, 
	$$\|s.z\|_{\ell_1^g \otimes_\pi \ell_\infty^n} = \sum_{i=1}^g s_i |z_{iJ(i)}| = \sum_{j \in J([g])} \sum_{i \in  J^{-1}(j)} s_i |z_{ij}| \leq \sum_{j \in J([g])} s_{i_*(j)} \sum_{i \in  J^{-1}(j)}  |z_{ij}|  \leq \sum_{j \in J([g])} s_{i_*(j)} \leq 1,$$
	where the last inequality follows from the hypothesis and the fact that $|J([g])| \leq n$.
	
	For the reverse implication, fix a subset $I \in [g]$ of cardinality at most $n$. For $z_{ij} = 1$ if $i=j \in I$ and $0$ elsewhere, we have $\|z\|_{\ell_1^g \otimes_\epsilon \ell_\infty^n} \leq 1$ and 
	$$\|s.z\|_{\ell_1^g \otimes_\pi \ell_\infty^n}  = \sum_{i \in I} s_i.$$
	The assertion follows again from Theorem \ref{thm:pi=piprime}.
\end{proof}

\begin{cor}
	The largest diagonal element in $\Gamma(g; \ell_\infty^n)$ is $s=1/\min(g,n)$. 
\end{cor}

\subsection{Euclidean balls}\label{sec:euclidean}

We consider now the case of a centrally symmetric GPT with the unit ball of $\ell_2^n$ as the state space. This GPT is of particular interest because for $n=3$, one recovers the state space of qubits: the Bloch ball (see Example \ref{ex:qubits}). In this setting, we have, using Equation \eqref{eq:injective-norm},
$$\|z\|_{\ell_1^g \otimes_\epsilon \ell_2^n} = \sup_{y \, : \, \|y\|_{\ell_\infty^g \otimes_\pi \ell_2^n} \leq 1} |\langle y, z \rangle| = \sup_{\|y\|_2 \leq 1}\sum_{i=1}^g |\langle y, z_i \rangle| = \sup_{\epsilon \in \{\pm 1\}^g} \left\|\sum_{i=1}^g \epsilon_iz_i\right\|_2.$$

\begin{prop} \label{prop:comp-eucl-balls}
	For all $g,n$, 
	$$\mathrm{QC}_g := \left\{s \in [0,1]^g: \sum_{i = 1}^g s_i^2 \leq 1\right\} \subseteq \Gamma(g;\ell_2^n).$$
	If $g \leq n$, we have equality above. 
\end{prop}
\begin{proof}
	Consider $s \in [0,1]^g$ with the property that $ \sum_{i=1}^g s_i^2 \leq 1$. Let $z \in \ell_1^g \otimes \ell_2^n$ be such that $\|z\|_{\ell_1^g \otimes_\epsilon \ell_2^n} \leq 1$. We have
	$$\norm{s.z}_{\ell^g_1  \otimes_{\pi} \ell_2^n} = \sum_{i=1}^g s_i \|z_i\|_2 \leq \left(\sum_{i=1}^g s_i^2 \right)^{1/2} \left(\sum_{i=1}^g \|z_i\|_2^2 \right)^{1/2}$$
	by Cauchy-Schwarz and then use the expression for the injective norm to get
	$$\sum_{i=1}^g \|z_i\|_2^2 = \frac{1}{2^g}\sum_{\epsilon \in \{\pm 1\}^g} \left\| \sum_{i=1}^g \epsilon_iz_i\right\|_2^2\leq 1.$$
	
	For the upper bound in the case $g \leq n$, use $z_i = s_i e_i$, where $e_i$ is the standard basis of $\mathbb R^g$. The assertion follows from the previous calculations using Theorem \ref{thm:pi=piprime}.
\end{proof}

Let us consider now pairs of effects ($g = 2$). The projective crossnorm in $\ell^2_\infty\otimes X$ has an explicit expression (see e.g.~the proof of \cite[Proposition 15]{aubrun2018universal}): for $\phi=e_1\otimes x_1+ e_2\otimes x_2$,
we have
\[
\|\phi\|_\pi=\frac12 (\|x_1+x_2\|_X+\|x_1-x_2\|_X). 
\]
Using this expression, we obtain from \eqref{eq:gammaf} and  Corollary \ref{cor:unbiased}:

\begin{enumerate} 
\item For any GPT $(V,V^+,\mathds 1)$ and any pair $f=(f_1,f_2)$ of effects the compatibility degree is
\[
\gamma(f)\le \|f_1-f_2\|_A+ \|f_1-(\mathds 1-f_2)\|_A.
\]
\item For a centrally symmetric GPT and a pair $f=((1/2,\bar f_1),(1/2,\bar f_2))$, we have
\[
\gamma(f)=\|\bar f_1+\bar f_2\|_{\bar A}+\|\bar f_1-\bar f_2\|_{\bar A}.
\]
This is a generalization of a known result for qubits ($n = 3$), see e.g.~\cite[Proposition 4]{busch2008approximate}.
\end{enumerate}

\subsection{Cross polytopes} \label{sec:cross-polytopes}
Let us now consider the GPT with the $n$-dimensional $\ell_1$ unit ball as state space. Again, $\dim V = n+1$. 

Let 
\begin{equation*}
    x = \sum_{i = 1}^g e_i \otimes x^{(i)} \in \ell_1^g \otimes \ell_1^n,
\end{equation*}
where $\{e_i\}_{i \in [g]}$ is the canonical basis of $\mathbb R^g$. In this case, we have by Equation \eqref{eq:projective-norm}
\begin{equation*}
    \norm{x}_{\ell_1^g \otimes_{\pi} \ell_1^n} = \sum_{i = 1}^{g} \norm{x^{(i)}}_1 = \sum_{i = 1}^g \sum_{j = 1}^n |x_j^{(i)}|.
\end{equation*}
Moreover, by Equation \eqref{eq:injective-norm}
\begin{equation*}
      \norm{x}_{\ell_1^g \otimes_{\epsilon} \ell_1^n} = \sup_{\epsilon \in \{\pm 1\}^g} \norm{\sum_{i = 1}^g\epsilon_i x^{(i)}}_1 =\sup_{\epsilon \in \{\pm 1\}^g} \sum_{j = 1}^{n} \left| \sum_{i = 1}^g \epsilon_i x^{(i)}_j\right|.
\end{equation*}

\begin{prop} \label{prop:degree-cross-poly}
 For all $n$, $g \in \mathbb N$, it holds that
 \begin{equation*}
     \gamma(g;\ell_1^n) \geq \frac{1}{g 2^{g-1}} \left(\left\lfloor \frac{g}{2}\right\rfloor + 1\right) \binom{g}{\lfloor \frac{g}{2}\rfloor + 1}.
 \end{equation*}
 This bound is achieved for $n \geq 2^{g-1}$.
\end{prop}
\begin{proof}
We will again use Theorem \ref{thm:pi=piprime} to relate the ratio of norms to the compatibility degree. Let us now consider $y^{(i)} \in \mathbb R^{g}$, $i \in [n]$, such that $y_j^{(i)} = x_i^{(j)}$. Then, 
\begin{equation*}
     \norm{x}_{\ell_1^g \otimes_{\pi} \ell_1^n} = \sum_{i = 1}^g \sum_{j = 1}^n |y_i^{(j)}| \qquad \norm{x}_{\ell_1^g \otimes_{\epsilon} \ell_1^n} = \sup_{\epsilon \in \{\pm 1\}^g} \sum_{j =1}^{n} |\langle \epsilon, y^{(j)} \rangle|.
\end{equation*}
Our aim is to estimate $\inf\{\norm{x}_{\ell_1^g \otimes_{\epsilon} \ell_1^n}:  \norm{x}_{\ell_1^g \otimes_{\pi} \ell_1^n} = 1\}$. We find
\begin{equation*}
    \norm{x}_{\ell_1^g \otimes_{\epsilon} \ell_1^n} \geq \frac{1}{2^g} \sum_{\epsilon \in \{\pm 1\}^g} \sum_{j = 1}^n |\langle \epsilon, y^{(j)} \rangle| =: f(y^{(1)}, \ldots, y^{(n)}).
\end{equation*}
Let us therefore minimize $f(y^{(1)}, \ldots, y^{(n)})$ under the constraint imposed by the projective norm. Without loss of generality, we can assume $y^{(j)} \geq 0$ for all $j \in [n]$. Since $f$ is convex and invariant under permutation of the $y^{(j)}$ we have
\begin{align*}
    f(y^{(1)}, \ldots, y^{(n)}) &= \frac{1}{n!} \sum_{\sigma \in \mathfrak S_n} f(y^{(\sigma(1))}, \ldots, y^{(\sigma(n))}) \\
    &\geq f\left(\frac{1}{n!} \sum_{\sigma \in \mathfrak S_n} (y^{(\sigma(1))}, \ldots, y^{(\sigma(n))})\right) = f(\bar y, \ldots, \bar y).
\end{align*}
where $\mathfrak S_n$ is the symmetric group on $n$ symbols and $\bar y = 1/n \sum_{i = 1}^n y^{(i)}$. We thus need only minimize over $\bar y \in \mathbb R_+^{g}$ such that $\norm{\bar y}_1 = 1/n$. Since
\begin{equation*}
    \bar y \mapsto \sum_{\epsilon \in \{\pm 1\}^g} |\langle \epsilon, \bar y\rangle|
\end{equation*}
is Schur convex, the minimum is achieved at $\bar y = 1/(gn)(1, \ldots, 1)$. Thus,
\begin{equation*}
    \norm{x}_{\ell_1^g \otimes_{\pi} \ell_1^n} \geq \frac{\sum_{\epsilon \in \{\pm 1\}^g}\left| \sum_{i = 1}^g \epsilon_i\right|}{2^g g}
\end{equation*}
For $n \geq 2^{g-1}$, it can be checked that $\{y^{(j)}\}_{j \in [n]} = \{\{y^{(\epsilon)}\}_{\epsilon \in \{\pm 1\}^{g-1}}, 0, \ldots, 0\}$ achieves this bound, where $y^{(\epsilon)} = 1/(g2^{g-1})(1, \epsilon_1, \ldots \epsilon_{g-1})$. It is finally easy to compute that
\begin{align*}
    \sum_{\epsilon \in \{\pm 1\}^g}\left| \sum_{i = 1}^g \epsilon_i\right| &= \sum_{k = 0}^g |g - 2k| \binom{g}{k}\\
    & = \left(\left\lfloor \frac{g}{2}\right\rfloor + 1\right) \binom{g}{\lfloor \frac{g}{2}\rfloor + 1}.
\end{align*}
\end{proof}

\subsection{Relation to 1-summing norms} \label{sec:1-summing}

A different way to obtain bounds on the $\rho(\ell^g_1,X)$ quantity, which were mentioned at the beginning of this section, is to use results on  \emph{1-summing norms} of Banach spaces; this method was already used in \cite{aubrun2018universal}.  The 1-summing norm of (the identity operator of) a Banach space $X$ is the smallest constant $c$ with the property that for all $g \geq 1$	and $z_1, \ldots, z_g \in X$ we have
$$\sum_{i=1}^g \|z_i\|_X \leq c \sup_{\|y\|_{X^\ast} \leq 1} \sum_{i=1}^g |\langle y, z_i\rangle|.$$
We denote it by $\pi_1(X)$. 

\begin{prop}\label{prop:1-summing}
	The largest $s$ such that
	$$s(1, \ldots, 1) \in \bigcap_{g \geq 1}\Gamma(g; V, V^+)$$
	is $s = 1/\pi_1(\bar V)$, in the case where $(V, V^+, \mathds 1)$ is a centrally symmetric GPT. For a general GPT $(V, V^+, \mathds 1)$,
	\begin{equation*}
	    \gamma(g; V, V^+) \geq \frac{1}{\pi_1(V)},
	\end{equation*}
	where $V$ is the corresponding base norm space.
\end{prop}
\begin{proof}
It is clear from the expressions for the injective (Equation \eqref{eq:injective-norm}) and projective (Equation \eqref{eq:projective-norm}) norm that $\sup_{g \in \mathbb N} \rho(\ell_1^g, X) = \pi_1(X)$. The assertion follows from Theorem \ref{thm:pi=piprime} and Corollary \ref{cor:norm-ratio}.
\end{proof}

The following constants were computed in \cite[Theorem 2]{gordon1969onp}:
\begin{itemize}
	\item $\pi_1(\ell_1^n) = \frac{2^{n-1}n}{(1+\lfloor n/2 \rfloor) \binom{n}{1+\lfloor n/2 \rfloor}} \sim \sqrt{\frac{\pi n}{2}}$
	\item $\pi_1(\ell_2^n)= 
	\frac{\sqrt{\pi}\Gamma(\frac{n+1}{2})}{\Gamma(\frac{n}{2})}
	\sim \sqrt{\frac{\pi n}{2}}$
	\item $\pi_1(\ell_\infty^n) = n$,
\end{itemize}
where $\sim$ expresses the asymptotic behavior for large $n \in \mathbb N$. 

\begin{remark}
As seen above, the 1-summing norms give bounds on the compatibility degree of arbitrarily many dichotomic measurements. The asymptotic growth in the $\ell^n_1$ and the $\ell_2^n$ case show that the bounds improve over the naive bound $1/n$.
\end{remark}

\begin{prop} \label{prop:1-summing-QM}
Let $d \in \mathbb N$ and let $S^d_1$ be the Banach space of self-adjoint matrices equipped with the Schatten 1 norm. Then,
\begin{equation*}
    \pi_1(S^d_1) \leq c d,
\end{equation*}
where $c$ is a constant. In particular, we can take $c = 7.79$.
\end{prop}
\begin{proof}
By \cite{Garling1971}, it follows that $\pi_1(X) \lambda(X) = n$ for any $n$-dimensional Banach space $X$ with enough symmetries in the sense of \cite[§16]{Tomczak-Jaegermann1989}. Here, $\lambda(X)$ is the projection constant of $X$ \cite[§32]{Tomczak-Jaegermann1989}. From \cite[Proposition 32.7]{Tomczak-Jaegermann1989}, it follows that 
\begin{equation*}
    \lambda(X) \geq c^\prime [C_2(X)(1 + \log C_2(X))^{1/2}]^{-1} \sqrt{n},
\end{equation*}
where $c^\prime$ is a universal constant. Here, $C_2(X)$ is the Rademacher cotype 2 constant \cite[§4]{Tomczak-Jaegermann1989}. The space $X = S^d_1$ has enough symmetries, has dimension $n = d^2$ and $C_2(S^d_1) \leq \sqrt{e}$ \cite{Tomczak-Jaegermann1974}. This proves the assertion. A more careful analysis of the proof of Theorem 10.14 and Proposition 32.7 of \cite{Tomczak-Jaegermann1989} shows that we can take \cite[Equation (10.17)]{Tomczak-Jaegermann1989}
\begin{equation} \label{eq:1-summing-QM}
    c = \inf_{q \in (2,\infty)} (\sqrt{q} C_2(S^d_1))^{\frac{q}{q-2}}.
\end{equation}
With $C_2(S^d_1) \leq \sqrt{e}$, we obtain the desired bound on $c$ through numerical optimization.
\end{proof}

\subsection{Upper bounds for centrally symmetric GPTs}
For centrally symmetric GPTs, the results of \cite{aubrun2018universal} yield also upper bounds on the compatibility degree. In particular, we can use \cite[Prop. 15]{aubrun2018universal}, \cite[Lemma 19]{aubrun2018universal} and \cite[Theorem 6]{aubrun2018universal}:
\begin{enumerate}
\item For any centrally symmetric GPT with $\dim(\bar V)\ge 2$, we have 
\begin{equation}\label{eq:gamma-upper-2}
\gamma(2;V,V^+)\le 1/\sqrt{2}.
\end{equation}
\item For any centrally symmetric GPT with $\dim(\bar V)\ge g$, we have
\begin{equation}\label{eq:gamma-upper-g}
\gamma(g;V,V^+)\le \sqrt{2/g}
\end{equation}
\item For a centrally symmetric GPT related to $\ell_2^n$, we have for $g\ge n$
\begin{equation} \label{eq:gamma-upper-l2}
   \gamma(g;V,V^+)\le 1/\sqrt{n}. 
\end{equation}
\item There exists a constant $c > 0$ such that for any centrally symmetric GPT with $\dim(\bar V) = n$
\begin{equation}\label{eq:gamma-upper-complicated}
\gamma(g;V,V^+) \leq c \frac{\log(\min\{g,n\})}{(\min\{g,n\})^{1/8}}.
\end{equation}
\end{enumerate}

\begin{remark}
If one drops the condition that the GPT $(V,V^+, \mathds 1)$ be centrally symmetric, no general non-trivial upper bounds for the compatibility degree can be derived. Indeed, for the classical mechanics GPT $\mathrm{CM}_d$, since all crossnorms agree \cite{Aubrun2019,Aubrun2019a}, we have, for all $g$,$d \in \mathbb N$,
$$\Gamma(g; \mathrm{CM}_d) = [0,1]^g \quad \text{ and } \quad \gamma(g; \mathrm{CM}_d)=1.$$
The central symmetry condition rules out simplicial cones, forcing all incompatibility regions for $g \geq 2$ to be non-trivial.
\end{remark}

\section{Results} \label{sec:results}
In this section, we collect the results on the different compatibility regions for GPTs of interest obtained in this work and compare them to previously known bounds. We focus on results for dichotomic measurements. Note that one can always lift the results for dichotomic measurements to bounds on measurements with more outcomes using the symmetrization techniques from Theorem \ref{thm:symmetrization}.

\subsection{Centrally symmetric GPTs}

\subsubsection{Hypercubes}
If the state space is the unit ball of the Banach space $\ell_\infty^n$, Proposition \ref{prop:comp-region-hypercube} characterizes the compatibility region completely. We find that the compatibility region is
\begin{equation*}
    	\Gamma(g; \ell_\infty^n) = \{s \in [0,1]^g \, : \, \forall I \subseteq [g] \text{ s.t. } |I| \leq n, \ \sum_{i \in I} s_i \leq 1\}.
\end{equation*}
We have plotted this enlarged simplex in Figure \ref{fig:enlarged-simplex}. In particular, this implies that 
\begin{equation}
    \gamma(g; \ell_\infty^n)= \frac{1}{\min\{g,n\}}.
\end{equation}
Before this work, this result was shown for $g = 2$ and $n = 2$ in \cite[Proposition 2]{busch2013comparing} (see also \cite{Jencova2017}) and generalized to $g\leq n$, $n \in \mathbb N$ in \cite[Corollary 5]{jencova2018incompatible}. 

\begin{figure} 
    \centering
    \includegraphics[scale=.5]{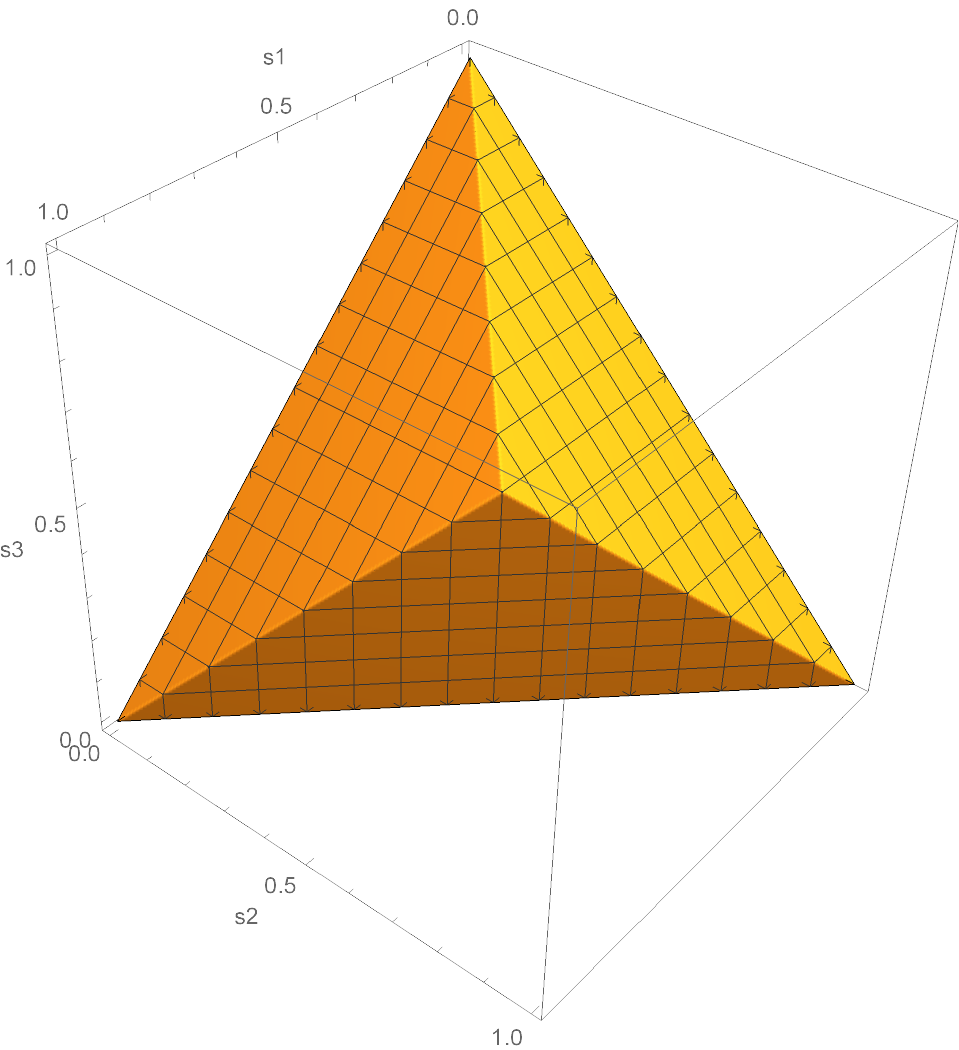}
    \caption{The compatibility region $\Gamma(g;\ell_\infty^n)$ for $n = 2$, $g = 3$. The polyhedron is the convex hull of the five points $(0,0,0)$, $(1,0,0)$, $(0,1,0)$, $(0,0,1)$, and $(1/2,1/2,1/2)$; it strictly contains the simplex $\Delta_3$.}
    \label{fig:enlarged-simplex}
\end{figure}

\subsubsection{Euclidean balls} \label{sec:results-eucl-balls}
For Euclidean balls, i.e.\ if the state space is the unit ball of $\ell_2^n$, Proposition \ref{prop:comp-eucl-balls} shows that
\begin{equation*}
    \Gamma(g;\ell_2^n) = \left\{s \in [0,1]^g: \sum_{i = 1}^g s_i^2 \leq 1\right\} =: \mathrm{QC}_g \qquad \forall g, n \in \mathbb N \mathrm{~s.t.~}g \leq n.
\end{equation*}
In particular,  $\gamma(g;\ell_2^n) = 1/\sqrt{g}$ if $g \leq n$. Thus, we have a complete characterization for few measurements.

For $g > n$,  Proposition \ref{prop:comp-eucl-balls} shows that 
\begin{equation}\label{eq:lower-l2}
      \Gamma(g;\ell_2^n) \supseteq \mathrm{QC}_g.
\end{equation}
Moreover, combining this with Equation \eqref{eq:gamma-upper-l2} and the $1$-summing constants in Section \ref{sec:1-summing}, we obtain
\begin{equation*}
    \frac{1}{\sqrt{n}} \geq \gamma(g;\ell_2^n) \geq \max \left\{1/\sqrt{g}, \frac{\Gamma(\frac{n}{2})}{\sqrt{\pi}\Gamma(\frac{n+1}{2})} \right\} \qquad \forall g,n \in \mathbb N \mathrm{~s.t.~} g \geq n.
\end{equation*}
Hence, for $g > (\sqrt{\pi} \Gamma(\frac{n+1}{2})/\Gamma(\frac{n}{2}))^2$ the inclusion in Equation \eqref{eq:lower-l2} is strict. See Figure \ref{fig:l2_phase_diagram} for a phase diagram of our findings. 

\begin{figure}
    \centering
    \includegraphics[scale=0.8]{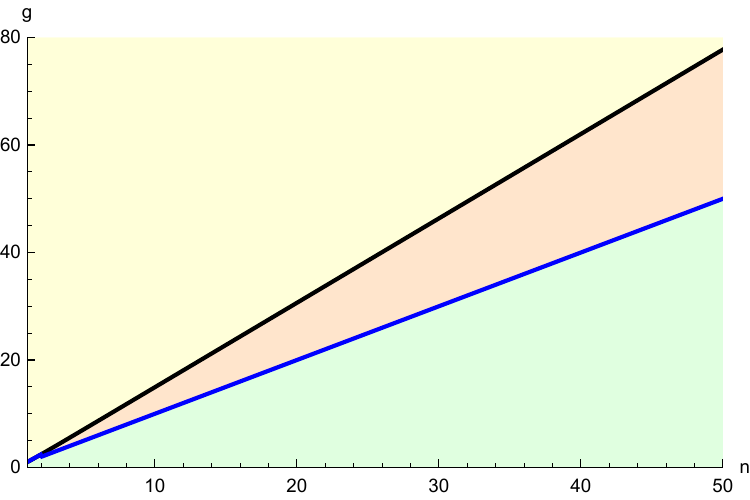}
    \caption{The blue function is $n \mapsto n$. The black function is $n \mapsto h(n)^{-2}$. In the green region, we know $\Gamma(g;\ell_2^n) = \mathrm{QC_g}$. In the yellow region, we know $\Gamma(g;\ell_2^n) \supsetneq \mathrm{QC_g}$. In the orange region, we only know $\Gamma(g;\ell_2^n) \supseteq \mathrm{QC_g}$, but not whether the inclusion is strict.}
    \label{fig:l2_phase_diagram}
\end{figure}

It is easy to compute that
\begin{equation*}
    h(n):=\frac{\Gamma(\frac{n}{2})}{\sqrt{\pi}\Gamma(\frac{n+1}{2})} = \begin{cases}\frac{4^m}{\pi m} \binom{2m}{m}^{-1} & n = 2m \\ 4^{-m} \binom{2m}{m} & n = 2m+1 \end{cases}
\end{equation*}
and asymptotically, the bound behaves like $\sqrt{2/(\pi n)}$. See Figure \ref{fig:h-function} for a plot of the function. Prior to this work, only results for $n = 3$ were known, see Section \ref{sec:results-QM} below.

\begin{figure}
    \centering
    \includegraphics[scale=0.8]{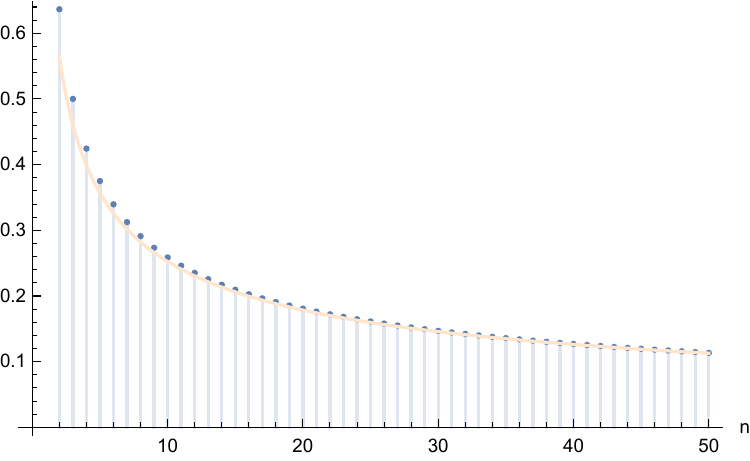}
    \caption{The functions $h(n)$ (in blue) and $\sqrt{2/(\pi n)}$ (in orange), for $n \in \{2, \ldots, 50\}$.}
    \label{fig:h-function}
\end{figure}

\subsubsection{Cross polytopes}
If the state space is the unit ball of $\ell_1^n$, we have only bounds on the compatibility degree. For $n \geq 2^{g-1}$, Proposition \ref{prop:degree-cross-poly} proves that the compatibility degree is given as
\begin{equation*}
    \gamma(g;\ell_1^n) = \frac{1}{g 2^{g-1}} \left(\left\lfloor \frac{g}{2}\right\rfloor + 1\right) \binom{g}{\lfloor \frac{g}{2}\rfloor + 1}=:f(g).
\end{equation*}
Asymptotically, $f(g)$ behaves like $\sqrt{2/(\pi g)}$. See Figure \ref{fig:f-function} for a plot of this function. For $n<2^{g-1}$, Proposition \ref{prop:degree-cross-poly} and the 1-summing constant in Section \ref{sec:1-summing} merely provide lower bounds on the compatibility degree. Combined with Equation \eqref{eq:gamma-upper-2}, we obtain for $n \geq 2$:
\begin{equation*}
    \frac{1}{\sqrt{2}} \geq \gamma(g;\ell_1^n) \geq \max\{f(n), f(g)\}. 
\end{equation*}
For $n \geq g$, the upper bound can be improved to $\sqrt{2/g} \geq \gamma(g;\ell_1^n)$ using Equation \eqref{eq:gamma-upper-g}. Using \cite[Equation (61)]{aubrun2018universal}, Proposition 30 of \cite{lami2018ultimate} shows that in general $\sqrt{2/\min\{n,g\}} \geq \gamma(g;\ell_1^n)$.
It remains an open question to find any bounds on $\Gamma(g; \ell_1^n)$ for non-diagonal elements. Our lower bounds improve over the lower bounds derived in \cite{lami2018ultimate}.
\begin{figure}
    \centering
    \includegraphics[scale=0.8]{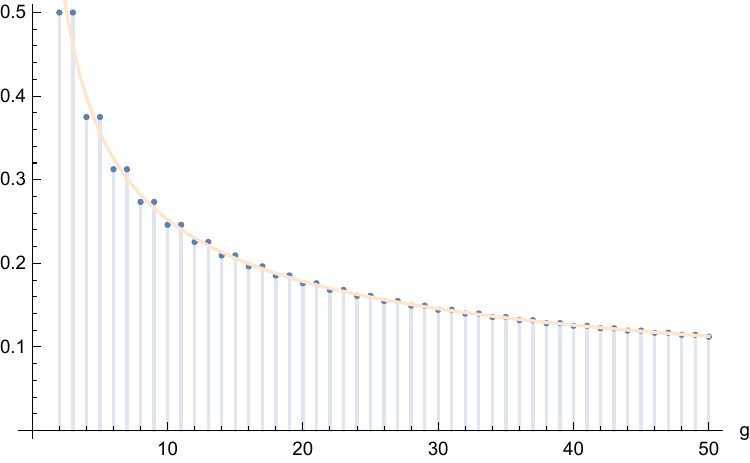}
    \caption{The functions $f(g)$ (in blue) and $\sqrt{2/(\pi g)}$ (in orange), for $g \in \{2, \ldots, 50\}$.}
    \label{fig:f-function}
\end{figure}

\subsection{General centrally symmetric GPTs}
For general centrally symmetric GPTs, we have only bounds on the compatibility degree. Let $(V,V^+, \mathds 1)$ be a centrally symmetric GPT where $V$ has dimension $n+1$. Equations \eqref{eq:trivial-lower-bounds-csGPTs} and \eqref{eq:gamma-upper-2} imply for $n \geq 2$:
\begin{equation*}
    \frac{1}{\sqrt{2}} \geq \gamma(g;V,V^+) \geq \frac{1}{\min\{n,g\}}.
\end{equation*}
The upper bound shows that all centrally symmetric GPTs with $n \geq 2$ have incompatible dichotomic measurements (otherwise we would have $\gamma(g;V,V^+) = 1$), which is consistent with the findings of \cite{Plavala2016}.
For $n \geq g$, the upper bound can be improved to $ \sqrt{2/g} \geq \gamma(g;V,V^+)$ using Equation \eqref{eq:gamma-upper-g}. Moreover, Equation \eqref{eq:gamma-upper-complicated} implies that there exists a constant $c > 0$ such that
\begin{equation*}
\gamma(g;V,V^+) \leq c \frac{\log(\min\{g,n\})}{(\min\{g,n\})^{1/8}}.
\end{equation*}
This bound is useful for us in the regime $n < g$ and large $n$. 

\subsection{Other GPTs}
\subsubsection{Quantum mechanics}\label{sec:results-QM}
For quantum systems of dimension $2$, the state space corresponds to the unit ball of $\ell_2^3$ via the Bloch ball and is therefore centrally symmetric. We can hence use the results outlined in Section \ref{sec:results-eucl-balls} for $n=3$. For $g=2$ and $g=3$, we have that
\begin{equation*}
    \Gamma(g;\mathrm{QM}_2) = \mathrm{QC}_g.
\end{equation*}
This recovers the results of \cite{busch1986unsharp, Brougham2007, busch2008approximate, Pal2011, busch2013comparing, busch2016quantum, bluhm2018joint} (see Section 9.1 of \cite{bluhm2018joint} for the contributions of the respective works). In particular, the compatibility degree is $\gamma(g;\mathrm{QM}_2) = 1/\sqrt{g}$ in this case.

For $g \geq 4$, we still have 
\begin{equation*}
    \Gamma(g;\mathrm{QM}_2) \supseteq \mathrm{QC}_g,
\end{equation*}
which recovers the result of \cite{bluhm2018joint}. However, we find that the compatibility degree is bounded by
\begin{equation}\label{eq:gamma-QM2}
   \forall g \geq 4, \qquad 0.58 \approx \frac{1}{\sqrt{3}} \geq \gamma(g;\mathrm{QM}_2) \geq \frac{1}{2}.
\end{equation}
The upper bound follows from Equation \eqref{eq:gamma-upper-l2} and the lower bound from the 1-summing norm of $\ell_2^3$. This improves over the lower bound $\max\{1/\sqrt{g}, 1/4\}$ obtained in \cite{bluhm2018joint} and recovers the best known upper bound imposed by the compatibility degree for $g=3$ (see Proposition \ref{prop:gamma-decreases}).
Numerics using the SDP in \cite{wolf2009measurements} indicate that the SIC-POVM constructed by choosing vertices of a tetrahedron on the Bloch sphere has $\gamma(f) = 1/\sqrt{3}$, whereas we can find random examples with $\gamma(f) \leq 0.56<1/\sqrt3$ for $g=4$. Thus, $\gamma(4;\mathrm{QM}_2)<1/\sqrt3$. However, the exact value of $\gamma(4;\mathrm{QM}_2)$ remains presently open.

\begin{question}
We know, from Proposition \ref{prop:1-summing}, that
$$\lim_{g \to \infty} \gamma(g; \mathrm{QM}_2) = \inf_{g \geq 1} \gamma(g; \mathrm{QM}_2) = \frac{1}{\pi_1(\ell_2^3)} = \frac 1 2.$$
Is the lower bound in Equation \eqref{eq:gamma-QM2} optimal, that is, 
\begin{equation*}
    \forall g \geq 4 \quad  \gamma(g; \mathrm{QM}_2) = \frac{1}{2} \iff \gamma(4; \mathrm{QM}_2) = \frac{1}{2} \quad ?
\end{equation*}
\end{question}

Beyond qubits, Proposition \ref{prop:1-summing-QM} gives
\begin{equation*}
    \gamma(g; \mathrm{QM}_d) \geq \frac{1}{7.79 d} \qquad \forall g \in \mathbb N.
\end{equation*}
This bound has the same dimension dependence as the bound $ \gamma(g; \mathrm{QM}_d) \geq \frac{1}{2 d}$ derived in \cite{bluhm2018joint}, but the constant is worse. One could therefore hope to improve the lower bound on the compatibility degree using better bounds on $C_2(S^d_1)$. However, even putting the cotype 2 constant to $1$ in Equation \eqref{eq:1-summing-QM}, which is the minimal value possible, gives a constant of approximately $3.81$, which is still larger than $2$. Thus, we need a better way to compute $\pi_1(S^d_1)$ in order to improve over the bound on $\gamma(g;\mathrm{QM}_d)$ in \cite{bluhm2018joint}.

\subsubsection{General bounds}
For general GPTs, we have the generic lower bound in Equation \eqref{eq:gamma-generic-lower-bound}:
\begin{equation*}
    \gamma(g;V,V^+) \geq \frac{1}{\min\{g,\dim{V}\}}
\end{equation*}
While the bound in terms of $g$ was known (see e.g.\ \cite{Heinosaari2016}), the bound in terms of the dimension seems to be new.

\bigskip

\emph{Acknowledgements.} A.B.\ and I.N.\ would like to thank Guillaume Aubrun for helpful discussions and for explaining to us how to bound the $1$-summing norm of $S_1$. Moreover, the authors would like to thank Guillaume Aubrun also for organizing the workshop GPT \& QIT in Lyon. Furthermore, A.B.\ and A.J.\ would like to thank Mil\'an Mosonyi for organizing the QIMP 2018 conference, during which the idea for this project was born. A.B.\ acknowledges support from the VILLUM FONDEN via the QMATH Centre of Excellence (Grant no.  10059) and from the QuantERA ERA-NET Cofund in Quantum Technologies implemented within the European Union’s Horizon 2020 Programme (QuantAlgo project) via the Innovation Fund Denmark. A.J.\ was supported by the grants APVV-16-0073 and VEGA 2/0142/20. I.N.~ was supported by the ANR project \href{https://esquisses.math.cnrs.fr/}{ESQuisses}, grant number ANR-20-CE47-0014-01.

\appendix

\section
{Conic programming}
\label{sec:conic-programming}

\subsection{Background}
To begin, let us briefly recapitulate the theory of conic programming. We follow \cite[Section 4]{Gaertner2012}. All vector spaces we consider will be finite dimensional. We will equip them with an inner product choosing a pair of dual bases for the vector space and its dual.

\begin{defi}[Conic program {\cite[Definition 4.6.1]{Gaertner2012}}] \label{defi:conic}
Let $L^+ \subseteq L$, $M^+ \subseteq M$ be closed convex cones, let $b \in M$, $c \in L^\ast$ and let $A: L \to M$ be a linear operator. A \emph{conic program} is an optimization problem of the form
\begin{align*}
\mathrm{maximize} \qquad& \langle c, x \rangle \\
\mathrm{subject~to} \qquad& b -A(x) \in M^+ \\
& x \in L^+
\end{align*}
\end{defi}
The dual problem is then given by \cite[Section 4.7]{Gaertner2012}
\begin{align*}
\mathrm{minimize} \qquad& \langle b, y \rangle \\
\mathrm{subject~to} \qquad& A^\ast(y) - c \in (L^+)^\ast \\
& y \in (M^+)^\ast
\end{align*}
Weak duality always hold, i.e.\ the value of the primal problem is upper bounded by the value of the dual program if the dual program is feasible. A sufficient condition for strong duality to hold is the following version of Slater's condition:
\begin{thm}[{\cite[Theorem 4.7.1]{Gaertner2012}}] \label{thm:slater} If the conic program in Definition \ref{defi:conic} is feasible, has finite value $\gamma$ and has an interior point $\tilde x$, then the dual program is also feasible and has the same value $\gamma$.
\end{thm}
If $M^+ \neq \{0\}$, $\tilde x$ is an interior point if $\tilde x \in \mathrm{int}(L^+)$ and $b - A(\tilde x) \in \mathrm{int}(M^+)$ \cite[Definition 4.6.4]{Gaertner2012}. 

\subsection{Map extension} \label{sec:conic-programming-map}

In this section, we will show that the existence of a positive extension of a positive map can be checked using a conic program. We can give the following generalization of the results in \cite{Heinosaari2012}.

\begin{thm} \label{thm:conic_extension}
Let $E \subseteq \mathbb R^d$ be a $g$-dimensional subspace such that $E \cap \operatorname{ri}(\mathbb R_+^d) \neq \emptyset$ and $E^+ = E \cap \mathbb R^d_+$. Moreover, let $(V,V^+, \mathds 1)$ be a GPT. Finally, let $\Phi: E \to A$ be given by a basis $\{e_i\}_{i \in [g]} \subset E$ of $E$, $g \in \mathbb N$, and $\{f_i\}_{i \in [g]} \subset A$ such that $\Phi(e_i) = f_i$ for all $i \in [g]$. Then, there exists a positive extension $\tilde \Phi: (\mathbb R^{d}, \mathbb R_+^{d}) \to (A, A^+)$ of $\Phi$ if and only if the conic program
\begin{align*}
\mathrm{maximize} \qquad& -\langle s, (h_1^+-h_1^-, \ldots, h_g^+-h_g^-) \rangle \\
\mathrm{subject~to} \qquad& \sum_{i \in [g]} e_i \otimes (h_i^+-h_i^-) \in E^+ \otimes_{\mathrm{max}} V^+ \\
& 1 - \langle \mathds 1, h_i^\pm \rangle \geq 0 \\
& h^{\pm}_i \in V^+ \qquad \forall i \in [g]
\end{align*}
has value $0$. Here, $s: V^g \to \mathbb R$ is given as
\begin{equation*}
\langle s, h_1, \ldots, h_g \rangle = \langle \chi_V, \sum_{i \in [g]} f_i \otimes (h^+_i-h_i^-) \rangle  = \sum_{i \in [g]} f_i(h_i^+ - h_i^-).
\end{equation*}
\end{thm}
\begin{proof}
Let $z \in E \otimes V$. Using that the $e_i$ form a basis, we can write
\begin{equation*}
z = \sum_{i \in [g]} e_i \otimes z_i,
\end{equation*}
where $z_i \in V$ for all $i \in [g]$. Since $V^+$ is proper, we can decompose each $z_i = z_i^+ - z_i^-$, where $z_i^\pm \in V^+$. Then, $\langle s, (z_1, \ldots, z_g)\rangle = s_\Phi(z)$, where $s_\Phi$ is defined as in Proposition \ref{prop:positivity-extension}.  By linearity, it suffices to restrict to $z$ such that $\langle \mathds{1}, z_i^\pm \rangle \leq 1$ for all $i \in [g]$ in order to check positivity of $s_\Phi$.  Thus, the conic program has value $0$ if and only if $s_\Phi$ is positive. The assertion then follows from Proposition \ref{prop:positivity-extension}.
\end{proof}

\begin{remark}
Of course, we could write down a trivial conic program for map extension just checking whether the corresponding tensor $\phi^\Phi$ is in $(E^+)^\ast \otimes_{\min} A^+$ using Proposition \ref{prop:positivity-extension}. However, checking membership in this cone might be hard in practice. If $E^+$ is a polyhedral cone, 
\begin{equation*}
    \sum_{i \in [g]} e_i \otimes (h_i^+-h_i^-) \in E^+ \otimes_{\mathrm{max}} V^+ 
\end{equation*}
can be checked as
    \begin{equation*}
    \sum_{i \in [g]} \langle \alpha_j, e_i \rangle (h_i^+-h_i^-) \in V^+ \qquad \forall j,
\end{equation*}
where the $\alpha_j$ are the extremal rays of $(E^+)^\ast$ (compare to Lemma \ref{lem:elements-in-the-max}). This is arguably easier since it does not involve checking the membership in a tensor cone directly. Moreover, it recovers the result that in quantum mechanics, where $V^+ = \mathrm{PSD}_d$: compatibility via map extension can be checked using a semidefinite program \cite{Heinosaari2012}. 
\end{remark}

\begin{thm} The conic program in Theorem \ref{thm:conic_extension} is feasible and satisfies strong duality.
\end{thm}
\begin{proof}
In the following, we identify $z = \sum_{i \in [2g]}z_i \otimes \mu_i \in  L \otimes \mathbb R^{2g}$ with the vector $(z_1, \ldots z_{2g})$, where $\{\mu_i\}_{i \in [2g]}$ is an orthonormal basis of $\mathbb R^{2g}$. Comparing the conic program to Definition \ref{defi:conic}, we identify
\begin{align*}
M & = (E \otimes V) \times \mathbb R^{2g} \\
M^+ & = (E^+ \otimes_{\mathrm{max}} V^+) \times \mathbb R_+^{2g}   \\
L &= V \otimes \mathbb R^{2g} \\
L^+ &= (V^+)\otimes \mathbb R_+^{2g} \\
c & = -s \\
b &= (0, \underbrace{1, \ldots, 1}_{2g}) \in (E \otimes V) \times \mathbb R^{2g} \\
A(h_1^\pm, \ldots, h_g^\pm) & = \left(- \sum_{i \in [g]} e_i \otimes (h^+_i-h^-_i), \langle \mathds{1}, h^\pm_1 \rangle, \ldots, \langle \mathds{1}, h^\pm_g \rangle\right).
\end{align*}
It can be verified that dual conic program is thus given by
\begin{align*}
\mathrm{minimize} \qquad& \sum_{i \in [2g]} y_i \\
\mathrm{subject~to} \qquad& s + \sum_{i \in [2g]} y_i \mathds{1} \otimes \mu_i^\ast + B(z) \in A^+ \otimes \mathbb R_+^{2g} \\
& z \in (E^+)^\ast \otimes_{\mathrm{min}} A^+ \\
& y_i \in \mathbb R_+ \qquad \forall i \in [2g].
\end{align*}
Here, $B(z) \in A \otimes \mathbb R^{2g}$ is given as $B(z)(h_1, \ldots, h_{g}) = \langle z, \sum_{i \in [g]} e_i \otimes h_i \rangle$ and $\mu^\ast_i$ is the dual basis of $\epsilon_i$. Letting $y_1 = \ldots = y_{2g}$ and realizing that $\mathds{1} \otimes (1, \ldots 1)$ is an order unit in $A^+ \otimes \mathbb R_+^{2g}$, for any $z \in (E^+)^\ast \otimes_{\mathrm{min}} A^+$ we can find a $y_1 > 0$ such that 
\begin{equation*}
s +  y_1 \mathds{1} \otimes (1, \ldots, 1) + B(z) \in \mathrm{int}\left(A^+ \otimes \mathbb R_+^{2g}\right).
\end{equation*}
This is true, since the order unit is an interior point of $A^+ \otimes \mathbb R_+^{2g}$, hence there is a $y_1$ such that 
\begin{equation*}
\frac{1}{y_1} (s + B(z)) +  \mathds{1} \otimes (1, \ldots, 1) \in \mathrm{int}\left(A^+ \otimes \mathbb R_+^g\right).
\end{equation*}
Since the interior points of $A^+ \otimes \mathbb R_+^g$ are those points $w$ such that $\langle w, x \rangle > 0$ for all $x \in V^+ \otimes \mathbb R_+^{2g} \setminus \{0\}$, multiplication by $y_1$ preserves the fact that the point is in $\mathrm{int}\left(A^+ \otimes \mathbb R_+^{2g}\right)$. Therefore, the dual problem has an interior point. This also implies that the value of the dual program is finite, since it is lower bounded by $0$. The remarks at the beginning of \cite[Section 4.7]{Gaertner2012} imply that Theorem \ref{thm:slater} still applies if we interchange the primal and the dual problem. Thus, the assertion follows.
\end{proof}

\subsection{Computing \texorpdfstring{$\norm{\cdot}_c$}{norm-rho}} \label{sec:computing-rho}

Finally, we note that the norm $\norm{\cdot}_c$, introduced in Theorem \ref{thm:effect-tensors} and Proposition \ref{prop:cross_norms}, can also be computed by a conic program, namely: for any $\bar \phi \in E_g^\ast \otimes A$, $-\norm{\bar \phi}_c$ is the value of the following conic program
\begin{align*}
    \mathrm{maximize}\qquad& - \lambda\\
    \mathrm{subject~to} \qquad& \lambda \check 1_g \otimes \mathds 1 - \bar \phi \in (E^+_g)^\ast \otimes_{\min} A^+ \\
    & \lambda \in \mathbb R_+
\end{align*}
Since $\norm{\bar \phi}_c$ is finite and since $\lambda \check 1_g \otimes \mathds 1 - \bar \phi \in \operatorname{int}{(E^+_g)^\ast \otimes_{\min} A^+}$ for $\lambda$ large enough since $\check 1_g \otimes \mathds 1 \in  (E^+_g)^\ast \otimes_{\min} A^+$ is an order unit, strong duality holds by Theorem \ref{thm:slater}. The dual conic program is
\begin{align*}
    \mathrm{minimize}\qquad& - \langle \bar \phi, y \rangle\\
    \mathrm{subject~to} \qquad& 1- \langle \check 1_g \otimes \mathds 1, y \rangle \geq 0 \\
    & y \in E_g^+ \otimes_{\max} V^+
\end{align*}
The membership in $E_g^+ \otimes_{\max} V^+$ can be decided only evaluating the equations in Equation \eqref{eq:entrywise-condition} by Lemma \ref{lem:elements-in-the-max}. So we can give the alternative formulation of the dual conic program which only checks membership in $V^+$:
\begin{align*}
    \mathrm{minimize}\qquad& - \langle \bar \phi,  1_g \otimes y_0 + \sum_{i = 1}^g c_i \otimes (y_i^+ - y_i^-) \rangle\\
    \mathrm{subject~to} \qquad& 1- \langle \mathds 1, y_0 \rangle \geq 0 \\
    & y_0 + \sum_{i = 1}^g \epsilon_i(y_i^+ - y_i^-) \in V^+ \qquad \forall \epsilon \in \{\pm 1\}^g\\
    & y_0, y_i^\pm  \in  V^+ \qquad \forall i \in [g].
\end{align*}
Due to Theorem \ref{thm:effect-tensors}, this recovers the result that in quantum mechanics, where $V^+ = \mathrm{PSD}_d$, compatibility of effects can be checked using a semidefinite program (see e.g.~\cite{wolf2009measurements}). It also recovers the result that compatibility in GPTs can be checked with a conic program \cite{Plavala2016}, although the one we give here is different.

\bibliographystyle{alpha}
\bibliography{spectralit}
\end{document}